\documentclass[
    10pt,
    english,
    final,
]{article}
\usepackage{amsmath}
\usepackage{amssymb}
\usepackage{amsthm}
\usepackage{accents}
\usepackage[english]{babel}
\usepackage{bbm}
\usepackage{chngcntr}
\usepackage{enumitem}
\usepackage{listings}
\usepackage{lmodern} 
\usepackage{mathtools}
\usepackage[unicode=true]{hyperref}
\usepackage{setspace,graphicx,tikz,tabularx} 
\usepackage{parskip} 
\usepackage{theoremref}
\usepackage[backend=biber,style=numeric,url=false,sortcites=true,isbn=false,doi=true,eprint=true,date=year,maxbibnames=9]{biblatex}
\usepackage{csquotes}

\DeclareNameAlias{default}{family-given}

 \usepackage{biblatex}

\counterwithin{equation}{section}
\mathtoolsset{showonlyrefs}

\DeclarePairedDelimiter\abs{\lvert}{\rvert}%
\DeclarePairedDelimiter\norm{\lVert}{\rVert}%

\DeclareMathOperator*{\esssup}{ess\,sup}

\newcommand{\<}{\langle}
\renewcommand{\>}{\rangle}

\let\predefinedC\c
\renewcommand{\c}[1]{{\mathcal{#1}}}

\renewcommand{\b}[1]{{\mathbb{#1}}}

\newcommand{\N}{{\mathbb{N}}}
\newcommand{\R}{{\mathbb{R}}}
\renewcommand{\P}{{\mathbb{P}}}
\DeclareMathOperator{\E}{{\mathbb{E}}}
\DeclareMathOperator{\Var}{Var}
\newcommand{\Q}{{\mathbb{Q}}}
\newcommand{\indep}{\mathrel{\perp\!\!\!\perp}}

\newcommand{\1}{\mathbbm{1}}

\newtheorem{theorem}{Theorem}[section]
\newtheorem{lemma}[theorem]{Lemma}
\newtheorem{remark}[theorem]{Remark}
\newtheorem{corollary}[theorem]{Corollary}

\newtheorem{definition}[theorem]{Definition}
\newtheorem{standingAssumption}{Standing assumption}
\newtheorem{example}[theorem]{Example}

\begin{document}

\title{Mean-field control problems with multi-dimensional singular controls\thanks{We thank Tiziano De Angelis, Giorgio Ferrari, and seminar participants at various conferences and seminars for valuable comments and suggestions. Financial support from the International Research Training Group (IRTG) 2544 {\sl Stochastic Analysis in Interaction} and the TRR 388 {\sl Rough Analysis, Stochastic Dynamics and Related Fields}, Project B05, is gratefully acknowledged.}}

    \author{Robert Denkert\thanks{Humboldt University Berlin, Department of Mathematics, Unter den Linden 6, 10099 Berlin} \quad \quad Ulrich Horst\thanks{Humboldt University Berlin, Department of Mathematics and School of Business and Economics, Unter den Linden 6, 10099 Berlin}}
    
    \maketitle
    
\vspace{-2.6mm}
\begin{abstract}
We consider mean-field control problems with mul\-ti-di\-men\-sion\-al singular controls. A key challenge when analysing singular controls are jump costs. When controls are one-dimensional, jump costs are most naturally computed by linear interpolation. When the controls are multi-dimensional the situation is more complex, especially when the model parameters depend on an additional mean-field interaction term, in which case one needs to ``jointly'' and ``consistently'' interpolate jumps both on a distributional and a pathwise level. This is achieved by introducing the novel concept of two-layer parametrisations of stochastic processes. Two-layer parametrisations allow us to equivalently rewrite rewards in terms of continuous functions of parametrisations of the control process and to derive an explicit representation of rewards in terms of minimal jump costs. From this we derive a dynamic programming principle for {mean-field control problems} with multi-dimensional singular controls. Under the additional assumption that the value function is continuous we characterise the value function as the minimal super-solution to a certain quasi-variational inequality in the Wasserstein space.
\end{abstract}
    
    

\maketitle


\section{Introduction and problem formulation}


We consider the {mean-field singular control problem} of maximizing a reward function of the form
\begin{equation}\label{reward}
    J(\xi) = \E\bigg[\int_0^T f(t,m_t,X_t,\xi_t) dt + g(m_T,X_T,\xi_T) - \int_{[0,T]} c(t,m_t,X_t,\xi_t) d\xi_t \bigg]
\end{equation}
over the set of adapted, non-decreasing, $l$-dimensional $(l \in \N)$ càdlàg processes $\xi$ subject to the $d$-dimensional $(d \in \N)$ state dynamics 
\begin{equation} \label{eq diffusion x}
    dX_t = b(t,m_t,X_t,\xi_t) dt + \sigma(t,m_t,X_t,\xi_t) dW_t + \gamma(t) d\xi_t.
\end{equation}
Here, $T \in (0,\infty)$ is the terminal time, $W$ is an $m$-dimensional Brownian motion on some filtered probability space $(\Omega,\c F,\b F,\P)$, and
\begin{equation} \label{eq mt}
    m_t\coloneqq \P_{(X_t,\xi_t)} 
\end{equation}
denotes the joint law of the state-control process $(X_t,\xi_t)$ at time $t \in [0,T]$. {The $d\xi_t$-integral in \eqref{reward} is to be understood as the sum of component-wise Riemann-Stieltjes integrals.}

We allow for a dependence of the running, terminal and the jump cost as well as the drift and diffusion coefficient of the state process on both the current state-control pair and their joint distribution. So called extended Mean-field (or McKean-Vlasov) control problems with only regular controls have been studied by many authors including \cite{acciaio_extended_2019,pham_bellman_2018,djete_extended_2022,buckdahn_partial_2021}. 

{Singular control problems involving mean-field terms} where the state dynamics and the running and terminal payoffs depend only on the law of the state process and the jump cost depends only on time have been considered in the context of mean-field games (MFGs) by many authors. \cite{fu_mean_2017} studied MFGs with singular and regular controls; ergodic and supermodular MFGs where the players only interact through the reward functional but not the state dynamics have recently been studied in \cite{cao_stationary_2022,dianetti_unifying_2022};  selected MFGs with singular controls and their respective $N$-player approximations were studied in, e.g.\@ \cite{campi_mean-field_2022,cao_approximation_2023,cao_mfgs_2022,guo_stochastic_2019}. 

{
The dynamic programming principle (DPP) and the characterisation of the value function via quasi-variational inequalities (QVIs) for classical singular control problems were first established using probabilistic techniques by Haussmann and Suo \cite{haussmann_singular_1995,haussmann_singular_1995-1}. Subsequent research as in \cite{federico_characterization_2014} focused on characterising the resulting optimal free boundaries arising from such QVIs. De Angelis and Milazzo \cite{de_angelis_dynamic_2023} provided a rigorous DPP derivation for a large class of problems combining singular control with classical control and discretionary stopping. Žitković \cite{zitkovic_dynamic_2014} developed an abstract dynamic programming framework using controlled Markov families and martingale measures. Talbi, Touzi and Zhang \cite{talbi2021dynamic,talbi2022viscosity} considered the closely related optimal stopping problems in a mean-field setting, establishing a DPP and QVI characterisation of the value function.}

Introducing an additional dependence of the cost functions on the current state, the control and the law of the state-control process is straightforward for continuous controls. The situation is very different when singular controls are allowed. Even in the absence of any mean-interaction the reward function \eqref{reward} may no longer be upper semi-continuous when the controls are allowed to be discontinuous; approximating a discontinuity in the control may lead to a lower cost than using the control itself.
\cite{alvarez2000singular} gives an explicit example where approximating jumps in a singular control by a number of consecutive smaller jumps generates a higher payoff than the limiting singular control; the author calls such approximating strategies \emph{chattering strategies}.\footnote{Different classes of mean-field control problems with this kind of cost structure have been investigated in \cite{hu2017singular,agram_singular_2019} using the stochastic maximum principle approach.}

To avoid chattering-type strategies it seems natural to modify the reward functional in such a way that singular controls generate the same reward as approximating strategies. In \emph{one-dimensional} singular control problems ($\xi \in \R$) without mean-field interaction, jumps are most naturally approximated by linear interpolation. This leads to reward functionals for the form 
\begin{equation}\label{eq reward functional one-dimensional non-mean-field case}
\begin{split}
    J(\xi) & ~ = \E\bigg[\int_0^T f(t,X_t) dt + g(X_T) - \int_0^T c(t,X_t) d\xi_t^c  \\
    & \qquad \qquad - \sum_{t\in [0,T]} \int_0^{\xi_t - \xi_{t-}} c(t,X_{t-} + \gamma(t) \zeta) d\zeta \bigg],
\end{split}
\end{equation}
that has been considered by many authors including \cite{min_singular_1987,zhu_generalized_1992,taksar_infinite-dimensional_1997,dufour_singular_2004,de2018stochastic}. Here, 
$\xi^c$ denotes the continuous part of $\xi$ and the third term on the right-hand side corresponds to an interpolation of the jump in space at the jump time. 

{Whereas the cost structure \eqref{eq reward functional one-dimensional non-mean-field case} is well-understood in one dimension, many practical problems are most naturally analyzed in a multi-dimensional setting. A typical example is optimal liquidation of multi-asset portfolios with singular controls (as in e.g.~\cite{HN}) and cross market impact (as in e.g.~\cite{HX}) where different assets can be traded simultaneously and where buying/selling one asset, say Apple, increases/decreases prices and/or inverse market depths\footnote{Market depth is a standard liquidity measure. It essentially measures the liquidity (number if limit orders) standing at different price levels of a limit order book. A high market depth means high liquidity and hence small impact of market orders on asset prices.} not only of the traded asset but also of correlated assets from the same sector, say Microsoft.\footnote{Such cross-dependencies are well documented in the financial economics literature.} In such models, $X$ usually denotes a vector of best bid and ask prices and portfolio holdings (hence a $3n$-dimensional vector if $n$ assets can be traded), $\xi$ denotes a vector of aggregate buying and selling activities (a $2n$-dimensional vector), $c$ represents  buy- and sell-side inverse market depths of the different assets, the $d\xi$-integral in \eqref{eq diffusion x} describes the impact of own trading on best bid/ask price dynamics and portfolio holdings, and the singular cost term $\int c \cdot d\xi_t = \int \sum_{i=1}^l c^{(i)} d\xi^{(i)}_t$ in \eqref{reward} represents market impact costs. Fu \cite{fu_extended_2022} considers a liquidation models that allows for simultaneous trading of multiple assets but, unlike our framework, does not allow for dependencies in impact costs. 
Our mean-field term $m_t = \P_{(X_t, \xi_t)}$ allows for an additional dependence of asset prices and markets depths aggregate market behaviour as in \cite{FHX}.}  

\subsection{Jump interpolation in higher dimensions}

{Dependencies in the cost functions substantially complicate the definition of the reward functional compared to the one-dimensional case \eqref{eq reward functional one-dimensional non-mean-field case}.} In the one-dimensional case, the interpolation of jumps is unique up to a reparametrisation of time. In the multi-dimensional case considered in this paper the situation is very different, as there are many ways in which jumps can be interpolated. {When jumps can occur simultaneously across different dimensions (e.g. if different assets can be traded simultaneously) and if the cost coefficient $c$ depends on the control (e.g. if market depths depend on trading strategies), then the way jumps are interpolated matters and different interpolations may result in different costs. In particular, extending the one-dimensional formula \eqref{eq reward functional one-dimensional non-mean-field case} using simple linear interpolation (as in \cite{zhu_generalized_1992}) proves insufficient in multi-dimensional settings, as illustrated by the following Example \ref{example path dependent jump costs in 2 d}.}

\begin{example}\label{example path dependent jump costs in 2 d}
{Consider a simple deterministic case where the control process $\xi \in D([0,2]; \R^2)$ is given by $\xi_t = (0,0)$ for $t<1$ and $\xi_t = (1,1)$ for $t\geq 1$. Let the singular cost cost function be $c(\xi) = \xi^{(2)}$ for $\xi = (\xi^{(1)},\xi^{(2)})$. 
\begin{itemize}
    \item \textbf{Linear Interpolation:} Consider the sequence $(\xi^n)_n$ defined by convolution with a uniform kernel:
    \begin{equation}\label{eq:approx1}
        \xi^n_s := n \int_{s-\frac 1 n}^s \xi_r dr.
    \end{equation}
    For $s \in [1, 1+\frac 1 n]$, this yields the path $\xi^n_s = (n(s-1), n(s-1))$, corresponding to a linear traversal of the segment from $(0,0)$ to $(1,1)$ over the time interval $[1, 1+\frac 1 n]$. The associated costs are then given by the integral over this interval,
    \begin{equation}
    \begin{split}
        &\int_1^{1+\frac 1 n} \Big( \xi^{n,(2)}_s d\xi^{n,(1)}_s + \xi^{n,(2)}_s d\xi^{n,(2)}_s \Big) \\
        &= \int_1^{1+\frac 1 n} \Big( n(s-1) (n ds) + n(s-1) (n ds) \Big) 
        = 2n^2 \int_1^{1+\frac 1 n} (s-1) ds 
        = 1.
    \end{split}
    \end{equation}
    Thus the limiting cost for this sequence is 1.    
    This value corresponds to the cost calculated along the straight line segment representing the jump, consistent with the approach in \cite{zhu_generalized_1992}.
    \item \textbf{Alternative Interpolation:} Consider instead the sequence $(\tilde{\xi}^n)_n$ defined component-wise by
    \begin{equation}\label{eq:approx2}
        \tilde{\xi}^n_s \coloneqq \bigg( n \int_{s-\frac 1 n}^s \xi^{(1)}_r dr \, , \, n \int_{s-\frac 2 n}^{s-\frac 1 n} \xi^{(2)}_r dr \bigg).
    \end{equation}
    The path $\tilde{\xi}^n_s$ first moves from $(0,0)$ to $(1,0)$ over the time interval $[1, 1+\frac 1 n]$, and then moves from $(1,0)$ to $(1,1)$ over the interval $[1+\frac 1 n, 1+\frac 2 n]$. The associated costs are 
    \begin{equation}
    \begin{split}
        & \int_1^{1+\frac 1 n} \tilde{\xi}^{n,(2)}_s d\tilde{\xi}^{n,(1)}_s + \int_{1+\frac 1 n}^{1+\frac 2 n} \tilde{\xi}^{n,(2)}_s d\tilde{\xi}^{n,(2)}_s \\
        &= \int_1^{1+\frac 1 n} 0 \cdot (n ds) + \int_{1+\frac 1 n}^{1+\frac 2 n} n(s-(1+\tfrac 1 n)) (nds)  = \frac 1 2.\\
    \end{split}
    \end{equation}
    Thus the limiting cost for this sequence is $\frac 1 2$. This value corresponds to the cost accrued along the path interpolating the jump from (0,0) to (1,1) via the point (1,0).
\end{itemize}}
\end{example}

The situation is even more complex when the state dynamics and/or the reward function depend on a mean-field term, {as illustrated by Example \ref{example mean-field dependent jump costs} below}. In {general, we will need to} -- in a well-defined sense -- ``jointly'' and ``consistently'' interpolate jumps both on a distributional level and on a pathwise level for each trajectory of the control process separately. 

\begin{example}\label{example mean-field dependent jump costs}
{Consider a one-dimensional mean-field scenario where all particles jump simultaneously, $\xi_t = 0$ for $t<1$ and $\xi_t = 1$ for $t \ge 1$, so the measure flow jumps from $m_{1-} = \delta_0$ to $m_1 = \delta_1$. Let the singular control cost function depend only on the measure, specifically let $c(m) = \int x^2 \,m(dx)$. Then the costs accumulated during the jump at time $t=1$ can be represented as $\int_0^1 c(m_{\text{interp}}(p)) dp$, where $p \in [0,1]$ parametrises the progress of the jump and $m_{\text{interp}}(p)$ is an interpolation path from $\delta_0$ to $\delta_1$. We compare two canonical interpolation paths:
\begin{itemize}    \item \textbf{Displacement Interpolation (Geodesic):} Consider the path $[0,1]\ni p \mapsto \delta_p$. Then $c(\delta_p) = p^2$, and the associated cost is
    \[
        \int_0^1 c(\delta_p) dp = \int_0^1 p^2 \, dp = \frac{1}{3}.
    \]
    \item \textbf{Mixture Interpolation:} Consider the path $[0,1]\ni p \mapsto (1-p)\delta_0 + p\delta_1$. It is not difficult to show that $c((1-p)\delta_0 + p\delta_1) = p$. The associated cost is
    \[
        \int_0^1 c((1-p)\delta_0 + p\delta_1) dp = \int_0^1 p \, dp = \frac{1}{2}.
    \]
\end{itemize}
In particular, the costs associated with a jump in the measure flow depend on the chosen interpolation path on the distributional level when the control costs $c$ depends on $m$.}
\end{example}

\subsection{Weak solution}
In what follows we describe our control problem in greater detail. We adopt the weak formulation as our solution concept and choose probability measures on the canonical path space instead of adapted stochastic processes as our controls. For a given initial time $t\in [0,T]$ -- and to capture the possibility of jumps at the initial time -- we fix some $\varepsilon > 0$ and choose the set
\[
    D^0 \coloneqq  D^0([t,T];\R^d\times\R^l) \subseteq D([t-\varepsilon,T];\R^d\times \R^l)
\]
of all càdlàg $\R^d \times \R^l$-valued functions on $[t-\varepsilon,T]$ that are constant on $[t-\varepsilon,{t})$ as the canonical path space for our state-control process. 

The path space will be equipped with the weak $M_1$ ($WM_1$) topology; this topology has been successfully utilised in the context of singular control by many authors, including \cite{cohen2021singular,fu_mean_2017,fu_extended_2022}. 
We denote by $\mathbb F^{X,\xi}$ the filtration generated by the canonical process $(X,\xi)$, and for any measurable space $(E,\mathcal E)$ we denote by $\c P_2(E)$ the set of all probability measures on $E$ with finite second moment, equipped with the 2-Wasserstein topology. We use the shorthand notation $\c P_2 = \c P_2(\R^d\times\R^l)$. {Finally, given random variables $(Y_1,\dotso,Y_k)$, we write $\P_{(Y_1,\dotso,Y_k)} \coloneqq \P\circ (Y_1,\dotso,Y_k)^{-1}$ for the joint law of $(Y_1,\dotso,Y_k)$ under $\P$.}

The weak solution concept adopted in this paper requires the following notion of a weak solution to the SDE \eqref{eq diffusion x}.

\begin{definition}\label{definition weak solution sde}
    We consider the canonical space
    \[
    \tilde\Omega\coloneqq D^0 {([t,T];\R^d\times\R^l)}\times C([t,T];\R^m)
    \]
    equipped with the Borel $\sigma$-algebra, the canonical process $(\tilde X,\tilde\xi,\tilde W)$ and the natural filtration $\b F^{\tilde X,\tilde\xi,\tilde W}$. A probability measure $\tilde\P \in \c P_2(\tilde\Omega)$ is a called weak solution to our SDE \eqref{eq diffusion x} if under $\tilde\P$ the following holds:
    \begin{enumerate}[label=(\roman*)]
        \item $\tilde W$ is an $\b F^{\tilde X,\tilde\xi,\tilde W}$-Brownian motion,
        \item $\tilde X$ follows the dynamics \eqref{eq diffusion x} driven by the control $\tilde\xi$ and the Brownian motion $\tilde W$ on $[t-,T]$
        \footnote{This should be understood as $\tilde X_t = \tilde X_{t-} + \gamma(t) (\tilde\xi_t-\tilde\xi_{t-})$ and $(\tilde X,\tilde\xi,\tilde W)$ satisfying the SDE \eqref{eq diffusion x} on $[t,T]$.}
        with
        \[
        m_u \coloneqq \tilde\P_{(\tilde X_u,\tilde\xi_u)} \quad \mbox{for all}  \quad u\in [t-,T].
        \]
\end{enumerate}
\end{definition}

For a given initial distribution we define the set of admissible controls as the set of all probability measures on the path space $D^0$ such that (i) the law of the canonical process $(X,\xi)$ on $D^0$ under this measure equals the marginal distribution of a weak solution to our SDE and (ii) the control process is almost surely non-decreasing. More precisely, the set of admissible controls is defined as follows. 

\begin{definition}\label{definition admissible controls} 
For any pair $(t,m) \in [0,T] \times \c P_2$ the set $\c P(t,m)$  of admissible controls is given by the set of all probability measures $\P\in \c P_2(D^0)$ on the path space $D^0$ that satisfy the following properties:
    \begin{enumerate}[label=(\roman*)]
        \item $\P_{(X_{t-},\xi_{t-})} \eqqcolon m_{t-} = m$,
        \item $\xi$ is non-decreasing on $[t-,T]$ $\P$-a.s.,
        \item $\P = \P_{(X,\xi)} = \tilde\P_{(\tilde X,\tilde\xi)}$ is the marginal law of a weak solution $\tilde\P$ to the SDE \eqref{eq diffusion x} on $[t-,T]$.
    \end{enumerate}
\end{definition}

\begin{remark}
{Under natural assumptions on the data, the set $\c P(t,m)$ is non-empty for any pair $(t,m)\in [0,T]\times\c P_2$. Indeed, choosing the constant control $\xi \equiv \xi_{t-}$ on $[t-,T]$ reduces the SDE \eqref{eq diffusion x} to a standard McKean-Vlasov SDE for $X$, for which the existence of a solution is well-known, see e.g.\@ \cite[Theorem 4.21]{carmona2018probabilistic}.}
\end{remark}

To define the reward function we denote by $\c C(t,m)\subseteq \c P(t,m)$ the subset of continuous controls, that is the subset of probability measures $\P \in \c P_2(D^0)$ such that the function $u\mapsto \xi_u(\omega)$ -- and thus also the function $u\mapsto X_u(\omega)$ -- is continuous $\P$-almost surely. For continuous controls $\P \in \c C(t,m)$ the reward function 
\begin{equation}
    J_{\c C}(t,m,\P) \coloneqq \E^\P\bigg[\int_t^T f(u,m_u,X_u,\xi_u) du + g(m_T,X_T,\xi_T) - \int_t^T c(u,m_u,X_u,\xi_u) d\xi_u \bigg]
\end{equation}
is well defined. {Under mild conditions on the model parameters this function turns out to be continuous with respect to the Wasserstein distance}. 

We prove that the set $\c L(t,m) \subseteq \c C(t,m)$ of continuous controls of bounded velocity is dense in the set of admissible controls $\c P(t,m)$ viewed as a subset of $\c P_2(D^0)$. A bounded velocity control is a probability measure $\P \in \c P(t,m)$ such that the the process $\xi$ is almost surely of the form
\begin{equation}
\xi_s = \xi_{t-} + \int_t^s u_r dr,\qquad s\in [t,T],
\end{equation}
for some bounded, $\b F^{X,\xi}$-adapted,  non-negative process $(u_s)_{s\in [t,T]}$. To avoid chattering strategies we extend the reward function from continuous to general singular controls as the maximal reward that can be obtained by continuous approximation of admissible controls. Specifically, for $\P \in P(t,m)$ we set  
\begin{equation}\label{eq reward function general definition}
    J(t,m,\P) \coloneqq {\sup_{\substack{\P^n\to \P\text{ in }\c P_2(D^0)\\(\P^n)_n\subseteq\c C(t,m)}} \limsup_{n\to\infty}} \  J_{\c C}(t,m,\P^n),
\end{equation}
where the limit superior is taken over all approximating sequences of continuous (or, equivalently, bounded velocity) controls. 

Our goal is to derive a dynamic programming principle (DPP) and to characterise the value function
\[
    V(t,m) \coloneqq \sup_{\P\in \c P(t,m)} J(t,m,\P)
\]
in terms of a quasi-variational inequality (QVI) on the Wasserstein space. For this, the representation \eqref{eq reward function general definition} of the reward function -- albeit natural -- turns out to be very inconvenient.

\subsection{Main results}

To establish our DPP and the QVI characterisation of the value function we establish an alternative and more transparent representation of the reward function in terns of minimal jump costs akin to the representation \eqref{eq reward functional one-dimensional non-mean-field case} in the one-dimensional case. This requires a ``joint'' and ``consistent'' interpolation of jumps both on a distributional and a pathwise level. The interpolation can potentially be achieved in many ways; the challenge is to find an interpolation that results in a ``convenient'' and ``transparent'' cost structure. A minimal ``transparency'' requirement  is that the interpolation reduces to \eqref{eq reward functional one-dimensional non-mean-field case} in the one-dimensional case.

Our approach is based on the novel concept of two-layer parametrisations. The {key} idea is {to ensure continuity in the state-control process by introducing two successive reparametrisations}. 
{First we} interpolate jumps at the distributional level by deterministically rescaling the state-control process in time in such a way that the measure flow associated with the rescaled process {becomes} continuous. {Importantly,} the law of the rescaled process {must still} correspond in a well-defined manner to a weak solution (on a different time-scale) of the SDE \eqref{eq diffusion x}. 
This \emph{first-layer parametrisation} results in a state-control process that is continuous on a distributional but not necessarily on a pathwise level.

{Second, to ensure} continuity on a pathwise level, we {introduce a second} -- this time random -- reparametrisation in time that rescales any realisation of the first-layer process separately. The rescaling again needs to be done in a way that is consistent with the idea of the resulting \emph{second-layer parametrisation} being a weak solution (again on a different time-scale) of the SDE \eqref{eq diffusion x}.  A subtle, yet important difference to the first layer parametrisation is the fact that the new time scale is different for different realisations of the state-control process so that we can no longer use the distribution of the reparametrised process as the mean-field term in our dynamics. Instead, we will use the reparametrised measure flow from the previous layer. Undoing the said transformations on both layers we recover the original singular control.
      
Having introduced the abstract notion of a two-layer parametrisation we prove that approximating a singular control by continuous ones is equivalent to the existence of a two-layer parametrisation of the singular control that can be approximated by a sequence of continuous two-layer parametrisations of the approximating continuous controls. In other words, working with (weak) controls and working with parametrisations is equivalent. This allows us to represent our reward function in terms of  two-layer parametrisations. 

The new representation has several major advantages that justify our choice of interpolation. First, the new reward functional is continuous, which the original one (defined on the level of controls) was not. Second, the new representation of the reward function allows us to compute minimal jump costs. The proof uses the fact that discontinuities can be classified as discontinuities of the first kind that materialise only on a pathwise level where only a negligible set of ``particles'' jump at a given time, and discontinuities of the second kind where a non-negligible number of particles jumps simultaneously at the same time and which hence materialise both on a pathwise and a distributional level. Specifically, we establish a representation of the reward function of the form
\begin{equation}
\begin{split}
    &J(t,m,\P)\\
    &= \E^\P\bigg[\int_t^T f(u,m_u,X_u,\xi_u) du + g(m_T,X_T,\xi_T) - \sum_{J^d_{[t,T]}(m)} C_{\c P_2}(u,m_{u-},m_u)\\
    &\ \ -\sum_{J^c_{[t,T]}(m)\cap J^d_{[t,T]}(\xi)} C_{D^0}(u,m_u,X_{u-},\xi_{u-},\xi_u) - \int_{J^c_{[t,T]}(m)\cap J^c_{[t,T]}(\xi)} c(u,m_u,X_u,\xi_u) d\xi_u\bigg],
\end{split}
\end{equation}
where the functions $C_{\c P_2}$ and $C_{D^0}$ describe the minimal jump cost on the distributional and pathwise level, respectively, and for any function $h$ the set $J^{d/c}_{[t,T]}(h)$ denotes the points of (\mbox{dis-})\allowbreak continuity of $h$ on the time interval $[t,T]$. If the cost terms are independent of the mean-field term and the controls are one dimensional, then the above functional reduces to \eqref{eq reward functional one-dimensional non-mean-field case}. 

The above representation of the reward function easily allows us to establish the following DPP for multi-dimensional singular control problems:
\begin{equation}
\begin{split}
        &V(t,m)\\
    &= \sup_{\P \in \c P(t,m)} \E^\P\bigg[ V(s,m^\P_{s-}) + \int_t^s f(u,m^\P_u,X_u,\xi_u) du - \sum_{J^d_{[t,s)}(m^\P)} C_{\c P_2}(u,m^\P_{u-}, m^\P_u)\\
    &\ -\sum_{J^c_{[t,s)}(m^\P)\cap J^d_{[t,s)}(\xi)} C_{D^0}(u,m^\P_u,X_{u-},\xi_{u-},\xi_u)
    - \int_{J^c_{[t,s)}(m^\P)\cap J^c_{[t,s)}(\xi)} c(u,m^\P_u,X_u,\xi_u) d\xi_u\bigg].
\end{split}
\end{equation}
In the benchmark case where the volatility is independent of the mean-field term we can use a comparison result for regular control problems obtained by \cite{cosso2021master}  along with the DPP to characterise the value function as the minimal viscosity supersolution to a quasi-variational inequality (QVI) in the Wasserstein space. 

We assume throughout that the coefficients of our SDE \eqref{eq diffusion x} and the reward functions satisfy the following standing assumptions.

\begin{standingAssumption}\label{assumptions A}
The coefficients
\[
    b:[0,T] \times \c P_2 \times \R^d \times \R^l \to \R^d,
    \quad \sigma: [0,T] \times \c P_2 \times \R^d \times \R^l \to \R^{d \times m}
    \quad \text{and} \quad \gamma:[0,T] \to \R^{d \times l}
\]
satisfy the following standing assumptions,
\begin{enumerate}[resume,label=(C\arabic*)]
    \item The functions $b$ and $\sigma$ are continuous in $t$ and Lipschitz continuous in $m,x,\xi$ uniformly in $t$,
    \item \label{assumption gamma continuous} The function $\gamma$ is continuous.
\end{enumerate}
Further, the reward and cost functions 
\[
    f:[0,T] \times \c P_2 \times \R^d \times \R^l \to \R,
    \quad g: \c P_2 \times \R^d \times \R^l \to \R
    \quad \text{and}\quad c:[0,T] \times \c P_2 \times \R^d \times \R^l \to \R^{1 \times l}
\] 
satisfy the following standing assumptions,
\begin{enumerate}[label=(A\arabic*)]
    \item \label{assumption f,c continuous} The functions $f$ and $c$ are locally uniformly continuous.
\footnote{{Since the Wasserstein space $\c P_2$ is not locally compact, locally uniform continuity is slightly stronger than pure continuity.}}
    \item \label{assumption g continuous} The function $g$ is continuous.
    \item \label{assumption c linear growth} The function $c$ is of linear growth in $x,\xi$ locally uniformly in $m$, uniformly in $t$.
    \item \label{assumption f,g quadratic growth} The functions $f$ and $g$ are of quadratic growth in $x,\xi$ locally uniformly in $m$, uniformly in $t$.
\end{enumerate}
\end{standingAssumption}

The rest of this paper is organised as follows.  In Section \ref{section parametrisations}, we introduce two-layer parametrisations {of singular controls and connect them to approximating sequences of continuous controls. Building upon this framework, Section \ref{section reward functional} defines the reward functional in terms of two-layer parametrisations and derives its explicit representation involving minimal jump costs. Finally, Section \ref{section applications} applies these results to establish a dynamic programming principle in Section \ref{section dpp} and, under additional assumptions, to characterise the value function in terms of a quasi-variational inequality in Section \ref{section qvi}.}

\section{Parametrisations of the state-control process}\label{section parametrisations}

This section {introduces the key technical tool of our paper: two-layer parametrisations. Such parametrisations will be used in Section \ref{section reward functional} to} derive a more convenient representation  of the reward functional.

    To motivate the subsequent analysis, let $(\P^n)_n$ be a sequence of continuous controls that converges to a continuous control $\P$ in $\c P_2(D^0)$. Then $(\P^n)_n$ converges also in $\c P_2(C([0,T]{;}\R^d\times\R^l))$ and hence   
    \[
    J_{\c C}(t,m,\P^n)\to J_{\c C}(t,m,\P).
    \]
    For continuous controls $\P\in \c C(t,m)$ we can hence compute the reward directly in terms of the measure $\P$; there is no need to consider all possible approximations. 
    
    In what follows we introduce a technique that allows us to do the same for singular controls. To this end, we equip the state space $D^0$ with the $WM_1$ topology that is defined as follows. 

\begin{definition}[{\cite[Chapter 12]{whitt2002stochastic}}]\label{definition wm1 parametrisation whitt}
    The \emph{thick graph} $G_y$ of a càdlàg path $y = (x,\xi)\in D^0$ is given by
    \begin{equation}
        \begin{split}
    G_y \coloneqq \bigl\{(z,s) \in \R^{d+l}\times [t,T] \bigm\vert z^i \in [y^i(s-),y^i(s)] {\ (\text{resp. }[y^i(s), y^i(s-)]\text{ if }y^i(s) <  y^i(s-))}\\
    \text{for all components }i = 1,\dotso,d+l \bigr\},
        \end{split}
    \end{equation}
    and equipped with the order relation
    \[
    (z_1,s_1) \leq (z_2,s_2)\text{ if }\begin{cases}s_1 < s_2,\text{ or }\\s_1 = s_2\text{ and }|y^i(s-) - z^i_1| \leq |y^i(s-) - z^i_2|\text{ for all }i=1,\dotso,d+l.\end{cases}
    \]
    A \emph{(weak) parametric representation} or \emph{parametrisation of $y$} is a continuous non-decreasing (with respect to the above order relation) mapping $(\hat y,\hat r):[0,1] \to G_y $ with
    \[
    \hat y(0) = y(t-),\qquad \hat y(1) = y(T),\qquad \hat r(0) = t,\qquad \hat r(1) = T.
    \]
    For $y,z\in D^0$ we define
    \[
        d_w(y,z) \coloneqq \sup_{\substack{(\hat y,\hat r)\text{ parametrisation of }y\\(\hat z,\hat s)\text{ parametrisation of }z}} \{\norm{\hat y - \hat z}_{{\infty}} \lor     \norm{\hat r - \hat s}_{{\infty}}\},
    \]
    and say that $x_n\to x$ in the $WM_1$ topology if and only if $d_w(x,x_n) \to 0$. 
\end{definition}    

\begin{example}
{
Let $y = (x, \xi) \in D^0([0, 2]; \mathbb{R}^2)$ be the càdlàg process defined by  
\[
x_s = \xi_s = 
\begin{cases}
0, & s \in [0-, 1), \\
2, & s \in [1, 2],
\end{cases}
\]
with a common jump at $s = 1$.
Then the following are both valid parametrisations 
of $y$: 
\begin{enumerate}[label=(\roman*)]
    \item \emph{Simultaneous interpolation of components}: Both $\hat x^1$ and $\hat\xi^2$ jointly jump along the diagonal  
    \[
    (\hat{x}^1_\lambda, \hat{\xi}^1_\lambda, \hat{r}^1_\lambda) = 
    \begin{cases}
    (0, 0, 3\lambda), & \lambda \in [0, \tfrac{1}{3}], \\
    \big(6\lambda - 2, 6\lambda - 2, 1\big), & \lambda \in (\tfrac{1}{3}, \tfrac{2}{3}], \\
    (2, 2, 3\lambda - 1), & \lambda \in (\tfrac{2}{3}, 1].
    \end{cases}
    \]
    \item \emph{Sequential interpolation of components}: First $\hat x^2$ jumps, then $\hat\xi^2$ jumps  
    \[
    (\hat{x}^2_\lambda, \hat{\xi}^2_\lambda, \hat{r}^2_\lambda) = 
    \begin{cases}
    (0, 0, 3\lambda), & \lambda \in [0, \tfrac{1}{3}], \\
    \big(6\lambda - 2, 0, 1\big), & \lambda \in (\tfrac{1}{3}, \tfrac{1}{2}], \\
    \big(2, 6\lambda - 3, 1\big), & \lambda \in (\tfrac{1}{2}, \tfrac{2}{3}], \\
    (2, 2, 3\lambda - 1), & \lambda \in (\tfrac{2}{3}, 1].
    \end{cases}
    \]
\end{enumerate}
}
\end{example}


\subsection{Motivation and heuristics}\label{subsection motivation parametrisations}

Let us start with a heuristic motivation and switch to the strong formulation of the control problem for simplicity. That is, we assume for the moment that our controls are stochastic processes $\xi$ instead {of} probability measures on the canonical space, and that we optimise the reward functional $J(t,m,(X,\xi))$ introduced in \eqref{reward}. Let us further assume that the functions
\[
    f(t,m,x,\xi) \coloneqq f(t,x,\xi),\quad g(m,x,\xi) \coloneqq g(x,\xi),\quad c(t,m,x,\xi) \coloneqq c(t,x,\xi)
\]
are independent of the mean-field term $m$ and let $(X^n,\xi^n)$ be a sequence of continuous state-control processes that converge to a possibly singular limit $(X,\xi)$ in the $WM_1$ topology a.s.~and in $L^2$. 

The a.s.~convergence in the $WM_1$ topology implies that for almost every $\omega \in \Omega$ and $n\in\N$ there exist sufficiently close parametrisations of $(X^n(\omega),\xi^n(\omega))$ and $(X(\omega),\xi(\omega))$. Although the convergence does not guarantee that the parametrisations can be chosen independently of $n \in \mathbb N$, results like the one established in \cite{cohen2021singular} suggest that there might exists a fixed limit parametrisation
\[
    (\bar X(\omega),\bar\xi(\omega),\bar r(\omega)) \quad \mbox{of} \quad (X(\omega),\xi(\omega))
\]
and suitable parametrisations 
\[
    (\bar X^n(\omega),\bar\xi^n(\omega),\bar r^n(\omega)) \quad \mbox{of} \quad (X^n(\omega),\xi^n(\omega)) 
\]    
in the sense of Definition \ref{definition wm1 parametrisation whitt} such that
\[
(\bar X^n,\bar \xi^n,\bar r^n) \to (\bar X,\bar\xi,\bar r) \text{ in } C\big([t,T]{;}\R^d\times\R^l\times\R\big)\quad\text{in $L^2$ and a.s.}
\]

Assuming that this can indeed be achieved in our case, and since the processes $(X^n,\xi^n)$ are continuous and therefore by definition $X^n_{\bar r^n_u} = \bar X^n_u$ and $\xi^n_{\bar r^n_u} = \bar\xi^n_u$ for all $u\in [0,1]$,
we see that
\begin{equation}
\label{eq motivation reward functional parametrisation}
\begin{split}
J_{\c C}(t,m,(X^n,\xi^n))
&= \E\bigg[\int_t^T f(u,X^n_u,\xi^n_u) du + g(X^n_T,\xi^n_T) - \int_t^T c(u,X^n_u,\xi^n_u) d\xi^n_u \bigg]\\
&= \E\bigg[\int_0^1 f(\bar r^n_u,X^n_{\bar r^n_u},\xi^n_{\bar r^n_u}) d\bar r^n_u + g(\bar X^n_{{1}},\bar \xi^n_{{1}}) - \int_0^1 c(\bar r^n_u,X^n_{\bar r^n_u},\xi^n_{\bar r^n_u}) d\xi^n_{\bar r^n_u} \bigg]\\
&= \E\bigg[\int_0^1 f(\bar r^n_u,\bar X^n_u,\bar \xi^n_u) d\bar r^n_u + g(\bar X^n_{{1}},\bar \xi^n_{{1}}) - \int_0^1 c(\bar r^n_u,\bar X^n_u,\bar \xi^n_u) d\bar \xi^n_u \bigg]\\
&\to \E\bigg[\int_0^1 f(\bar r_u,\bar X_u,\bar \xi_u) d\bar r_u + g(\bar X_{{1}},\bar \xi_{{1}}) - \int_0^1 c(\bar r_u,\bar X_u,\bar \xi_u) d\bar \xi_u \bigg].
\end{split}
\end{equation}
Defining by 
\[
    J(\bar X,\bar \xi,\bar r) \coloneqq
    \E\bigg[\int_0^1 f(\bar r_u,\bar X_u,\bar \xi_u) d\bar r_u + g(\bar X_{{1}},\bar \xi_{{1}}) - \int_0^1 c(\bar r_u,\bar X_u,\bar \xi_u) d\bar \xi_u \bigg]
\]
the cost of a parametrisation $(\bar X,\bar \xi,\bar r)$ suggests that
\begin{equation}\label{eq reward functional motivation supremum parametrisation}
    J(t,m,(X,\xi)) = \sup_{\substack{(\bar X,\bar \xi,\bar r)\text{ \emph{achievable}}\\\text{parametrisation on }[t,T]}} J(\bar X,\bar\xi,\bar r).
\end{equation}
By an \emph{achievable} parametrisation we mean that we can find an approximating sequence of continuous controls with parametrisations that converges to the said parametrisation.

It will turn out that an achievable parametrisation $(\bar X,\bar \xi,\bar r)$ of a given control $(X,\xi)$ is uniquely determined by $(\bar\xi,\bar r)$. 
If the control $\xi$ is one-dimensional, as in e.g.\@ \cite{min_singular_1987,zhu_generalized_1992,taksar_infinite-dimensional_1997,dufour_singular_2004,de2018stochastic}, then there exists only one parametrisation $(\bar\xi,\bar r)$ of $\xi$ in the sense of Definition \ref{definition wm1 parametrisation whitt} up to a reparametrisation of time.\footnote{{This uniqueness also implies that parametrisation is achievable, since it is well-known that every singular control $\xi$ can be approximated in the $WM_1$ topology by a sequence of continuous controls, which then converge to the unique parametrisation of $\xi$. For example, one may define $\xi^n_u= n\int_{(u-\frac 1 n)\lor t}^{u}\xi_s\,ds$ for $u\in [t,T]$, so that each $\xi^n$ is continuous and non-decreasing.}} 
Denoting by $\xi^c$ the continuous part of $\xi$, and by $r_u \coloneqq \inf\{v\in [0,1] | \bar r_v > u\} \land 1$ the inverse of $\bar r$
this implies that in this case
\begin{align*}
J(t,m,(X,\xi)) & = J(\bar X,\bar\xi,\bar r) \\
&= \E\bigg[\int_0^1 f(\bar r_u,\bar X_u,\bar \xi_u) d\bar r_u + g(\bar X_1,\bar \xi_1) - \int_0^1 c(\bar r_u,\bar X_u,\bar \xi_u) d\bar \xi_u \bigg]\\
&= {\E\bigg[\int_t^T f(u,X_u,\xi_u) du + g(X_T,\xi_T) - \int_{\{v \mid \xi_{\bar r_v -} = \xi_{\bar r_v}\}} c(\bar r_v,\bar X_v,\bar\xi_v) d\bar\xi_v}\\*
&\qquad {- \int_{\{v | \xi_{\bar r_v-}\not= \xi_{\bar r_v}\}} c(\bar r_v,\bar X_v,\bar\xi_v) d\bar\xi_v \bigg]}\\
&= {\E\bigg[\int_t^T f(u,X_u,\xi_u) du + g(X_T,\xi_T) - \int_{\{v \mid \xi_{\bar r_v -} = \xi_{\bar r_v}\}} c(\bar r_v,X_{\bar r_v},\xi_{\bar r_v}) d\xi_{\bar r_v}}\\*
&\qquad {- \int_{\{v | v\in [\bar r_{u-},\bar r_u], u\in [t,T], \xi_{u-}\not= \xi_{u}\}}  c(\bar r_v,\bar X_v,\bar\xi_v) d\bar\xi_v  \bigg]}\\
&= {\E\bigg[\int_t^T f(u,X_u,\xi_u) du + g(X_T,\xi_T) - \int_{\{u \mid \xi_{u-} = \xi_u\}} c(u,X_u,\xi_u) d\xi_u}\\*
&\qquad {- \sum_{\xi_{u-}\not= \xi_{u}} \int_{{r_{u-}}}^{{r_u}} c(u,\bar X_{r_u -} + \gamma(u)(\bar\xi_v-\bar\xi_{r_u-}),\bar\xi_v) d\bar\xi_v}\\
&= {\E}\bigg[\int_t^T f(u,X_u,\xi_u) du + g(X_T,\xi_T) - \int_t^T c(u,X_u,\xi_u) d\xi^c_u\\*
&\qquad- \sum_{\xi_{u-}\not= \xi_u} \int_{\xi_{u-}}^{\xi_u} c(u,X_{u-} + \gamma(u) (\zeta - \xi_{u-}), \zeta) d\zeta \bigg].
\end{align*}
In particular, in the one-dimensional case there is no need to consider suprema over achievable parametrisations and we retrieve the cost structure that has already been studied in \cite{min_singular_1987,zhu_generalized_1992,taksar_infinite-dimensional_1997,dufour_singular_2004,de2018stochastic}.

The main challenge in establishing a similar representation in our setting is that our controls are multi-dimensional and that our cost functions additionally depend on a mean-field term. Since the above parametrisations were constructed for each $\omega$ separately, we can not expect the reparametrised joint law $m_{\bar r_u}$ to coincide with the law ${\P_{(\bar X_u,\bar\xi_u)}}$ of the reparametrised processes, {and thus in general $$\P_{(X^n_{\bar r^n_u},\xi^n_{\bar r^n_u})} = m^n_{\bar r^n_u} \not\to m_{\bar r_u} = \P_{(X_{\bar r_u},\xi_{\bar r_u})}.$$ In particular, if we allow our singular cost term $c(t,m,x,\xi)$ to additionally depend on $m$, in the above illustration the reward functional $J_C(t,m,(X^n,\xi^n))$ will not converge to $J(\bar X,\bar\xi,\bar r)$ as $n\to \infty$.}
To overcome this problem, we introduce two-layer parametrisations. The first layer parametrises the $\c P_2$-valued function $u\mapsto m_u = \P_{(X_u,\xi_u)} $ - and hence jointly for all $\omega$. The second layer parametrises the paths of the resulting state-control process for each $\omega$ separately as in the motivation above, suitably adapted to our weak solution concept. 

We formally introduce the concept of two-layer parametrisations of singular controls $\P\in\c P(t,m)$ in the next section. 

\subsection{Two-layer parametrisations}

\subsubsection{First layer}\label{subsubsection first layer}

Our first objective {is} to introduce $\c P_2$-parametrisations that turn the discontinuous mapping $u\mapsto m_u$ into a continuous measure flow $u\mapsto \hat m_u \coloneqq {\hat\P_{(\hat X_u,\hat\xi_u)}}$ by introducing a new time scale $\hat r$ with corresponding state processes $(\hat X,\hat\xi)$. The reparametrisation takes place jointly for all $\omega$ and so the new time scale $\hat r$ is deterministic, and hence so is the inverse time change
\begin{equation}
\label{eq inverse r time change definition}
    r_s \coloneqq \inf \{u\in [0,1] \mid \hat r_u > s\}\land 1.
\end{equation}
The inverse time change $r$ maps {first layer processes $(\hat X,\hat\xi)$ back to the processes $(X,\xi)$ on the original time scale} via
\begin{equation}
(X_u,\xi_u) = (\hat X_{r_u},\hat\xi_{r_u}),\quad\text{for all }u\in [t-,T].
\label{eq subsubsection first layer inverse time change layer transformation}
\end{equation}

The canonical space for the first-layer state-control-dynamics will be $D^0([0,1];\R^d\times\R^l)$; the canonical process on this space will be denoted $(\hat X,\hat\xi)$; our time change will be a deterministic, non-decreasing element of $C([0,1];[t,T])$. {Given a probability measure $\hat\P$ on their canonical space, the} measure flow of the first layer process $(\hat X,\hat\xi)$ will be denoted
\[
    \hat m_s\coloneqq {\hat\P}_{(\hat X_s,\hat\xi_s)}.
\]

\begin{remark}
By construction, the inverse time scale $r$ will be discontinuous at the jump times of $u\mapsto m_u$. The ``extra time interval'' $[r_{u-},r_u]$ on the new time scale allows us to continuously execute the jump from $m_{u-}$ to $m_u$ via $v\mapsto \hat m_v$.
\end{remark}

We are particularly interested in approximating processes that satisfy the SDE \eqref{eq diffusion x} by continuous controls that also satisfy the SDE \eqref{eq diffusion x}. As a result, we require our parametrisations to also satisfy the SDE \eqref{eq diffusion x}, albeit on the time scale $\hat r$ of the parametrisation.\footnote{On the intervals $[r_{u-},r_u]$, we temporarily halt the original state-control dynamics and only execute the jump term.} 
We define weak solutions for the time changed SDE in the spirit of Definition \ref{definition weak solution sde}. The rescaling will lead to a slightly more complex structure of the canonical filtration for the Brownian motion, though. 

\begin{definition}\label{definition weak solution W2 time changed SDE}
    We consider the canonical space
    \[
    \hat\Omega\coloneqq D^0([0,1];\R^d\times\R^l)\times C({[t,T]};\R^m)
    \]
    equipped with the Borel $\sigma$-algebra {and} the canonical process $(\hat X,\hat\xi,\hat W)$. {Undoing the time change of $\hat r$, we consider} the filtration $\hat{\b F}$ {on the original time horizon $[t,T]$} given by
    \[
    \hat{\c F}_u \coloneqq \c F^{\hat W}_u \lor \c F^{\hat X,\hat\xi}_{r_u},\quad{u\in [t,T]},
    \]
    {with $r$ defined by \eqref{eq inverse r time change definition}.}
    We call a tuple $(\hat\P,\hat r) \in \c P_2(\hat\Omega)\times C([0,1];[t,T])$ a weak solution to the following time changed SDE 
    \begin{equation}\label{eq weak solution W2 time changed SDE}
        d\hat X_u = b(\hat r_u,\hat m_u,\hat X_u,\hat\xi_u) d\hat r_u + \sigma(\hat r_u,\hat m_u,\hat X_u,\hat\xi_u) d\hat W_{\hat r_u} + \gamma(\hat r_u) d\hat\xi_u
    \end{equation}
    on the time scale $\hat r$ if and only if the following holds under $\hat\P$:
    \begin{enumerate}[label=(\roman*)]
        \item $\hat r_0 = t$, $\hat r_1 = T$,
        \item $\hat r$ is non-decreasing on $[0,1]$,
        \item $\hat W$ is an $\hat{\b F}$-Brownian motion,
        \item $\hat X$ follows the dynamics \eqref{eq weak solution W2 time changed SDE} driven by the control $\hat\xi$ and the Brownian motion $\hat W$ under the time change $\hat r$ on $[0-,1]$ with $\hat m_u \coloneqq \hat\P_{(\hat X_u,\hat\xi_u)}$ for all $u\in [0-,1]$.
        \footnote{This should be understood as $\hat X_0 = \hat X_{0-} + \gamma(\hat r_0) (\hat\xi_0-\hat\xi_{0-})$ and $(\hat X,\hat\xi,\hat W)$ satisfying the SDE \eqref{eq weak solution W2 time changed SDE} on $[0,1]$.}
    \end{enumerate}
\end{definition}

Having introduced the notion of a time-changed weak solution, we are now ready to introduce the following concept of $\c P_2$-parametrisations of admissible controls.  

\begin{definition}\label{definition W2 SDE parametrisation}
    A \emph{$\c P_2$-parametrisation} of an admissible control $\P\in \c P(t,m)$ is a tuple 
    \[
    (\hat\P,\hat r) \in \c P_2(D^0([0,1]; \R^d\times\R^l))\times C([0,1];[t,T])
    \]
    such that the following holds: 
    \begin{enumerate}[label=(\roman*)]
        \item $\hat\P_{(\hat X_{r_u},\hat \xi_{r_u})_{u\in [0-,1]}} = \P$,
        \item $\hat\xi$ is non-decreasing on $[0-,1]$ $\hat\P$-a.s.,
        \item the measure flow $[0-,1]\ni u\mapsto \hat m_u\coloneqq \hat \P_{(\hat X_u,\hat\xi_u)} \in\c P_2$ is continuous,
        \item $(\hat\P,\hat r)$ is the marginal law of a weak solution to the time changed SDE \eqref{eq weak solution W2 time changed SDE} under the time change $\hat r$ on $[0-,1]$.
    \end{enumerate}
    If $\hat r$ is additionally Lipschitz continuous, we call $(\hat\P,\hat r)$ a \emph{Lipschitz $\c P_2$-parametrisation}.
\end{definition}

\begin{lemma}\label{lemma existence W2 parametrisations}
{
    For every admissible control $\P\in\c P(t,m)$, there exists a $\c P_2$-parametrisation $(\hat\P,\hat r)$ of $\P$.
}
\end{lemma}
\begin{proof}
{
A $\mathcal{P}_2$-parametrisation can be constructed by introducing a time change that inserts linear interpolations at the jump times of the measure flow $t \mapsto m_t = \P_{(X_t, \xi_t)}$. At each jump time $s\in [t,T]$ of $(X,\xi)$, we consider the linear interpolation given by
\[
[0,1] \ni \lambda \mapsto (\lambda X_s + (1-\lambda) X_{s-},\  \lambda \xi_{s} + (1-\lambda) \xi_{s-}).
\]
A precise construction of the time change $\hat r$ and the resulting process $(\hat{X}, \hat{\xi})$ to ensure that $(\hat{\P}, \hat{r})$ satisfies the conditions of a $\mathcal{P}_2$-parametrisation is detailed in Step 1 of the proof of Theorem \ref{theorem reward functional alternative form}. That proof uses $\varepsilon$-optimal interpolations; for this lemma, the simpler linear interpolations above suffice.
}
\end{proof}

While the previously introduced parametrisation guarantees a continuous measure flow, we need to further reparametrise our process on an pathwise level to also obtain continuous sample paths. We hence proceed to our second layer parametrisation. 

\subsubsection{Second layer}\label{subsubsection second layer}
Our next goal is to make the paths $u\mapsto (\hat X_u,\hat\xi_u)$ continuous almost surely by introducing a second time scale $\bar s$ and corresponding processes $(\bar X,\bar\xi)$ on this new time scale. In contrast to the first layer, the time flow $\bar s$ will depend on $\omega$, and thus its right-continuous inverse
\begin{equation}
\label{eq inverse s time change definition}
s_u(\omega) \coloneqq \inf\{v\in [0,1]\mid \bar s_v(\omega) > u\}\land 1
\end{equation}
will be random as well. As for the first layer, the inverse time change $s$ maps {second layer processes $(\bar X,\bar\xi)$ back to first layer processes $(\hat X,\hat \xi)$} via
\begin{equation}
(\hat X_u,\hat\xi_u) = (\bar X_{s_u},\bar\xi_{s_u}),\qquad\text{for all }u\in [0-,1].
\label{eq subsubsection second layer inverse time change layer transformation}
\end{equation}
We {continue working} with the weak formulation and thus with probability measures on the corresponding canonical space $C([0,1];\R^d\times\R^l\times [0,1])$, equipped with the canonical processes $(\bar X,\bar\xi,\bar s)$. 

{Our focus remains on} processes that satisfy our SDE \eqref{eq diffusion x} {and thus} additionally assume that the parametrisations satisfy a corresponding SDE using the new time scales $\hat r$ and $\bar s$. We will again start by introducing the notion of weak solutions for this time changed SDE. The subtle, yet important difference to the first layer parametrisation is the fact that the time scale is now different for different $\omega$ so that we can no longer use the current distribution of the reparametrised process $(\bar X,\bar\xi)$ as the mean-field term in our dynamics. Instead, we will use the process $\hat m$ from the previous layer. 

With a slight abuse of notation we denote the canonical processes for the canonical state space of the weak solution by $(\bar X,\bar\xi,\bar W,\bar s)$ and the weak solution, a probability measure on this space, by $\bar\P$. 

\begin{definition}\label{definition weak solution two-layer time changed SDE}
    We consider the canonical space 
    \[
    \bar\Omega \coloneqq C([0,1];\R^d\times\R^l{\times [0,1])\times C([t,T];\R^m)}
    \]
    equipped with the Borel $\sigma$-algebra {and} the canonical process $(\bar X,\bar\xi,{\bar s,\bar W})$. {Undoing the time change of $\hat r$ and $\bar s$, we consider} the filtration $\bar{\b F}$ {on the original time horizon $[t,T]$} given by 
    \[
    \bar{\c F}_u \coloneqq \c F^{\bar W}_u \lor \c F^{\bar X,\bar\xi,\bar s}_{s_{r_u}},\quad{u\in [t,T]},
    \]
    {with $r$ and $s$ defined by \eqref{eq inverse r time change definition} and \eqref{eq inverse s time change definition}.}
    We call a tuple $(\bar \P,\hat r) \in \c P_2(\bar\Omega)\times C([0,1];[t,T])$ a weak solution to the following time changed SDE 
    \begin{equation}\label{eq weak solution two-layer time changed SDE}
        d\bar X_u = b(\hat r_{\bar s_u},\hat m_{\bar s_u},\bar X_u,\bar\xi_u) d\hat r_{\bar s_u} + \sigma(\hat r_{\bar s_u},\hat m_{\bar s_u},\bar X_u,\bar\xi_u) d\bar W_{\hat r_{\bar s_u}} + \gamma(\hat r_{\bar s_u}) d\bar\xi_u,
    \end{equation}
    under the time changes $\hat r$ and $\bar s$ if and only if under $\bar\P$ the following holds:
    \begin{enumerate}[label=(\roman*)]
        \item {$\hat r_0 = t$, $\hat r_1 = T$,}
        \item {$\bar s_0=0$, $\bar s_1=1$ $\bar\P$-a.s.,}
        \item $\hat r$ and $\bar s$ are non-decreasing $\bar\P$-a.s.,
        \item $\bar W$ is an $\bar{\b F}$-Brownian motion,
        \item $\bar X$ follows the dynamics \eqref{eq weak solution two-layer time changed SDE} driven by the control $\bar\xi$ and the Brownian motion $\bar W$ under the time changes $\hat r$ and $\bar s$ on $[0,1]$ with $\hat m_u \coloneqq \bar\P_{(\bar X_{s_u},\bar\xi_{s_u})}$ for all $u\in [0,1]$.
    \end{enumerate}
\end{definition}

We are now ready to introduce the notion of a two-layer parametrisation. 

\begin{definition}\label{definition two-layer SDE parametrisation}
    A \emph{two-layer parametrisation} of an admissible control $\P\in\c P(t,m)$ is a tuple 
    \[
    \big((\hat r,\hat m),\bar\P \big)\in C([0,1];[t,T]\times\c P_2)\times \c P_2(C([0,1];\R^d\times\R^l\times [0,1]))
    \]
    such that the following holds:
    \begin{enumerate}[label=(\roman*)]
        \item $(\hat\P,\hat r) \coloneqq (\bar\P_{(\bar X_{s_u},\bar \xi_{s_u})_{u\in [0-,1]}},\hat r)$ is a $\c P_2$-parametrisation of $\P$,
        \item $\hat m_u = \hat\P_{(\hat X_u,\hat\xi_u)} = \bar\P_{(\bar X_{s_u},\bar\xi_{s_u})}$ for all $u\in [0-,1]$,
        \item $\bar\xi$ is non-decreasing $\bar\P$-a.s.,
        \item $(\bar\P,\hat r)$ is the marginal law of a weak solution to the time changed SDE \eqref{eq weak solution two-layer time changed SDE} under the time changes $\hat r$ and $\bar s$.
    \end{enumerate}
    If further $(\bar\xi,\bar s)$ are Lipschitz continuous, uniformly for $\bar\P$-almost all $\bar\omega$ and $(\hat\P,\hat r)$ is also a Lipschitz $\c P_2$-parametrisation, then we call $((\hat r,\hat m),\bar\P)$ a \emph{Lipschitz two-layer parametrisation}.
\end{definition}

\begin{lemma}\label{lemma existence two layer parametrisations}
{
    For every admissible control $\P\in\c P(t,m)$, there exists a two-layer parametrisation $((\hat m,\hat r),\bar\P)$ of $\P$.
}
\end{lemma}
\begin{proof}
{
Given a $\mathcal{P}_2$-parametrisation $(\hat{\P}, \hat{r})$ of $\P$ (which exists by Lemma \ref{lemma existence W2 parametrisations}), a two-layer parametrisation can be constructed by introducing a further time change that inserts linear interpolations at the jump times of the process $(\hat{X}, \hat{\xi})$. At each jump time $\tau(\hat\omega)$ of the first layer process $(\hat X(\hat\omega),\xi(\hat\omega))$, we consider the pathwise linear interpolation given by
\[
[0,1] \ni \lambda \mapsto (\lambda\hat X_{\tau(\hat\omega)}(\hat\omega) + (1-\lambda)\hat X_{\tau(\hat\omega)-}(\hat\omega),\  \lambda \hat \xi_{\tau(\hat\omega)}(\hat\omega) + (1-\lambda) \hat \xi_{\tau(\hat\omega)-}(\hat\omega)).
\]
A precise construction of the time change $\bar{s}$ and the resulting processes $(\bar{X}, \bar{\xi})$, including the verification of measurability and the conditions for a two-layer parametrisation, is detailed in Step 2 of the proof of Lemma \ref{lemma reward functional lower bound}. That construction uses optimal interpolations; for this lemma, the simpler linear interpolations above suffice.
}
\end{proof}

\subsubsection{From two-layer parametrisations back to controls}\label{subsubsection recover the previous layers from bar P}

Given a two-layer parametrisation $((\hat m,\hat r),\bar\P)$ we have seen in Section \ref{subsubsection second layer} that, additionally to the canonical processes $(\bar X,\bar\xi,\bar s)$ on $\bar\Omega$ under $\bar\P$, we can also recover the first layer via
\begin{equation}\label{eq two layer parametrisations hat given bar}
(\hat X_u,\hat\xi_u) \coloneqq (\bar X_{s_u},\bar\xi_{s_u}).
\end{equation}
From this, we have seen in Section \ref{subsubsection first layer} that we can also recover the original singular control via
\begin{equation}\label{eq two layer parametrisations tilde given hat}
(\tilde X_u,\tilde \xi_u) \coloneqq (\hat X_{r_u},\hat\xi_{r_u}) = (\bar X_{s_{r_u}},\bar\xi_{s_{r_u}}).
\end{equation}
Further, if we consider the extended space where $\bar W$ is an $\bar{\b F}$-Brownian motion such that $\bar\P$ is a weak solution to the time changed SDE \eqref{eq weak solution two-layer time changed SDE}, then $(\bar\P_{(\hat X,\hat\xi,\bar W)},\hat r)$ will also be a weak solution to the time changed SDE \eqref{eq weak solution W2 time changed SDE} corresponding to the $\c P_2$-parametrisation $(\hat\P,\hat r)$ and $\bar\P_{(\tilde X,\tilde\xi,\bar W)}$ will be a weak solution to the original SDE \eqref{eq diffusion x} corresponding to the singular control $\P$.
Combining the SDEs \eqref{eq diffusion x}, \eqref{eq weak solution W2 time changed SDE} and \eqref{eq weak solution two-layer time changed SDE} with \eqref{eq two layer parametrisations hat given bar} and \eqref{eq two layer parametrisations tilde given hat} also shows the following relation between the different layers,
\begin{equation}\label{eq two layer parametrisations X hat, X bar given xi hat, xi bar}
    \hat X_u = \tilde X_{\hat r_u -} + \gamma(\hat r_u) (\hat\xi_u - \tilde\xi_{\hat r_u -}),\qquad \bar X_u = \hat X_{\bar s_u -} + \gamma(\hat r_{\bar s_u}) (\bar\xi_u - \hat\xi_{\bar s_u -}).
\end{equation}

Recovering the previous layers with corresponding weak solutions will be helpful throughout the rest of the paper. In what follows, whenever we consider a two-layer parametrisation, we will always assume implicitly that the aforementioned construction has been carried out.

\subsection{Convergence of two-layer parametrisations}\label{subsection two layer convergence}
Let us now introduce the convergence concept for two-layer parametrisation that will allow us to construct a continuous reward function for two-layer parametrisations in the next section.

\begin{definition}
We say that a sequence $\big((\hat m^n,\hat r^n),\bar\P^n\big)_{n \in \N}$ of two-layer parametrisations converges to a two-layer parametrisation $((\hat m,\hat r),\bar\P)$ on the time interval $[t,T]$, denoted
\[
((\hat m^n,\hat r^n),\bar\P^n) \to {(}(\hat m,\hat r),\bar\P),
\]
if the following holds:
\begin{enumerate}[label=(\roman*)]
    \item $(\hat m^n,\hat r^n) \to (\hat m,\hat r)$ in $C([0,1];\c P_2\times [t,T])$,
    \item $\bar\P^n\to \bar\P$ in $\c P_2(C([0,1];\R^d\times\R^l\times [0,1]))$.
\end{enumerate}
\end{definition}

Convergence of two-layer parametrisation implies also convergence of the underlying singular controls.
The next result shows that given an approximating sequence of continuous controls, we can obtain convergent two-layer parametrisations, similar to our motivation in Section \ref{subsection motivation parametrisations}. Its proof can be found in the {Section \ref{appendix proof of theorem construct limit parametrisation from convergent sequence}}.

\begin{theorem}\label{theorem construct limit parametrisation from convergent sequence}
    Let $(\P^n)_{n}\subseteq \c C(t,m)$ and $\P \in \c P(t,m)$ be such that
    \[
    \P^n\to\P\text{ in }\c P_2(D^0).
    \]
    Then there exists a subsequence $(\P^{n_k})_{k}$ with two-layer parametrisations $((\hat m^{n_k},\hat r^{n_k}), \bar\P^{n_k})$ and a two-layer parametrisation $((\hat m, \hat r), \bar\P )$ of $\P$ such that
    \[
    ( (\hat m^{n_k},\hat r^{n_k}),\bar\P^{n_k} ) \to ( (\hat m,\hat r),\bar\P ).
    \]
\end{theorem}

While the previous result showed how to obtain a limit two-layer parametrisation from a given sequence of approximating continuous controls, the following theorem can be seen as establishing a converse result. Given a limit two-layer parametrisation, it allows us to construct an approximating sequence of two-layer parametrisations corresponding to continuous, bounded velocity controls. 
The proof can be found in the {Section \ref{appendix proof of theorem two layer parametrisation bounded velocity approximation}}.

\begin{theorem}\label{theorem two layer parametrisation bounded velocity approximation}
    For every two-layer parametrisation $((\hat m,\hat r),\bar\P)$ of a singular control $\P\in \c P(t,m)$ we can find a sequence of continuous bounded velocity controls $(\P^n)_{n}\subseteq \c L(t,m)$ with corresponding sequence of two-layer parametrisations $((\hat m^n,\hat r^n),\bar\P^n)_n$ such that
    \[
    ((\hat m^n,\hat r^n),\bar \P^n) \to ((\hat m,\hat r),\bar\P).
    \]
\end{theorem}

We conclude this section with the following corollary. {It} shows that the set of continuous controls is dense in the set of admissible controls. In particular, the reward function $J$ in \eqref{eq reward function general definition} is well-defined. 

\begin{corollary}{
For any initial data $(t,m)$ the set  $\c L(t,m)$ is dense in $\c P(t,m)$.}
\end{corollary}
\begin{proof}
{Given an admissible control $\P\in\c P(t,m)$, Lemma \ref{lemma existence two layer parametrisations}} shows that we can construct a two-layer parametrisation of $\P$. {Then}, Theorem \ref{theorem two layer parametrisation bounded velocity approximation} implies the existence of a sequence of continuous, bounded velocity controls {$(\P^n)_n\subseteq\c L(t,m)$} approximating $\P$ together with its two-layer parametrisation, {and in particular satisfying $\P^n\to\P$.}
\end{proof}

\subsubsection{{Proof of Theorem \ref{theorem construct limit parametrisation from convergent sequence}}}\label{appendix proof of theorem construct limit parametrisation from convergent sequence}

    We start by setting up our probability space. To this end, let $(\tilde\P^n)_n\subseteq\c P_2(\tilde\Omega)$ be weak solutions to the SDE \eqref{eq diffusion x} corresponding to the admissible controls $(\P^n)_n$. We recall that $\tilde\Omega$ is equipped with the canonical process $(\tilde X,\tilde\xi,\tilde W)$. Then the sequence $(\tilde\P^n)_n$ is relatively compact in $\c P_2(\tilde\Omega)$ and we can w.l.o.g. assume that there exists some $\tilde\P\in \c P_2(\tilde\Omega)$ such that
    \[
    \tilde\P^n \to \tilde\P\qquad\text{in }\c P_2(\tilde\Omega).
    \]

    We already know that $\tilde\P_{(\tilde X,\tilde\xi)} = \P$. Additionally, since $(\tilde X,\tilde\xi,\tilde W)$ satisfies the SDE \eqref{eq diffusion x} under $\tilde\P^n$ for every $n\in\N$, we can conclude from \cite[Theorem 7.10]{Kurtz1996} that $(\tilde X,\tilde\xi,\tilde W)$ also satisfies \eqref{eq diffusion x} under $\tilde\P$.
    Moreover, since $\tilde W$ is an $\b F^{\tilde X,\tilde \xi,\tilde W}$-Brownian motion for all $\tilde\P^n$, it follows that $\tilde W$ will also be an $\b F^{\tilde X,\tilde\xi,\tilde W}$-Brownian motion under the limit measure $\tilde\P$. Thus $\tilde\P$ is a weak solution to the SDE \eqref{eq diffusion x} corresponding to the singular control $\P$.

    By Skorokhod's representation theorem there exists a common probability space $(\check\Omega,\check{\c F},\check\P)$ with processes
    \[
    (X^n,\xi^n,W^n) \to (X,\xi,W)\text{ in }{D^0\times C([t,T];\R^m)}\qquad\check\P\text{-a.s. and in }L^2,
    \]
    such that $\check\P_{(X^n,\xi^n,W^n)} = \tilde\P^n$ and $\check\P_{(X,\xi,W)} = \tilde\P$. Moreover, we introduce the following auxiliary processes:
    \[
    L^n_u \coloneqq \int_t^u b(v,m^n_v,X^n_v,\xi^n_v) dv + \int_t^u \sigma(v,m^n_v,X^n_v,\xi^n_v) dW^n_v,\qquad \Gamma^n_u \coloneqq \int_{[t-,u]} \gamma(v) d\xi^n_v,
    \]
    and
    \[
    L_u \coloneqq \int_t^u b(v,m_v,X_v,\xi_v) dv + \int_t^u \sigma(v,m_v,X_v,\xi_v) d W_v,\qquad \Gamma_u \coloneqq \int_{[t-,u]} \gamma(v) d\xi_v.
    \]
    By construction,
    \[
    \Gamma^n \to \Gamma\text{ in }D^0\qquad\check\P\text{-a.s. and in }L^2,
    \]
    and thus also
    \[
    L^n = X^n - X^n_{t-} - \Gamma^n \to X - X_{t-} - \Gamma = L\text{ in }D^0\qquad\check\P\text{-a.s. and in }L^2.
    \]
    Since $L^n$ and $L$ are continuous, the above convergence can be strengthened to 
    \begin{equation}\label{eq theorem 3.6 uniform convergence L}
    L^n \to L\text{ in }C([t,T];\R^d)\qquad\check\P\text{-a.s. and in }L^2.
    \end{equation}
    This completes our probabilistic setup, and we are now ready to construct the desired two-layer parametrisations of $(\P^n)_n$ and the limit $\P$. We constructing the two layers successively. Afterwards, we show that the limit is indeed a valid two-layer parametrisation of $\P$.
\begin{enumerate}[wide]
   
    \item[\textbf{ Step 1. \emph{The first layer.}}] To construct the desired parametrisations we start by constructing the first layer $(\hat \P,\hat r)$. To this end, we first introduce suitable $\c P_2$-parametrisations $(\hat \P^n,\hat r^n)$ of $\P^n$. Inspired by the construction in \cite{cohen2021singular},
    we start by defining $r^n$ as
    \[
    r^n_u \coloneqq \frac{(u-t) + \norm{\xi^n_u - \xi^n_{t-}}_{L^2(\check\P)}}{(T-t) + \norm{\xi^n_T - \xi^n_{t-}}_{L^2(\check\P)}}, \qquad\text{for all } u\in [t-,T].
    \]
    Since the control process $\xi^n$ is monotone and continuous, we see that $r^n$ is strictly increasing and continuous, and $r^n_{t-} = 0$, $r^n_T = 1$. Now we define the time change $\hat r^n$ as the inverse of $r^n$, so
    \[
    \hat r^n_u \coloneqq (r^n)^{-1}_u,
    \]
    and introduce the corresponding parametrised processes
    \[
    \hat X^n_u = X^n_{\hat r^n_u},\qquad \hat \xi^n_u = \xi^n_{\hat r^n_u},\qquad \hat L^n_u = L^n_{\hat r^n_u},\qquad \hat \Gamma^n_u = \Gamma^n_{\hat r^n_u},\qquad\text{for all } u\in [0,1].
    \]
    By construction, the pairs $(\check\P_{(\hat X^n,\hat\xi^n)},\hat r^n)$ are $\c P_2$-parametrisations of $\P^n$, since $\hat r^n$ is continuous and bijective and $W^n$ is an $\b F^{X^n,\xi^n,W^n}$-Brownian motion together with $\c F^{W^n}_u\lor \c F^{\hat X^n,\hat\xi^n}_{r^n_u} = \c F^{X^n,\xi^n,W^n}_u$.
    
    Next, we show that,
    \[
    (\check\P_{(\hat X^n,\hat\xi^n)},\hat r^n)_{n} \subseteq \c P_2(D^0)\times C([t,T];[0,1])
    \]
    is relatively compact. For the $\hat r^n$-component, the result follows from the Arzelà-Ascoli theorem since by monotonicity of $u\mapsto \norm{\xi^n_u-\xi^n_{t-}}_{L^2(\check\P)}$,
    \[
    |r^n(u) - r^n(v)| \geq \frac{|u-v|}{(T-t) + \norm{\xi^n_T - \xi^n_{t-}}_{L^2(\check\P)}},
    \]
    which implies that the inverse $\hat r^n$ is uniformly Lipschitz continuous {since $\xi^n\to\xi$ in $L^2$ 
    and in particular $\sup_n \norm{\xi^n_T - \xi^n_{t-}}_{L^2(\check\P)} < \infty$}. For $\check\P_{(\hat X^n,\hat\xi^n)}$, we note that since $\hat r^n$ is uniformly Lipschitz and $(\check\P_{(X^n,\xi^n)})_n\subseteq \c P_2(D^0)$ is relatively compact, it follows from \cite[Theorem 12.12.2]{whitt2002stochastic} that $(\check\P_{(\hat X^n,\hat\xi^n)})_n$ is also relatively compact in $\c P(D^0)$. Finally, since $(\check\P_{(X^n,\xi^n)})_n$ is $\c P_2$-tight, the sequence $(\check\P_{(\hat X^n,\hat\xi^n)})_n$ is $\c P_2$-tight as well; together with the previous result we obtain the relative compactness in $\c P_2(D^0)$.
    In particular, we can assume without loss of generality that
    \[
    (\check\P_{(\hat X^n,\hat\xi^n)},\hat r^n) \to (\hat\P,\hat r) { = (\hat\P_{(\hat X,\hat\xi)},\hat r)} \qquad\text{in }\c P_2(D^0)\times C([t,T];[0,1]),
    \]
    for some $(\hat\P,\hat r) \in \c P_2(D^0)\times C([t,T];[0,1])$.%
    \footnote{{$(\hat X,\hat\xi)$ denotes the canonical process on $D^0$, see also Definitions \ref{definition weak solution W2 time changed SDE} and \ref{definition W2 SDE parametrisation}.}}
    Note that since $\hat\xi^n$ is non-decreasing $\check\P$-a.s., this implies that $\hat\xi$ is also non-decreasing under $\hat\P$.
    
   Using the notation $\hat m^n_u\coloneqq \check\P_{(\hat X^n_u,\hat\xi^n_u)}$, our next step is to show that
    \[
    \bigl((\hat m^n_u)_{u\in [0,1]}\bigr)_n \subseteq C([0,1];\c P_2)
    \]
    is relatively compact. Again we want to apply Arzelà-Ascoli, see \cite[Theorem 3.4.20]{engelking1989general},
    and thus have to show equicontinuity and pointwise relative compactness. Regarding pointwise relative compactness, we note that for every $u\in [0,1]$
    \[
    (\hat m^n_u)_n = (\check\P_{(X^n_{\hat r^n_u},\xi^n_{\hat r^n_u})})_n,
    \]
    is relatively compact since $(\check\P_{(X^n,\xi^n)})_n$ is relatively compact in $\c P_2(D^0)$.
    For the equicontinuity, we note that under $\check\P$, by monotonicity, for $t-\leq v \leq u\leq T$,
    \begin{equation}\label{eq theorem 3.6 monotonicity of xi}
    \begin{split}
    \norm{\xi^n_u - \xi^n_v}_{L^2(\check\P)}^2
    &= \norm{\xi^n_u - \xi^n_{t-}}_{L^2(\check\P)}^2 - \norm{\xi^n_v - \xi^n_{t-}}_{L^2(\check\P)}^2 - 2 \<\xi^n_v - \xi^n_{t-},\xi^n_u - \xi^n_v\>_{L^2(\check\P)}\\
    &\leq \norm{\xi^n_u - \xi^n_{t-}}_{L^2(\check\P)}^2 - \norm{\xi^n_v - \xi^n_{t-}}_{L^2(\check\P)}^2,
    \end{split}
    \end{equation}
    and thus, for $0-\leq v \leq u\leq 1$,
    \begin{equation}\label{eq theorem 3.6 <= hat xi lipschitz}
    \begin{split}
        \norm{\hat\xi^n_u - \hat\xi^n_v}_{L^2(\check\P)}^2
        &= \norm{\xi^n_{\hat r^n_u} - \xi^n_{\hat r^n_v}}_{L^2(\check\P)}^2\\
        &\leq \norm{\xi^n_{\hat r^n_u} - \xi^n_{t-}}_{L^2(\check\P)}^2 - \norm{\xi^n_{\hat r^n_v} - \xi^n_{t-}}_{L^2(\check\P)}^2\\
        &= \Bigl|\bigl((T-t)+\norm{\xi^n_T - \xi^n_{t-}}_{L^2(\check\P)}\bigr) u - (\hat r^n_u-t)\Bigr|^2\\
        &\qquad- \Bigl|\bigl((T-t)+\norm{\xi^n_T - \xi^n_{t-}}_{L^2(\check\P)}\bigr) v - (\hat r^n_v-t)\Bigr|^2\\
        &\leq \Bigl|\bigl((T-t)+\norm{\xi^n_T - \xi^n_{t-}}_{L^2(\check\P)}\bigr) (u+v) - (\hat r^n_u-t) - (\hat r^n_v-t)\Bigr|\\
        &\qquad\cdot\Bigl|\bigl((T-t)+\norm{\xi^n_T - \xi^n_{t-}}_{L^2(\check\P)}\bigr) (u-v) - (\hat r^n_u-\hat r^n_v)\Bigr|\\
        &\leq C (u-v),
    \end{split}
    \end{equation}
    where $C$ is a constant that is independent of $u,v$ and $n$. Now using that $\hat r^n$ is Lipschitz uniformly in $n$, the time changed SDE \eqref{eq weak solution W2 time changed SDE} implies that the $\hat X^n$-component is also equicontinuous
    and thus
    \[
    \c W_2(\hat m^n_u,\hat m^n_v) \leq \norm{\hat X^n_u - \hat X^n_v}_{L^2(\check\P)}^2 + \norm{\hat \xi^n_u - \hat \xi^n_v}_{L^2(\check\P)}^2 \leq C(u-v).
    \]
    
    Therefore we can assume w.l.o.g. that $((\hat m^n_u)_{u\in [0,1]})_n$ converges in $C([0,1];\c P_2)$. Finally, using the uniqueness of the pointwise limits $(\hat m^n_u)_n$ for every $u\in [0,1]$, we get
    \[
    (\hat m^n_u)_{u\in [0,1]} \to (\hat m_u)_{u\in [0,1]} \coloneqq (\hat\P_{(\hat X_u,\hat\xi_u)})_{u\in [0,1]}\in C([0,1];\c P_2).
    \]
    
    \item[\textbf{ Step 2. \emph{The second layer.}}] Now we construct the second layer. We start by constructing suitable two-layer parametrisations $((\hat m^n,\hat r^n),\bar\P^n)$ for $\P^n$ building on the previously constructed $\c P_2$-parametrisations $(\hat \P^n,\hat r^n)$. 
    
    Since we reparametrise the process $\check\omega$-wise, we need to ensure that the resulting process remains measurable. Furthermore, we require the original Brownian motions $\tilde W^n$ to remain a Brownian motion for the reparametrised layer. For this we need to guarantee that the processes $s^n$ and $\bar s^n$ are adapted to the filtration $\b F^{\hat X^n,\hat\xi^n}$. This motivates our following definition, which is again inspired by \cite[Theorem 3.1]{cohen2021singular}.
    We define for each $\check\omega\in \check\Omega$
    \[
    s^n_u(\check\omega) \coloneqq \frac{u + \arctan(\Var(\hat\xi^n(\check\omega), [0,u]))}{1+\frac\pi 2},
    \]
    where $\Var$ denotes the total variation. Also as before we denote the inverse of $s^n(\check\omega)$ by
    \[
    \bar s^n_u(\check\omega)\coloneqq \inf\{v\in [0,1]\mid s^n_v(\check\omega) > u\}\land 1,
    \]
    and introduce the reparametrised processes
    \[
    \bar X^n_u(\check\omega) \coloneqq \hat X^n_{\bar s^n_u(\check\omega)}(\check\omega),\quad
    \bar\xi^n_u(\check\omega) \coloneqq \hat\xi^n_{\bar s^n_u(\check\omega)}(\check\omega),\quad
    \bar L^n_u(\check\omega) \coloneqq \hat L^n_{\bar s^n_u(\check\omega)}(\check\omega),\quad
    \bar \Gamma^n_u(\check\omega) \coloneqq \hat\Gamma^n_{\bar s^n_u(\check\omega)}(\check\omega).
    \]
    By construction both $s^n$ and $\bar s^n$ are $\b F^{\hat\xi^n}$-adapted. Since $\bar s^n$ is continuous and bijective, the fact that $W^n$ is an $\hat{\b F}^n$-Brownian motion, where
    \[
    \hat{\c F}^n_u \coloneqq \c F^{W^n}_u\lor \c F^{\hat X^n,\hat\xi^n}_{r^n_u},
    \]
    implies that it is also an $\bar{\b F}^n$-Brownian motion, where
    \[
    \bar{\c F}^n_u \coloneqq \c F^{W^n}_u\lor \c F^{\bar X^n,\bar\xi^n,\bar s^n}_{s^n_{r^n_u}}.
    \]
   
    The main difference to the previous ($\c P_2$-level) layer is the way we extract a suitable subsequence. It is not enough to find subsequences for each $\check\omega\in\check\Omega$, but we need to construct one subsequence that guarantee convergence in $\c P_2(C([0,1];\R^d\times\R^l\times [t,T]))$. {To this end, we prove that the sequence}
    \[
    (\bar\P^n)_n \coloneqq (\check\P_{(\bar X^n,\bar \xi^n,\bar s^n)})_n\subseteq \c P_2(C([0,1];\R^d\times\R^l\times [0,1]))
    \]
    is {relatively} compact.

    {We start with the $\bar\xi^n$- and $\bar s^n$-components; our main tool will be Arzelà-Ascoli's theorem.}
    Regarding the {required} equicontinuity, we first observe that due to the monotonicity of {$\bar s^n$, $\bar\xi^n$ and $\arctan$}, for all $n\in\N$,
    \begin{equation}
    \begin{split}
    &|\bar s^n_u(\check\omega) - \bar s^n_v(\check\omega)|\\
    &= {(1 + \tfrac\pi 2)| u - \arctan(\Var(\hat\xi^n(\check\omega),[0,\bar s^n_u(\check\omega)])) - v + \arctan(\Var(\hat\xi^n(\check\omega),[0,\bar s^n_v(\check\omega)]))|} \\
    &\leq (1+\tfrac\pi 2)|u-v|,\qquad\check\P\text{-a.s.},
    \end{split}
    \end{equation}
    {using that $u\mapsto \arctan(\Var(\hat\xi^n(\check\omega),[0,\bar s^n_u(\check\omega)]))$ is non-decreasing,} and also
    \begin{equation}
    \begin{split}
    &\bigl|\arctan(\Var(\bar\xi^n(\check\omega),[0,u])) - \arctan(\Var(\bar\xi^n(\check\omega),[0,v]))\bigr|\\
    &{= \bigl|\arctan(\Var(\hat\xi^n(\check\omega),[0,\bar s^n_u(\check\omega)]))  - \arctan(\Var(\hat\xi(\check\omega),[0,\bar s^n_v(\check\omega)]))\bigr|}\\
    &\leq \bigl|\bar s^n_u(\check\omega) + \arctan(\Var(\hat\xi^n(\check\omega),[0,\bar s^n_u(\check\omega)])) - \bar s^n_v(\check\omega) - \arctan(\Var(\hat\xi({\check\omega}),[0,\bar s^n_v(\check\omega)]))\bigr|\\ 
    &= (1 + \tfrac \pi 2) |u-v|.
    \end{split}
    \end{equation}
    
    Using that $\arctan$ is uniformly continuous on every compact interval together with the inequality
    \[
    |\bar\xi^n_u(\check\omega) - \bar\xi^n_v(\check\omega)| \leq |\Var(\bar\xi^n(\check\omega),[0,u]) - \Var(\bar\xi^n(\check\omega),[0,v])|,\qquad\check\P\text{-a.s.},
    \]
    we obtain the desired equicontinuity for the $\bar\xi^n$-component.
    {Noting that $\bar\xi^n_0 = \xi^n_{t-} \to \xi_{t-}$ in $L^2$ and $\bar s^n_0 = 0$, we also see that $(\P_{(\bar\xi^n_0,\bar s^n_0)})_n\subseteq \c P(\R^l\times [0,1])$ is tight.}
    Thus, Arzelà-Ascoli together with Prokhorov's theorem show that $(\check\P_{(\bar\xi^n,\bar s^n)})_n$ is relatively compact in the topology of weak convergence. To extend this result to the $\bar X^n$-component, we will use that, $\check\P$-a.s,
    \begin{equation}\label{eq theorem 3.6 composition bar X L Gamma}
    \bar X^n_u = X^n_{\hat r^n_{\bar s^n_u}} = X^n_{t-} + L^n_{\hat r^n_{\bar s^n_u}} + \Gamma^n_{\hat r^n_{\bar s^n_u}} = X^n_{t-} + \bar L^n_u + \bar\Gamma^n_u,\qquad\text{for all } u\in [0,1].
    \end{equation}
    We already know that $X^n_{t-} \to X_{t-}$ and due to
    \[
    \bar\Gamma^n_u = \int_0^u \gamma(\hat r^n_{\bar s^n_v}) d\bar\xi_v^n,
    \]
    together with the relative compactness of $(\check\P_{(\bar\xi^n,\bar s^n)})_n$, we also see that the sequence $(\check\P_{(\bar\Gamma^n)})_n$ {and thus also the sequence $(\check\P_{(\bar\Gamma^n,\bar\xi^n,\bar s^n)})_n$} is relatively compact in the topology of weak convergence. Finally, for $\bar L^n$, we note that due to the uniform convergence in \eqref{eq theorem 3.6 uniform convergence L},
    \[
    \bar L^n_u = L^n_{\hat r^n_{\bar s^n_u}} \to L_{\hat r_{\bar s_u}} \eqqcolon \bar L_u,\qquad\text{uniformly in } u\in [0,1],\qquad\check\P\text{-a.s. and in }L^2.
    \]

    Together with \eqref{eq theorem 3.6 composition bar X L Gamma}, this now implies that $(\bar\P^n)_n = (\check\P_{(\bar X^n,\bar\xi^n,\bar s^n)})_n$ is relatively compact in $\c P(C([0,1];\R^d\times\R^l\times [0,1]))$. We again strengthen this to relative compactness in $\c P_2(C([0,1];\R^d\times\R^l\times [0,1]))$ by using the $\c P_2$-tightness of $(\check\P_{(\hat X^n,\hat\xi^n)})_n$ which implies that $(\check\P_{(\bar X^n,\bar\xi^n,\bar s^n)})_n$ is also $\c P_2$-tight.
    We can thus assume w.l.o.g. that
    \footnote{{$(\bar X,\bar\xi,\bar s)$ denotes the canonical process on $C([0,1];\R^d\times\R^l\times [0,1])$, see also Definitions \ref{definition weak solution two-layer time changed SDE} and \ref{definition two-layer SDE parametrisation}.}}
    \[
    \check\P_{(\bar X^n,\bar\xi^n,\bar s^n)} = \bar\P^n\to \bar\P { = \bar\P_{(\bar X,\bar\xi, \bar s)}}\qquad\text{in }\c P_2(C([0,1];\R^d\times\R^l\times [0,1])).
    \]
    Since $\bar\xi^n$ is non-decreasing under $\check\P$, this implies that $\bar\xi$ is also non-decreasing under $\bar\P$.
    
    \item[\textbf{ Step 3. \emph{Verification.}}] In the previous steps we have constructed two-layer parametrisations $((\hat m^n,\hat r^n),\bar\P^n)_n$ of $(\P^n)_n$ and a corresponding limit $((\hat m,\hat r),\bar\P)$ such that 
    \[
    ((\hat m^n,\hat r^n),\bar\P^n) \to ((\hat m,\hat r),\bar\P).
    \]
    It only remains to show that $((\hat m,\hat r),\bar\P)$ is indeed a two-layer parametrisation of $\P$.
    To this end, we observe that we have shown in the previous steps that the marginal distributions
    {
    \[
    \begin{split}
        (\check\P_{(X^n,\xi^n,W^n)})_n &\subseteq \c P_2(\tilde\Omega),\\
        (\check\P_{(\hat X^n,\hat\xi^n)})_n &\subseteq \c P_2(D^0([0,1];\R^d\times\R^l)),\\
        (\check\P_{(\bar X^n,\bar\xi^n,\bar s^n)})_n &\subseteq \c P_2(C([0,1];\R^d\times\R^l\times [0,1])),
    \end{split}
    \]
    are relatively compact. Consequently, the sequence of their joint laws,}
    \[
    \big(\check\P_{(X^n,\xi^n,W^n,\hat X^n,\hat\xi^n,\bar X^n,\bar\xi^n,\bar s^n)}\big)_n\subseteq\c P_2\big(\tilde\Omega\times D^0([0,1];\R^d\times\R^l)\times C([0,1];\R^d\times\R^l\times [0,1])\big),
    \]
    is {also} relatively compact. {Therefore}, we can assume w.l.o.g. that there exists a limit $\check\P^\infty$ such that
    \begin{equation}
    \label{eq theorem 3.6 joint law convergence}
    \begin{split}
    & \check\P_{(X^n,\xi^n,W^n,\hat X^n,\hat\xi^n,\bar X^n,\bar\xi^n,\bar s^n)}  \to \check\P^\infty \\
    & \qquad \qquad \ \text{in }\c P_2\big(\tilde\Omega\times D^0([0,1];\R^d\times\R^l)\times C([0,1];\R^d\times\R^l\times [0,1])\big).
     \end{split}
    \end{equation}
    From the previous steps, we also know the following marginal distributions of this limit $\check\P^\infty$,
    \[
    \check\P^\infty_{(X,\xi,W)} = \tilde\P,\qquad \check\P^\infty_{(\hat X,\hat\xi)} = \hat\P,\qquad \check\P^\infty_{(\bar X,\bar\xi,\bar s)} = \bar\P,
    \]
    where $(X,\xi,W,\hat X,\hat\xi,\bar X,\bar\xi,\bar s)$ denotes the canonical process of $\tilde\Omega\times D^0([0,1];\R^d\times\R^l)\times C([0,1];\R^d\times\R^l\times [0,1])$.
    
    The {convergence in \eqref{eq theorem 3.6 joint law convergence}} allows us to use Skorokhod's representation theorem to lift {the joint} measures to a common probability space $(\breve\Omega,\breve{\c F},\breve\P)$ with processes
    \begin{equation}\label{eq theorem 3.6 convergence breve tuples}
    (\breve X^n,\breve \xi^n,\breve W^n,\breve{\hat X}^n,\breve{\hat\xi}^n,\breve{\bar X}^n,\breve{\bar\xi}^n,\breve{\bar s}^n) \to (\breve X,\breve \xi,\breve W,\breve{\hat X},\breve{\hat\xi},\breve{\bar X},\breve{\bar\xi},\breve{\bar s})\quad\breve\P\text{-a.s. and in }L^2,
    \end{equation}
    {defined on this space,
    such that
    \[
    \breve\P_{(\breve{X}^n, \breve{\xi}^n, \breve{W}^n, \breve{\hat{X}}^n, \breve{\hat{\xi}}^n, \breve{\bar{X}}^n, \breve{\bar{\xi}}^n, \breve{\bar{s}}^n)} = \check\P_{(X^n, \xi^n, W^n, \hat{X}^n, \hat{\xi}^n, \bar{X}^n, \bar{\xi}^n, \bar{s}^n)} \quad\text{and}\quad \breve\P_{(\breve{X}, \breve{\xi}, \breve{W}, \breve{\hat{X}}, \breve{\hat{\xi}}, \breve{\bar{X}}, \breve{\bar{\xi}}, \breve{\bar{s}})} = \check\P^\infty.
    \]
    Crucially, this preserves the relationships within each tuple $(\breve{X}^n, \breve{\xi}^n, \breve{W}^n, \breve{\hat{X}}^n, \breve{\hat{\xi}}^n, \breve{\bar{X}}^n, \breve{\bar{\xi}}^n, \breve{\bar{s}}^n)$ because we are applying Skorokhod's representation theorem to the joint distributions. In particular, the relationships between the layers - how $(\breve{\bar{X}}^n, \breve{\bar{\xi}}^n, \breve{\bar{s}}^n)$ is constructed from $(\breve{\hat{X}}^n, \breve{\hat{\xi}}^n)$, and how those are related to $(\breve{X}^n, \breve{\xi}^n, \breve{W}^n)$ - are maintained on the new probability space.
    }
    
    As usual, we denote the right-continuous inverse functions of $\breve{\bar s}^n$ and $\breve{\bar s}$ by $\breve s^n$ and $\breve s$.
    Now we can use that the processes $(\breve{\bar X}^n,\breve{\bar\xi}^n,\breve W^n,\breve{\bar s}^n,\hat r^n)$ satisfy \eqref{eq weak solution two-layer time changed SDE} together with \cite[Theorem 7.10]{Kurtz1996}, to conclude that $(\breve{\bar X},\breve{\bar\xi},\breve W,\breve{\bar s},\hat r)$ also satisfies \eqref{eq weak solution two-layer time changed SDE}.
    Further, since $\breve W^n$ are $\bar{\b F}^n$-Brownian motions, where $\bar{\c F}^n_u \coloneqq \c F^{\breve W^n}_u\lor \c F^{\breve{\bar X}^n,\breve{\bar\xi}^n}_{\breve s^n_{r^n_u}}$, the limit $\breve W$ is also an $\bar{\b F}$-Brownian motion, where $\bar{\c F}_u \coloneqq \c F^{\breve W}_u\lor \c F^{\breve{\bar X},\breve{\bar\xi}}_{\breve s_{r_u}}$.
    Therefore $(\breve\P_{(\breve{\bar X},\breve{\bar\xi},\breve W,\breve{\bar s})},\hat r)$ is a weak solution to the time changed SDE \eqref{eq weak solution two-layer time changed SDE}.
    
    This argument also shows that, since $(\breve\P_{(\breve{\hat X}^n,\breve{\hat\xi}^n,\breve W^n)},\hat r^n)$ are weak solutions to \eqref{eq weak solution W2 time changed SDE}, the limit $(\hat\P,\hat r)$ is also a weak solution to \eqref{eq weak solution W2 time changed SDE}.
    
    To establish the consistency between the previously constructed layers we note that for $\breve\P$-almost all $\breve\omega\in\breve\Omega$, $(\breve{\bar X}^n(\breve\omega),\breve{\bar\xi}^n(\breve\omega),\hat r^n_{\breve{\bar s}^n(\breve\omega)})$ is a $WM_1$ parametrisation of $(\breve X^n(\breve\omega),\breve\xi^n(\breve\omega))$ and $(\breve{\bar X}(\breve\omega),\breve{\bar\xi}(\breve\omega),\hat r_{\breve{\bar s}(\breve\omega)})$ is a $WM_1$ parametrisation of $(\breve{\bar X}_{\breve s_r(\breve\omega)}(\breve\omega),\breve{\bar\xi}_{\breve s_r(\breve\omega)}(\breve\omega))$.
    This implies that
    \[
    (\breve X^n,\breve\xi^n)\to (\breve{\bar X}_{\breve s_r},\breve{\bar\xi}_{\breve s_r})\text{ in }D^0,\qquad\breve\P\text{-a.s.}.
    \]
    However, since we know by \eqref{eq theorem 3.6 convergence breve tuples} that also $(\breve X^n,\breve\xi^n) \to (\breve X,\breve\xi)$ in $D^0$, $\breve\P$-a.s. and in $L^2$, we deduce that
    \[
    (\breve{\bar X}_{\breve s_{r_u}},\breve{\bar\xi}_{\breve s_{r_u}}) = (\breve X_u,\breve\xi_u),\qquad\text{for all }u\in [t-,T],\qquad\breve\P\text{-a.s.}.
    \]
    Using the same argument, we can also show that
    \[
    (\breve{\bar X}_{\breve s_u},\breve{\bar\xi}_{\breve s_u}) = (\breve{\hat X}_u,\breve{\hat \xi}_u),\qquad\text{for all }u\in [0-,1],\qquad\breve\P\text{-a.s.},
    \]
    which proves that the layers are indeed consistent, since
    \[
    \P = \bar\P_{(\bar X_{s_{r_u}},\bar\xi_{s_{r_u}})_{u\in [t-,T]}} = \hat\P_{(\hat X_{r_u},\hat\xi_{r_u})_{u\in [t-,T]}},\qquad \hat\P = \bar\P_{(\bar X_{s_u},\bar\xi_{s_u})_{u\in [0-,1]}}.
    \]
    Hence $((\hat m,\hat r),\bar\P)$ is indeed a two-layer parametrisation of $\P$.
\end{enumerate}

\subsubsection{{Proof of Theorem \ref{theorem two layer parametrisation bounded velocity approximation}}}\label{appendix proof of theorem two layer parametrisation bounded velocity approximation}
The result will be a consequence of the following two Lemmas \ref{lemma two layer parametrisation lipschitz approximation} and \ref{lemma lipschitz two layer parametrisation bounded velocity approximation}.

\begin{lemma}\label{lemma two layer parametrisation lipschitz approximation}
    For every two-layer parametrisation $((\hat m,\hat r),\bar\P)$ of a singular control $\P\in\c P(t,m)$ we can find a sequence of singular controls $(\P^n)_n\subseteq \c P(t,m)$ with Lipschitz two-layer parametrisations $((\hat m^n,\hat r^n),\bar\P^n)$ such that
    \[
    ((\hat m^n,\hat r^n),\bar\P^n) \to ((\hat m,\hat r),\bar\P).
    \]
\end{lemma}

\begin{proof}
    We start with a given two-layer parametrisation $((\hat m,\hat r),\bar\P)$ of a singular control $\P\in\c P(t,m)$. By Definitions \ref{definition two-layer SDE parametrisation} and \ref{definition weak solution two-layer time changed SDE}, there exists a corresponding weak solution $(\bar\P,\hat r) \in \c P_2(\bar\Omega)\times C([0,1];[t,T])$ to the time changed SDE \eqref{eq weak solution two-layer time changed SDE}, so
    \[
    d\bar X_u = b(\hat r_{\bar s_u},\hat m_{\bar s_u},\bar X_u,\bar\xi_u) d\hat r_{\bar s_u} + \sigma(\hat r_{\bar s_u},\hat m_{\bar s_u},\bar X_u,\bar\xi_u) d\bar W_{\hat r_{\bar s_u}} + \gamma(\hat r_{\bar s_u}) d\bar\xi_u,
    \]
    where we denote the canonical process on $\bar\Omega$ by $(\bar X,\bar\xi,\bar W,\bar s)$. Furthermore,  we denote the right-continuous inverse functions of $\hat r$ and $\bar s$ by $r$ and $s$. As discussed in Section \ref{subsubsection recover the previous layers from bar P} we recover $(\tilde X,\tilde\xi)$ and $(\hat X,\hat\xi)$ on $(\bar\Omega,\c B(\bar\Omega),\bar\P)$ and note that the same Brownian motion $\bar W$ is also an $\b F^{\tilde X,\tilde\xi,\bar W}$- and an $\hat{\b F}$-Brownian motion and $(\tilde X,\tilde\xi,\bar W)$ and $(\hat X,\hat\xi,\bar W,\hat r)$ satisfy the SDEs \eqref{eq diffusion x} and \eqref{eq weak solution W2 time changed SDE}:
    \begin{equation}
    \begin{split}
        d\tilde X_u &= b(u,m_u,\tilde X_u,\tilde\xi_u) du + \sigma(u,m_u,\tilde X_u,\tilde\xi_u) d\bar W_u + \gamma(u) d\tilde\xi_u,\\
        d\hat X_u &= b(\hat r_u,\hat m_u,\hat X_u,\hat\xi_u) d\hat r_u + \sigma(\hat r_u,\hat m_u,\hat X_u,\hat\xi_u) d\bar W_{\hat r_u} + \gamma(\hat r_u) d\hat\xi_u,
    \end{split}
    \end{equation}
    where $\hat m_u = \bar\P_{(\bar X_{s_u},\bar\xi_{s_u})} = \bar\P_{(\hat X_u,\hat\xi_u)}$ and $m_u = \bar\P_{(\tilde X_u,\tilde\xi_u)} = \hat m_{r_u}$.
    
    We construct the approximating parametrisations in multiple steps. First we truncate our parametrisation to obtain parametrisations with uniformly bounded $\xi$-components. This allows us to obtain approximating Lipschitz two-layer parametrisations by ``suitably'' perturbing $\hat r$ and $\bar s$. 
    
    \begin{enumerate}[wide]
        \item[\textbf{ Step 1. \emph{Truncating.}}]
        Let us fix $N\in\N$. We introduce the truncated controls
        \[
            \tilde\xi_u^N \coloneqq \tilde\xi_{t-} + (\tilde\xi_u - \tilde\xi_{t-}) \land N,\  \hat \xi_u^N \coloneqq \hat \xi_{0-} + (\hat\xi_u - \hat\xi_{0-}) \land N,\  \bar\xi_u^N \coloneqq \bar\xi_0 + (\bar\xi_u - \bar\xi_0)\land N,
        \]
        where $\xi\land N$ should be understood component-wise. The corresponding state processes are given by
        \begin{equation}
        \begin{split}
            d\tilde X^N_u &= b(u,m^N_u,\tilde X^N_u,\tilde\xi^N_u) du + \sigma(u,m^N_u,\tilde X^N_u,\tilde\xi^N_u) d\bar W_u + \gamma(u) d\tilde\xi^N_u,\  \tilde X^N_{t-} = \tilde X_{t-},\\
            d\hat X^N_u &= b(\hat r_u, \hat m^N_u, \hat X^N_u, \hat \xi^N_u) d\hat r_u + \sigma(\hat r_u, \hat m^N_u, \hat X^N_u, \hat \xi^N_u) d\bar W_{\hat r_u} + \gamma(\hat r_u) d\hat\xi^N_u,\  \hat X^N_{0-} = \hat X_{0-},\\
            d\bar X^N_u &= b(\hat r_{\bar s_u}, \hat m^N_{\bar s_u}, \bar X^N_u, \bar \xi^N_u) d\hat r_{\bar s_u} + \sigma(\hat r_{\bar s_u}, \hat m^N_{\bar s_u}, \bar X^N_u, \bar \xi^N_u) d\bar W_{\hat r_{\bar s_u}} + \gamma(\hat r_{\bar s_u}) d\bar\xi^N_u,\ \bar X^N_0 = \bar X_0,\!\!\!
        \end{split}
        \end{equation}
        where $m^N_u \coloneqq \bar\P_{(\tilde X^N_u,\tilde\xi^N_u)}$ and $\hat m^N_u \coloneqq \bar\P_{(\hat X^N_u,\hat \xi^N_u)}$.
        We see that $((\hat m^N,\hat r),\bar\P_{(\bar X^N,\bar\xi^N,\bar s)})$ is a two-layer parametrisation of the truncated control $\bar\P_{(X^N,\xi^N)}\in \c P(t,m)$. Furthermore, {since the processes $(\hat X^N,\hat\xi^N,\bar X^N,\bar\xi^N)$ and $(\hat X,\hat\xi,\bar X,\bar\xi)$ are defined on the same probability space $(\bar\Omega,\c B(\bar\Omega),\bar\P)$, we can follow} standard Gronwall-type arguments {which} imply that
        \[
        ((\hat m^N,\hat r),\bar\P_{(\bar X^N,\bar\xi^N,\bar s)}) \to ((\hat m,\hat r),\bar\P_{(\bar X,\bar\xi,\bar s)}) = ((\hat m,\hat r),\bar\P).
        \]
        
        \item[\textbf{ Step 2. \emph{The first layer.}}]
        Our next goal is for each $N\in\N$ to reparametrise the above processes $(\hat X^N,\hat\xi^N,\hat r)$ and $(\bar X^N,\bar \xi^N,\bar s)$ corresponding to the two-layer parametrisation 
        \[
        ((\hat m^N,\hat r),\bar\P_{(\bar X^N,\bar\xi^N,\bar s)} )
       \]        
        by introducing small perturbations into  $\hat r$ and $\bar s$ so that the resulting parametrisations are Lipschitz continuous. We start with the first layer $(\hat X^N,\hat \xi^N,\hat r)$ and define for each $\delta > 0$ the process
        \[
        \alpha^{N,\delta}_u \coloneqq \frac{u + \delta((\hat r_u-t)+ \norm{\hat\xi^N_u - \hat\xi^N_{0-}}_{L^2(\bar\P)})}{1 + \delta((T-t) + \norm{\hat\xi^N_1 - \hat\xi^N_{0-}}_{L^2(\bar\P)})},\qquad\text{for all }u\in [0-,1],
        \]
        and let $\hat \alpha^{N,\delta}_u$ be its inverse. Then we define
        \[
        \hat X^{N,\delta}_u \coloneqq \hat X^N_{\hat\alpha^{N,\delta}_u},\qquad \hat \xi^{N,\delta}_u \coloneqq \hat \xi^N_{\hat\alpha^{N,\delta}_u},\qquad \hat r^{N,\delta}_u \coloneqq \hat r_{\hat\alpha^{N,\delta}_u},\qquad\text{for }u\in [0-,1].
        \]
        \
        The process $\hat r^{N,\delta}$ is Lipschitz, due to the monotonicity of $\hat r$ and $\hat\xi^N$ that yields for all $u,v\in [0-,1]$ that
        \begin{equation}
        \begin{split}
        |\hat r^{N,\delta}_u - \hat r^{N,\delta}_v| = |\hat r_{\hat \alpha^{N,\delta}_u}-\hat r_{\hat \alpha^{N,\delta}_v}|
        &\leq \frac {1 + \delta((T-t) + \norm{\hat \xi^N_1-\hat \xi^N_{0-}}_{L^2(\bar\P)})}\delta |u - v| \\
        & \leq \frac{1 + \delta ((T-t) + lN)}\delta |u-v|.
        \end{split}
        \end{equation}
        At the same time we have for all $u\in [0-,1]$ that
        \begin{equation}\label{eq lemma a.1 hat alpha estimate}
        \begin{split}
        |\hat\alpha^{N,\delta}_u - u| &=  |\hat\alpha^{N,\delta}_u - \alpha^{N,\delta}_{\hat\alpha^{N,\delta}_u}| \\
        & =  \abs*{\hat\alpha^{N,\delta}_u - \frac{\hat\alpha^{N,\delta}_u + \delta((\hat r_{\hat\alpha^{N,\delta}_u}-t) + \norm{\hat\xi^N_{\hat\alpha^{N,\delta}_u} - \hat\xi^N_{0-}}_{L^2(\bar\P)})}{1+\delta((T-t)+\norm{\hat\xi^N_1 - \hat\xi^N_{0-}}_{L^2(\bar\P)})}}\\
        &= \abs*{\frac{\delta\bigl(\hat\alpha^{N,\delta}_u (T-t) - (\hat r_{\hat\alpha^{N,\delta}_u}-t) + \alpha^{N,\delta}_u \norm{\hat\xi^N_1 - \hat\xi^N_{0-}}_{L^2(\bar\P)} - \norm{\hat\xi^N_{\hat\alpha^{N,\delta}_u} - \hat\xi^N_{0-}}_{L^2(\bar\P)}\bigr)}{1+\delta((T-t)+\norm{\hat\xi^N_1 - \hat\xi^N_{0-}}_{L^2(\bar\P)})}}\\
        &\leq 2\delta ((T-t) + lN).
        \end{split}
        \end{equation}
        Since $[0,1]\ni u \mapsto (\hat m_u,\hat r_u) \in \c P_2\times [t,T]$ is uniformly continuous, the reparametrised measure flow $\hat m^{N,\delta}_u \coloneqq \bar\P_{(\hat X^{N,\delta}_u,\hat\xi^{N,\delta}_u)} = \hat m^N_{\hat\alpha^{N,\delta}_u}$ on $\c P_2$ satisfies
        \[
        (\hat m^{N,\delta},\hat r^{N,\delta}) \to (\hat m^N,\hat r)\qquad\text{in }C([0,1];\c P_2\times\R^d\times\R^l).
        \]
        
        Finally, by construction $\bar W$ is also an $\hat{\b F}^{N,\delta}$-Brownian motion with $\hat{\c F}^{N,\delta}_u \coloneqq \c F^{\bar W}_u\lor \c F^{\hat X^{N,\delta},\hat\xi^{N,\delta}}_{r^{N,\delta}_u}$ and thus $(\bar\P_{(\hat X^{N,\delta},\hat\xi^{N,\delta})},\hat r)$ are Lipschitz $\c P_2$-parametrisations of $\bar\P_{(\tilde X^N,\tilde\xi^N)}$.
        
        \item[\textbf{ Step 3. \emph{The second layer.}}]
        Now we turn to the second layer. We note that 
        \[
        ( (\hat m^{N,\delta},\hat r^{N,\delta}),\bar\P_{(\bar X^N_{\alpha^{N,\delta}},\bar\xi^N_{\alpha^{N,\delta}},\bar s_{\alpha^{N,\delta}})} )
        \]
        is a two-layer parametrisation of $\bar\P_{(\tilde X^N,\tilde\xi^N)}$. However, it is not necessarily a Lipschitz parametrisation yet. Thus we have to reparametrise the second layer as well. As above, we define for $\delta > 0$ the process 
        \[
        \beta^{N,\delta}_u(\omega) \coloneqq \frac{u + \delta(\bar s_{\alpha^{N,\delta}_u}(\omega) + \Var(\bar\xi^N_{\alpha^{N,\delta}}(\omega),[0,u]))}{1 + \delta(1 + lN)},\qquad\text{for all }u\in [0,1],
        \]
        where $\Var$ denotes the total variation, and let $\bar\beta^{N,\delta}_u(\omega) \coloneqq \inf\{v \in [0,1] \mid \beta^{N,\delta}_v(\omega) > u \} \land 1$ be its right-continuous inverse. Then we define
        \[
        \bar X^{N,\delta}_u(\omega) \coloneqq \bar X^N_{\alpha^{N,\delta}_{\bar\beta^{N,\delta}_u(\omega)}}(\omega),\qquad 
        \bar \xi^{N,\delta}_u(\omega) \coloneqq \bar\xi^N_{\alpha^{N,\delta}_{\bar\beta^{N,\delta}_u(\omega)}}(\omega),\qquad 
        \bar s^{N,\delta}_u(\omega) \coloneqq \bar s_{\alpha^{N,\delta}_{\bar\beta^{N,\delta}_u(\omega)}}(\omega).
        \]
        Just as before, $(\bar\xi^{N,\delta},\bar s^{N,\delta})$ is Lipschitz uniformly in $\omega$ since for all $u,v\in [0,1]$,
        \begin{equation}
        \begin{split}
        &|\bar\xi^{N,\delta}_u(\omega) - \bar\xi^{N,\delta}_v(\omega)| + |\bar s^{N,\delta}_u(\omega) - \bar s^{N,\delta}_v(\omega)|\\
        &= |\bar\xi^N_{\alpha^{N,\delta}_{\bar\beta^{N,\delta}_u(\omega)}}(\omega) - \bar\xi^N_{\alpha^{N,\delta}_{\bar\beta^{N,\delta}_v(\omega)}}(\omega)| + |\bar s_{\alpha^{N,\delta}_{\bar\beta^{N,\delta}_u(\omega)}}(\omega) - \bar s_{\alpha^{N,\delta}_{\bar\beta^{N,\delta}_v(\omega)}}(\omega)| 
         \leq \frac {1 + \delta(1 + lN)} \delta |u-v|.
         \end{split}
        \end{equation}
        Next, we consider the process $\bar\beta^{N,\delta}$; by definition the process is not necessarily surjective and thus not necessarily invertible on $[0,1]$. Instead, we note that $\beta^{N,\delta}_{\bar\beta^{N,\delta}_u}  = u\land \beta_1^{N,\delta}$, which results in an extra term in the following calculation. For all $u\in [0,1]$,
        \begin{equation}
        \begin{split}
        &|\bar\beta^{N,\delta}_u(\omega) - u|\\
        &\leq |\bar\beta^{N,\delta}_u(\omega) - \beta^{N,\delta}_{\bar\beta^{N,\delta}_u(\omega)}(\omega)| + |\beta^{N,\delta}_{\bar\beta^{N,\delta}_u(\omega)}(\omega) - u|\\
        &\leq \bigg|\bar\beta^{N,\delta}_u(\omega) - \frac{\bar\beta^{N,\delta}_u(\omega) + \delta\bigl(\bar s_{\alpha^{N,\delta}_{\bar\beta^{N,\delta}_u(\omega)}}(\omega) + \Var(\bar\xi^N_{\alpha^{N,\delta}}(\omega),[0,\bar\beta^{N,\delta}_u(\omega)])\bigr)}{1 + \delta(1 + lN)}\bigg| + |\beta^{N,\delta}_1(\omega) - 1|\!\!\\
        &\leq 2 \delta (1 + lN) + \frac{\delta |lN - \Var(\bar\xi^N_{\alpha^{N,\delta}}(\omega),[0,1])|}{1+\delta(1 + lN)}
        \leq 3 \delta (1 + lN).
        \end{split}
        \end{equation}
        In particular, the $\bar\P$-almost sure continuity of $(\bar X^N,\bar \xi^N,\bar s)$ implies that
        \[
        (\bar X^{N,\delta},\bar\xi^{N,\delta},\bar s^{N,\delta}) \to (\bar X^N,\bar \xi^N,\bar s)\qquad\text{in }C([0,1];\R^d\times\R^l\times [t,T]),\qquad\bar\P\text{-a.s.},
        \]
        which, together with the $\c P_2$-tightness of $(\bar\P_{(\bar X^{N,\delta},\bar\xi^{N,\delta},\bar s^{N,\delta})})_\delta$ - by its construction as a collection of reparamerisations of $\bar\P_{(\bar X^N,\bar\xi^N,\bar s)}$ - implies that
        \[
        \bar\P_{(\bar X^{N,\delta},\bar\xi^{N,\delta},\bar s^{N,\delta})} \to \bar\P_{(\bar X^N,\bar \xi^N,\bar s)}\qquad\text{in }\c P_2(C([0,1];\R^d\times\R^l\times [0,1])).
        \]
        
        Furthermore, by construction $\beta^{N,\delta}$ and $\bar\beta^{N,\delta}$ are $\b F^{\bar \xi^N_{\alpha^{N,\delta}},\bar s_{\alpha^{N,\delta}}}$-adapted and thus $\bar W$ is also an $\bar{\b F}^{N,\delta}$-Brownian motion with $\bar {\c F}^{N,\delta}_u\coloneqq \c F^{\bar W}_u\lor \c F^{\bar X^{N,\delta},\bar\xi^{N,\delta},\bar s^{N,\delta}}_{s^{N,\delta}_{r^{N,\delta}_u}}$, where $r^{N,\delta}$ and $s^{N,\delta}$ denote as usual the right-continuous inverses of $\hat r^{N,\delta}$ and $\bar s^{N,\delta}$. As a result, we see that $((\hat m^{N,\delta},\hat r^{N,\delta}),\bar\P_{(\bar X^{N,\delta},\bar\xi^{N,\delta},\bar s^{N,\delta})})_{N,\delta}$ are Lipschitz two-layer parametrisations and, along a suitable subsequence with $\delta\to 0$ and $N\to\infty$,
        \[
        ((\hat m^{N,\delta},\hat r^{N,\delta}),\bar\P_{(\bar X^{N,\delta},\bar\xi^{N,\delta},\bar s^{N,\delta})} ) \to ( (\hat m,\hat r),\bar\P_{(\bar X,\bar\xi,\bar s)}) = ((\hat m,\hat r),\bar\P ).
        \]
    \end{enumerate}
\end{proof}

\begin{lemma}\label{lemma lipschitz two layer parametrisation bounded velocity approximation}
    For every Lipschitz two-layer parametrisation $((\hat m,\hat r),\bar\P)$ of a singular control $\P\in\c P(t,m)$ we can find sequence of continuous bounded velocity controls $(\P^n)_n\subseteq \c L(t,m)$ with two-layer parametrisations $((\hat m^n,\hat r^n),\bar\P^n)$ such that
    \[
    ( (\hat m^n,\hat r^n),\bar\P^n) \to ((\hat m,\hat r),\bar\P ).
    \]
\end{lemma}

\begin{proof}
    Again we consider a weak solution $(\bar\P,\hat r)\in \c P_2(\bar\Omega)\times C([0,1];[t,T])$ corresponding to the two-layer parametrisation $((\hat m,\hat r),\bar\P)$. As outlined in Section \ref{subsubsection recover the previous layers from bar P}, we denote the canonical process on $\bar\Omega$ by $(\bar X,\bar\xi,\bar W,\bar s)$ and further recover $(\hat X,\hat\xi)$ and $(\tilde X,\tilde\xi)$ on $(\bar\Omega,\c B(\bar\Omega),\bar\P)$ such that $(\bar\P_{(\hat X,\hat\xi,\bar W)},\hat r)$ is a weak solution to the time changed SDE \eqref{eq weak solution W2 time changed SDE} and $(\bar\P_{(\tilde X,\tilde\xi,\bar W)})$ is a weak solution to \eqref{eq diffusion x}. As usual we will denote the right-continuous inverse functions of $\hat r$ and $\bar s$ by $r$ and $s$.
    
    For Lipschitz two-layer parametrisations, if $r$ and $s$ are both Lipschitz, then $\tilde\xi = \bar\xi_{s_r}$ is Lipschitz as well and thus $\bar\P_{(\tilde X,\tilde\xi)}\in \c L(t,m)$. In general $r$ and $s$ will not be Lipschitz; the previous observation provides nonetheless the main idea for our proof. Specifically, we show that we can construct 
    Lipschitz approximations by slightly perturbing $\hat r$ and $\bar s$ such that $r$ and $s$ become Lipschitz continuous.
   
   \begin{enumerate}[wide]
       \item[\textbf{ Step 1. \emph{The bounded velocity control.}}]
    We start mollifying the time scale $\hat r$ using a standard mollifier to obtain smooth, monotone functions $\hat r^\varepsilon \in C^\infty([0,1];[t,T])$ such that\footnote{{%
        To construct $\hat r^\varepsilon$, we first extend $\hat r$ to $[-1,2]$ by reflection,
        \[
            \hat r_u \coloneqq 
            \begin{cases} 
                2t - \hat r_{-u}, & u < 0, \\
                2T - \hat r_{2-u}, & u > 1,
            \end{cases}
        \]
        ensuring boundary conditions are preserved after mollification.
        Next we define $\hat r^\varepsilon \coloneqq \eta_\varepsilon * \hat r$, where $\eta_\varepsilon$ is the standard mollifier and $*$ denotes the convolution between two functions. 
        By \cite[Appendix C, Theorem 7]{evans_partial_2010}, $\hat r^\varepsilon$ and $\dot{\hat r}^\varepsilon$ inherit smoothness, monotonicity, and convergence properties, noting that $\dot{\hat r}^\varepsilon = 
        \eta_\varepsilon * \dot{\hat r}$. 
        The Lipschitz constant is preserved since $\norm{\dot{\hat r}^\varepsilon}_\infty = \norm{\eta_\varepsilon * \dot{\hat r}}_\infty \leq \norm{\dot{\hat r}}_\infty$.
    }}
    \begin{enumerate}[label=(\roman*)]
        \item $\hat r^\varepsilon_0 = \hat r_0 = t$ and $\hat r^\varepsilon_1 = \hat r_1 = T$,
        \item $\hat r^\varepsilon\to \hat r$ in $C([0,1];[t,T])$,
        \item $\dot{\hat r}^\varepsilon \to \dot{\hat r}$ a.e. and in $L^1([0,1];[t,T])$,
        \item $\hat r^\varepsilon$ is Lipschitz continuous with the same Lipschitz constant as $\hat r$ and in particular $\dot{\hat r}^\varepsilon$ is uniformly bounded in $\varepsilon$,
    \end{enumerate}
    where $\dot{\hat r}$ denotes the almost sure derivative of $\hat r$. This mollification {ensures} that not only $\hat r$ and $\bar s$ but also $\hat r_{\bar s}$ {are} approximated sufficiently well, see the later result \eqref{eq lemma a.2 convergence dot bar r epsilon,delta dot bar r}. Now we define for $\delta > 0$ process
    \begin{equation}\label{eq lemma a.2 definition r delta}
        \hat r^{\varepsilon,\delta}_u \coloneqq t + \frac{(\hat r^\varepsilon_u-t) + \delta u}{(T-t) + \delta} (T-t),\qquad\text{for }u\in [0,1],
    \end{equation}
    and denote its inverse by $r^{\varepsilon,\delta}$. For $\bar s(\omega)$ we skip the mollification and directly introduce the process
    \begin{equation}\label{eq lemma a.2 definition s delta}
        \bar s^\delta_u(\omega) \coloneqq \frac{\bar s_u(\omega)  + \delta u}{1 + \delta},\qquad\text{for }u\in [0,1],
    \end{equation}
    along with its inverse $s^\delta(\omega)$. Since $\hat r,\bar s(\omega)$ and $\bar \xi(\omega)$ are Lipschitz continuous independently of $\omega$, so are the processes $\hat r^{\varepsilon,\delta}, r^{\varepsilon,\delta},\bar s^\delta(\omega)$ and $s^\delta(\omega)$ as well. We then define our bounded velocity controls as
    \begin{equation}\label{eq lemma a.2 definition xi epsilon,delta}
        \tilde \xi^{\varepsilon,\delta}_u(\omega) \coloneqq \bar\xi_{s^{\delta}_{r^{\varepsilon,\delta}_u}(\omega)}(\omega)\qquad\text{for }u\in [t,T].
    \end{equation}

    \item[\textbf{ Step 2. \emph{The corresponding state process.}}]
    We now construct a state process corresponding to the previously introduced bounded velocity control.  To this end, we construct a specific Brownian motion that enables us to obtain the desired convergence of the resulting two-layer parametrisations. We start by recalling that $(\bar\P_{(\hat X,\hat\xi,\bar W)},\hat r)$ is a weak solution to \eqref{eq weak solution W2 time changed SDE}.
    Starting with $\bar W$ on $(\bar\Omega,\bar{\c F},\bar\P)$ we can construct a new Brownian motion $\hat B$ on some enlarged probability space $(\mathring\Omega,\mathring{\c F},\mathring\P)$ such that
    \[
    d\bar W_{\hat r_u} = \sqrt{\dot{\hat r}_u} d\hat B_u\qquad\text{for }u\in[0,1].
    \]
    We achieve this by sufficiently enlarging the probability space to allow for an independent Brownian motion $\mathring W$ and then defining
    \[
    \hat B_v \coloneqq \int_0^v \1_{\{\dot {\hat  r}_u \not= 0\}} \frac 1 {\sqrt{\dot {\hat  r}_u}} d\bar W_{\hat r_u} + \int_0^v \1_{\{\dot{\hat r}_u = 0\}} d\mathring W_u,\qquad\text{for }v\in [0,1].
    \]
    Then $\hat B$ is an $\b F^{\hat B,\hat X,\hat\xi}$-Brownian motion on $[0,1]$ since $\<\hat B\>_u = u$ and $(\hat X,\hat \xi,\hat B,\hat r)$ satisfies the equation
    \begin{equation}\label{eq lemma a.2 sde hat X}
        d\hat X_u = b(\hat r_u,\hat m_u,\hat X_u,\hat \xi_u) d\hat r_u + \sigma(\hat r_u,\hat m_u,\hat X_u,\hat \xi_u) \sqrt{\dot{\hat r}_u} d\hat B_u + \gamma(\hat r_u) d\hat \xi_u,\qquad \hat X_0 = \tilde X_{t-}.
    \end{equation}

    Similarly, starting with $\hat B$, we subsequently construct an $\b F^{\bar B,\bar X,\bar\xi,\bar s}$-Brownian motion $\bar B$ on some enlarged probability space $(\breve\Omega,\breve{\c F},\breve\P)$ using an independent Brownian motion $\breve W$ such that
    \[
    d\hat B_{\bar s_u(\omega)} = \sqrt{\dot{\bar s}_u(\omega)} d\bar B_u\qquad\text{for }u\in [0,1],
    \]
    and such that
    \begin{equation}\label{eq lemma a.2 definition bar W}
        d\bar B_u = \1_{\{\dot{\bar s}_u(\omega) \not= 0\}} \frac 1 {\sqrt{\dot{\bar s}_u(\omega)}} d\hat B_{\bar s_u(\omega)} + \1_{\{\dot{\bar s}_u(\omega) = 0\}} d\breve W_u,\qquad\text{for }u\in [0,1].
    \end{equation}
    so that $(\bar X,\bar \xi,\bar B,\bar s,\hat r)$ satisfies the equation 
    \begin{equation}
    \label{eq lemma a.2 sde bar X}
    \begin{split}
        d \bar X_u &= b(\hat r_{\bar s_u(\omega)},\hat m_{\bar s_u(\omega)},\bar X_u,\bar \xi_u) d\hat r_{\bar s_u(\omega)}\\*
        &\qquad+ \sigma(\hat r_{\bar s_u(\omega)},\hat m_{\bar s_u(\omega)},\bar X_u,\bar \xi_u) \sqrt{\dot{\hat r}_{\bar s_u(\omega)}\dot{\bar s}_u(\omega)} d\bar B_u + \gamma(\hat r_{\bar s_u(\omega)}) d\bar\xi_u,\qquad \bar X_0 = \tilde X_{t-}.
    \end{split}
    \end{equation}
    
    With this we are now ready to construct our Brownian motion for $\tilde X^{\varepsilon,\delta}$. We first define $\hat B^{\delta}$ such that
    \begin{equation}\label{eq lemma a.2 definition hat W,delta}
        d\hat B^{\delta}_{\bar s^{\delta}_u(\omega)} = \sqrt{\dot{\bar s}^{\delta}_u(\omega)} d\bar B_u,\qquad\text{for }u\in [0,1],
    \end{equation}
    by letting
    \[
    \hat B^{\delta}_v \coloneqq \int_0^v \sqrt{\dot{\bar s}^{\delta}_{s^{\delta}_u(\omega)}(\omega)} d\bar B_{s^{\delta}_u(\omega)},\qquad\text{for }v\in [0,1].
    \]
    Since $\bar s^\delta$ is $\b F^{\bar s}$-adapted, $\hat B^\delta$ is an $(\c F^{\hat B^\delta}_u \lor \c F^{\bar X,\bar\xi,\bar s}_{s^\delta_u})_u$-Brownian motion and thus also an $(\c F^{\hat B^\delta}_u\lor \c F^{\bar X,\bar\xi,\bar s^\delta}_{s^\delta_u})_u$-Brownian motion.
    In the same fashion we define an $\bar{\b F}^{\varepsilon,\delta}$-Brownian motion $\bar W^{\varepsilon,\delta}$, where $\bar{\c F}^{\varepsilon,\delta}_u \coloneqq \c F^{\bar W^{\varepsilon,\delta}}_u\lor \c F_{s^\delta_{r^{\varepsilon,\delta}_u}}^{\bar X,\bar\xi,\bar s^{\delta}}$, such that
    \[
    d\bar W^{\varepsilon,\delta}_{\hat r^{\varepsilon,\delta}_u} = \sqrt{\dot{\hat r}^{\varepsilon,\delta}_u} d\hat B^{\delta}_u,\qquad\text{for }u\in [0,1].
    \]
    Armed with the above results we can now introduce the state process corresponding to our bounded velocity control as
    \[
        d \tilde X^{\varepsilon,\delta}_u \coloneqq b(u,m^{\varepsilon,\delta}_u,\tilde X^{\varepsilon,\delta}_u,\tilde \xi^{\varepsilon,\delta}_u) du
        + \sigma(u,m^{\varepsilon,\delta}_u,\tilde X^{\varepsilon,\delta}_u,\tilde \xi^{\varepsilon,\delta}_u) d\bar W^{\varepsilon,\delta}_u + \gamma(u) d\tilde \xi^{\varepsilon,\delta}_u,\ \tilde X^{\varepsilon,\delta}_{t-} = \tilde X_{t-}.
    \]
    {Then $\breve\P_{(\tilde X^{\varepsilon,\delta},\tilde\xi^{\varepsilon,\delta},\tilde W^{\varepsilon,\delta})}$ is a weak solution to \eqref{eq diffusion x} in the sense of Definition \ref{definition weak solution sde}.}
    
    \item[\textbf{ Step 3. \emph{Weak solutions and candidate parametrisations.}}]
    We are now going to show how to obtain weak solutions to the time-changed SDEs \eqref{eq weak solution W2 time changed SDE} and  \eqref{eq weak solution two-layer time changed SDE}  and hence candidate parametrisations from the above state process. 
    
    We first notice that  $\bar W^{\varepsilon,\delta}$ is an $\b F^{\bar W^{\varepsilon,\delta},\tilde X^{\varepsilon,\delta},\tilde \xi^{\varepsilon,\delta}}$-Brownian motion and thus $\breve\P_{(\tilde X^{\varepsilon,\delta},\tilde\xi^{\varepsilon,\delta})}\in \c P(t,m)$.
    We then define the state processes
    \[
    \hat X^{\varepsilon,\delta}_u \coloneqq \tilde X^{\varepsilon,\delta}_{\hat r^{\varepsilon,\delta}_u},\qquad \hat\xi^{\delta}_u \coloneqq \tilde \xi^{\varepsilon,\delta}_{\hat r^{\varepsilon,\delta}_u} = \bar \xi_{s^{\delta}_u(\omega)},\qquad\text{for }u\in [0-,1].
    \]
    In terms of the notation $\hat m^{\varepsilon,\delta}_u \coloneqq \breve\P_{(\hat X^{\varepsilon,\delta}_u,\hat\xi^{\varepsilon,\delta}_u)}$, it follows that $\big(\hat X^{\varepsilon,\delta},\hat \xi^{\delta},\hat B^\delta,\hat r^{\varepsilon,\delta}\big)$ satisfies - similar to equation \eqref{eq lemma a.2 sde hat X} - the equation
    \begin{equation}\label{eq lemma a.2 sde hat X epsilon,delta}
    \begin{split}
        d \hat X^{\varepsilon,\delta}_u = b(\hat r^{\varepsilon,\delta}_u,\hat m^{\varepsilon,\delta}_u,\hat X^{\varepsilon,\delta}_u,\hat \xi^{\delta}_u) d\hat r^{\varepsilon,\delta}_u
        + \sigma(\hat r^{\varepsilon,\delta}_u,\hat m^{\varepsilon,\delta}_u,\hat X^{\varepsilon,\delta}_u,\hat \xi^{\delta}_u) \sqrt{\dot{\hat r}^{\varepsilon,\delta}_u} d \hat B^{\delta}_u + \gamma(\hat r^{\varepsilon,\delta}_u) d\hat \xi^{\delta}_u,\\
        \hat X^{\varepsilon,\delta}_0 = \tilde X_{t-}.
    \end{split}
    \end{equation}
    
    Since $\hat r^{\varepsilon,\delta}$ is bijective, $\bar W^{\varepsilon,\delta}$ is also an $\hat{\b F}^{\varepsilon,\delta}$-Brownian motion with $\hat{\c F}^{\varepsilon,\delta}_u \coloneqq \c F^{\bar W^{\varepsilon,\delta}}_u \lor \c F^{\hat X^{\varepsilon,\delta},\xi^\delta}_{r^{\varepsilon,\delta}_u}$. As a result,  $(\breve\P_{(\hat X^{\varepsilon,\delta},\hat \xi^{\delta},\bar W^{\varepsilon,\delta})},\hat r^{\varepsilon,\delta})$ is a weak solution to \eqref{eq weak solution W2 time changed SDE} and so 
    \[
        \big(\breve\P_{(\hat X^{\varepsilon,\delta},\hat \xi^{\delta})},\hat r^{\varepsilon,\delta}\big) 
    \]  
    is a $\c P_2$-parametrisation of $\breve\P_{(\tilde X^{\varepsilon,\delta},\tilde \xi^{\varepsilon,\delta})}$. Similarly, we define the state process
    \[
    \bar X^{\varepsilon,\delta}_u \coloneqq \hat X^{\varepsilon,\delta}_{\bar s^{\delta}_u(\omega)},\qquad\text{for }u\in [0,1],
    \]
    so that $(\bar X^{\varepsilon,\delta},\bar\xi,\bar B,\bar s^\delta,\hat r^{\varepsilon,\delta})$ satisfies - similar to equation \eqref{eq lemma a.2 sde bar X} - the equation
    \begin{equation}
        \label{eq lemma a.2 sde bar X epsilon,delta}
        \begin{split}
            d \bar X^{\varepsilon,\delta}_u &= b(\hat r^{\varepsilon,\delta}_{\bar s^\delta_u(\omega)},\hat m^{\varepsilon,\delta}_{\bar s^{\delta}_u(\omega)},\bar X^{\varepsilon,\delta}_u,\bar\xi_u) d\hat r^{\varepsilon,\delta}_{\bar s^\delta_u(\omega)}\\*
            &\qquad + \sigma(\hat r^{\varepsilon,\delta}_{\bar s^\delta_u(\omega)},\hat m^{\varepsilon,\delta}_{\bar s^{\delta}_u(\omega)},\bar X^{\varepsilon,\delta}_u,\bar \xi_u) \sqrt{\dot{\hat r}^{\varepsilon,\delta}_{\bar s^\delta_u(\omega)}\dot{\bar s}^\delta_u(\omega)} d\bar B_u + \gamma(\hat r^{\varepsilon,\delta}_{\bar s^\delta_u(\omega)}) d\bar \xi_u,\quad \bar X^{\varepsilon,\delta}_0 = \tilde X_{t-}.
        \end{split}
    \end{equation}
    By construction, the mapping $(u,\omega) \mapsto (\bar X^{\varepsilon,\delta},\bar\xi,\bar s^\delta)$ is measurable and $\bar W^{\varepsilon,\delta}$ is an $\bar{\b  F}^{\varepsilon,\delta}$-Brownian motion with $\bar{\c F}^{\varepsilon,\delta}_u \coloneqq \c F^{\bar W^{\varepsilon,\delta}}_u\lor \c F^{\bar X^{\varepsilon,\delta},\bar\xi,\bar s^{\delta}}_{s^\delta_{r^{\varepsilon,\delta}_u}}$. As a result, $(\breve\P_{(\bar X^{\varepsilon,\delta},\bar\xi,\bar W^{\varepsilon,\delta},\bar s^\delta)},\hat r^{\varepsilon,\delta})$ is a weak solution to \eqref{eq weak solution two-layer time changed SDE} and hence 
    \[
        ( (\hat m^{\varepsilon,\delta},\hat r^{\varepsilon,\delta}),\breve\P_{(\bar X^{\varepsilon,\delta},\bar\xi,\bar s^\delta)} ) 
    \]  
is a two-layer parametrisation of $\breve\P_{(\tilde X^{\varepsilon,\delta},\tilde \xi^{\varepsilon,\delta})}$.
    
    
    \item[\textbf{ Step 4. \emph{Verification.}}]
    In what follows we show that the previously constructed processes are in fact the desired approximations. We need to verify the following properties:
    \begin{enumerate}[label=(\alph*)]
        \item\label{proof lemma a.2 item a} $(\breve\P_{(\tilde X^{\varepsilon,\delta},\tilde \xi^{\varepsilon,\delta})})_{\varepsilon,\delta}\subseteq \c L(t,m)$,
        \item\label{proof lemma a.2 item b} $((\hat m^{\varepsilon,\delta},\hat r^{\varepsilon,\delta}),\breve\P_{(\bar X^{\varepsilon,\delta},\bar\xi,\bar s^\delta)}) \to ((\hat m,\hat r),\breve\P_{(\bar X,\bar\xi,\bar s)})$ as $\delta,\varepsilon\to 0$ (along a suitable subsequence).
    \end{enumerate}

    Property \ref{proof lemma a.2 item a} follows from \eqref{eq lemma a.2 definition xi epsilon,delta} since $\bar s^\delta(\omega), s^\delta(\omega), \bar\xi(\omega)$ and $r^{\varepsilon,\delta}$ are Lipschitz continuous uniform in $\omega\in\breve\Omega$. To verify property \ref{proof lemma a.2 item b} we first note that by \eqref{eq lemma a.2 definition r delta} for $u\in [0,1]$,
    \begin{equation}
    \begin{split}
        |\hat r^{\varepsilon,\delta}_u - \hat r_u| &= \Bigl|\frac{(\hat r^\varepsilon_u - t) + \delta u}{(T-t) + \delta} (T-t) - (\hat r_u - t)\Bigr| \\
        & \leq \frac {|\hat r^\varepsilon_u - \hat r_u| + \delta |u(T-t) - (\hat r_u -t)|} {(T-t)+\delta} 
         \leq \frac 1 {T-t} |\hat r^\varepsilon-\hat r| + \delta,
     \end{split}
    \end{equation}
    which tends to zero uniformly on $[0,1]$ for $\delta,\varepsilon\to 0$. Moreover, in view of \eqref{eq lemma a.2 definition s delta} we have that
    \begin{equation}\label{eq lemma a.2 estimate bar s}
        |\bar s^{\delta}_u(\omega) - \bar s_u(\omega)| = \Bigl|\frac{\bar s_u(\omega) + \delta u}{1 + \delta} - \bar s_u(\omega)\Bigr| = \frac \delta {1+\delta} |\bar s_u(\omega) - u| \leq 2\delta,
    \end{equation}
    Using these results, we see that
    \begin{equation}\label{eq lemma a.2 estimate bar r epsilon,delta bar r}
        |\hat r^{\varepsilon,\delta}_{\bar s^{\delta}_u(\omega)} - \hat r_{\bar s_u(\omega)}| \leq |\hat r^{\varepsilon,\delta}_{\bar s^{\delta}_u(\omega)} - \hat r_{\bar s^{\delta}_u(\omega)}| + |\hat r_{\bar s^{\delta}_u(\omega)} - \hat r_{\bar s_u(\omega)}|,
    \end{equation}
    which again tends to zero uniformly in $(u,\omega)$ as $\delta,\varepsilon\to 0$ since $\hat r$ is Lipschitz.

    \item[\textbf{ Step 4.1. \emph{Convergence of time scale derivatives.}}]
    To establish the convergence of the state-control process we need the following convergence results for derivatives of the time scales that are used to define the dynamics of the process $\hat X^{\varepsilon, \delta}$ and $\bar X^{\varepsilon, \delta}$, respectively. By  construction,
    \begin{equation}\label{eq lemma a.2 convergence dot hat r epsilon,delta dot hat r}
        |\dot{\hat r}^{\varepsilon,\delta}_u - \dot{\hat r}_u| \leq |\dot{\hat r}^{\varepsilon,\delta}_u - \dot{\hat r}^{\varepsilon}_u| + |\dot{\hat r}^{\varepsilon}_u - \dot{\hat r}_u|\to 0\qquad\text{a.e. on $[0,1]$ for }\delta,\varepsilon\to 0.
    \end{equation}
    Similarly,
    \begin{equation}\label{eq lemma a.2 convergence dot bar s delta dot bar s}
        |\dot{\bar s}^{\delta}_u(\omega) - \dot{\bar s}_u(\omega)| \to 0\qquad\text{uniformly in $(u,\omega)$ for }\delta\to 0.
    \end{equation}
    Putting both results together, we obtain
    \begin{equation}
    \begin{split}
        |\dot{\hat r}^{\varepsilon,\delta}_{\bar s^\delta_u(\omega)}\dot{\bar s}^\delta_u(\omega) - \dot{\hat r}_{\bar s_u(\omega)}\dot{\bar s}_u(\omega)|
        &\leq |\dot{\hat r}^{\varepsilon,\delta}_{\bar s^{\delta}_u(\omega)} - \dot{\hat r}^\varepsilon_{\bar s^{\delta}_u(\omega)}| \dot{\bar s}^{\delta}_u(\omega) + |\dot{\hat r}^{\varepsilon}_{\bar s^{\delta}_u(\omega)} - \dot{\hat r}^{\varepsilon}_{\bar s_u(\omega)}| \dot{\bar s}^{\delta}_u(\omega)\\*
        &\qquad+ \dot{\hat r}^{\varepsilon}_{\bar s_u(\omega)} |\dot{\bar s}^{\delta}_u(\omega) - \dot{\bar s}_u(\omega)| + |\dot{\hat r}^\varepsilon_{\bar s_u(\omega)} - \dot{\hat r}_{\bar s_u(\omega)} | \dot{\bar s}_u(\omega).
    \end{split}
    \end{equation}
   
    The first and third term vanishes for $\delta,\varepsilon\to 0$ since $\dot{\bar s}^\delta$ and $\dot{\hat r}^\varepsilon$ are uniformly bounded in $\delta$ and $\varepsilon$ for sufficiently small $\delta$. Furthermore, for every fixed $\varepsilon$, the second term tends to zero for $\delta\to 0$ since $\hat r^\varepsilon$ is smooth. W.l.o.g. we choose a suitable subsequence of $\delta,\varepsilon\to 0$ by sending $\delta\to 0$ sufficiently fast such that the first three terms above vanish. Regarding the last term, we note that as $\varepsilon\to 0$,
    \begin{equation}
        \int_0^1 |\dot{\hat r}^\varepsilon_{\bar s_u(\omega)} - \dot{\hat r}_{\bar s_u(\omega)} | \dot{\bar s}_u(\omega) du = \int_0^1 |\dot{\hat r}^\varepsilon_{\bar s_u(\omega)} - \dot{\hat r}_{\bar s_u(\omega)} | d\bar s_u(\omega)
        = \int_0^1 |\dot{\hat r}^\varepsilon_u - \dot{\hat r}_u | du \to 0,
    \end{equation}
    since $\hat r^\varepsilon$ was constructed as mollification of $\hat r$. Therefore, along a further subsequence,
    \[
    |\dot{\hat r}^\varepsilon_{\bar s_u(\omega)} - \dot{\hat r}_{\bar s_u(\omega)} | \dot{\bar s}_u(\omega) \to 0\qquad\breve\P\otimes du\text{-a.e. on }\breve\Omega\times [0,1],
    \]
    and thus
    \begin{equation}\label{eq lemma a.2 convergence dot bar r epsilon,delta dot bar r}
        \dot{\hat r}^{\varepsilon,\delta}_{\bar s^\delta_u(\omega)}\dot{\bar s}^\delta_u(\omega) \to \dot{\hat r}_{\bar s_u(\omega)}\dot{\bar s}_u(\omega)\qquad\breve\P\otimes du\text{-a.e. on }\breve\Omega\times [0,1].
    \end{equation}

    \item[\textbf{ Step 4.2. \emph{Convergence of the state-control process.}}]
    Let us now turn to the state-control-dynamics. Since $\hat\xi$ is monotone together with \eqref{eq lemma a.2 estimate bar s}, we also get
    \[
        \abs{\hat \xi^{\delta}_u - \hat\xi_u} = \abs{\bar \xi_{s^{\delta}_u(\omega)} - \hat \xi_u}
        = \abs{\hat \xi_{\bar s_{s^{\delta}_u(\omega)}(\omega)} - \hat\xi_u} \leq \abs{\hat\xi_{u + 2\delta} - \hat\xi_{u-2\delta}},\qquad\breve\P\text{-a.s.}
    \]
    In particular, since $[0,1]\ni u\mapsto \hat\xi_u\in L^2(\breve\Omega,\breve{\c F},\breve\P;\R^l)$ is uniformly continuous, we obtain
    \begin{equation}\label{eq lemma a.2 estimate xi hat epsilon,delta xi hat}
        \norm{\hat \xi^{\delta}_u - \hat\xi_u}_{L^2(\breve\P)} \leq \norm{\hat\xi_{u+2\delta} - \hat\xi_{u-2\delta}}_{L^2(\breve\P)} \to 0,\quad{\text{as }\delta\to 0},\qquad\text{uniformly in }u.
    \end{equation}

    We are left with proving the convergence of the state processes $\hat X^{\varepsilon,\delta}$ and $\bar X^{\varepsilon,\delta}(\omega)$. Our main tool will be Gronwall's inequality. Starting with $\hat X^{\varepsilon,\delta}$ we have by \eqref{eq lemma a.2 sde hat X}, \eqref{eq lemma a.2 sde hat X epsilon,delta} for all $u_* \in [0,1]$,
    \[
        \sup_{u\in [0,u_*]} \E^{\breve\P}\big[\abs{\hat X^{\varepsilon,\delta}_u - \hat X_u}^2 \big] \leq 4 (I_1 + I_2 + I_3),
    \]
    where
    \begin{equation}
    \begin{split}
        I_1 &\coloneqq \sup_{u\in [0,u_*]} \E^{\breve\P}\bigg[\bigg|\int_0^u b(\hat r^{\varepsilon,\delta}_v,\hat m^{\varepsilon,\delta}_v,\hat X^{\varepsilon,\delta}_v,\hat\xi^{\delta}_v) d\hat r^{\varepsilon,\delta}_v - \int_0^u b(\hat r_v,\hat m_v,\hat X_v,\hat \xi_v) d\hat r_v\bigg|^2 \bigg],\\
        I_2 &\coloneqq \sup_{u\in [0,u_*]} \E^{\breve\P}\bigg[\bigg| \int_0^u \sigma(\hat r^{\varepsilon,\delta}_v,\hat m^{\varepsilon,\delta}_v,\hat X^{\varepsilon,\delta}_v,\hat \xi^{\delta}_v) \sqrt{\dot{\hat r}^{\varepsilon,\delta}_v} d \hat B^{\delta}_v
        - \int_0^u \sigma(\hat r_v,\hat m_v,\hat X_v,\hat \xi_v) \sqrt{\dot{\hat r}_v} d\hat B_v \bigg|^2 \bigg],\\
        I_3 &\coloneqq \sup_{u\in [0,u_*]} \E^{\breve\P}\bigg[\bigg|\int_0^u \gamma(\hat r^{\varepsilon,\delta}_v) d\hat \xi^{\delta}_v - \int_0^u \gamma(\hat r_v) d\hat \xi_v \bigg|^2\bigg].
    \end{split}
    \end{equation}
    For the ease of reading, we introduce the following short-hand notations for $h\in \{b,\sigma\}$,
    \[
    \hat h_v \coloneqq h(\hat r_v,\hat m_v,\hat X_v,\hat \xi_v),\qquad \hat h^{\varepsilon,\delta}_v \coloneqq h(\hat r^{\varepsilon,\delta}_v, \hat m^{\varepsilon,\delta}_v, \hat X^{\varepsilon,\delta}_v,\hat\xi^{\delta}_v).
    \]
    Using the new notation, we can bound the first term $I_1$ using \eqref{eq lemma a.2 definition r delta} as follows
    \begin{equation}
    \begin{split}
        I_1 &= \sup_{u\in [0,u_*]} \E^{\breve\P}\bigg[\bigg|\int_0^u \hat b^{\varepsilon,\delta}_v d\hat r^{\varepsilon,\delta}_v - \int_0^u \hat b_v d\hat r_v\bigg|^2 \bigg]\\
        &\leq 2 \sup_{u\in [0,u_*]}\E^{\breve\P}\bigg[\bigg|\int_0^u (\hat b^{\varepsilon,\delta}_v - \hat b_v) d\hat r^{\varepsilon,\delta}_v\bigg|^2 \bigg] + 2 \sup_{u\in [0,u_*]}\E^{\breve\P}\bigg[\bigg|\int_0^u \hat b_v (d\hat r^{\varepsilon,\delta}_v - d\hat r_v) \bigg|^2\bigg]\\
        &\leq 2 \E^{\breve\P}\bigg[\int_0^{u_*} |\hat b^{\varepsilon,\delta}_v - \hat b_v|^2 |\dot{\hat r}^{\varepsilon,\delta}_v|^2  dv\bigg] + 2 \E^{\breve\P}\bigg[\int_0^{u_*} |\hat b_v|^2 \abs{\dot{\hat r}^{\varepsilon,\delta}_v - \dot{\hat r}_v}^2 dv\bigg].
    \end{split}
    \end{equation}
    Controlling the first term using the Lipschitz continuity of $b$ is standard and the second term vanishes as $\delta,\varepsilon\to 0$ by dominated convergence using \eqref{eq lemma a.2 convergence dot hat r epsilon,delta dot hat r}.
    
    For the second term $I_2$, we recall that $s_u(\omega)\coloneqq \inf\{v \in [0,1]\mid \bar s_v(\omega) > u\} \land 1$ denotes the right-continuous inverse of $\bar s(\omega)$ and use \eqref{eq lemma a.2 definition bar W} and \eqref{eq lemma a.2 definition hat W,delta},
    \begin{align} 
        I_2 &= \sup_{u\in [0,u_*]} \E^{\breve\P}\bigg[\bigg| \int_0^u \hat\sigma^{\varepsilon,\delta}_v \sqrt{\dot{\hat r}^{\varepsilon,\delta}_v} d \hat B^{\delta}_v - \int_0^u \hat\sigma_v \sqrt{\dot{\hat r}_v} d\hat B_v \bigg|^2 \bigg]\\
        &= \sup_{u\in [0,u_*]} \E^{\breve\P}\bigg[\bigg| \int_0^u \hat\sigma^{\varepsilon,\delta}_v \sqrt{\dot{\hat r}^{\varepsilon,\delta}_v} d \hat B^{\delta}_v - \int_0^u \hat\sigma_v \sqrt{\dot{\hat r}^{\varepsilon,\delta}_v} d\hat B^\delta_v\\*
        &\quad + \int_0^{s^\delta_u(\omega)} \hat\sigma_{\bar s^\delta_v(\omega)} \sqrt{\dot{\hat r}^{\varepsilon,\delta}_{\bar s^\delta_v(\omega)} \dot{\bar s}^\delta_u(\omega)} d\bar B_u - \int_0^u \hat\sigma_v \sqrt{\dot{\hat r}_v} d\hat B_v \bigg|^2 \bigg]\\
        &\leq 4 \sup_{u\in [0,u_*]} \E^{\breve\P}\bigg[\bigg| \int_0^u \hat\sigma^{\varepsilon,\delta}_v \sqrt{\dot{\hat r}^{\varepsilon,\delta}_v} d \hat B^{\delta}_v - \int_0^u \hat\sigma_v \sqrt{\dot{\hat r}^{\varepsilon,\delta}_v} d\hat B^{\delta}_v \bigg|^2 \bigg]\\*
        &\quad + 4\sup_{u\in [0,u_*]} \E^{\breve\P}\bigg[\bigg| \int_0^{\bar s_{s^{\delta}_u(\omega)}(\omega)} \hat\sigma_{\bar s^{\delta}_{s_v(\omega)}(\omega)} \1_{\{\dot{\bar s}_{s_v(\omega)}(\omega) \not= 0\}} \frac {\sqrt{\dot{\hat r}^{\varepsilon,\delta}_{\bar s^{\delta}_{s_v(\omega)}(\omega)} \dot{\bar s}^{\delta}_{s_v(\omega)}(\omega)}} {\sqrt{\dot{\bar s}_{s_v(\omega)}(\omega)}} d\hat B_v\\*
        &\quad - \int_0^u \hat\sigma_v \sqrt{\dot{\hat r}_v} d\hat B_v\bigg|^2 \bigg]+ 4\sup_{u\in [0,u_*]} \E^{\breve\P}\bigg[\bigg|\int_0^{s^{\delta}_u(\omega)} \hat\sigma_{\bar s^{\delta}_v(\omega)} \sqrt{\dot{\hat r}^{\varepsilon,\delta}_{\bar s^{\delta}_v(\omega)} \dot{\bar s}^{\delta}_v(\omega)} \1_{\{\dot{\bar s}_v(\omega) = 0\}} d\breve W_v \bigg|^2 \bigg]\\
        &\leq  4 \E^{\breve\P}\bigg[\int_0^{u_*} \big| \hat\sigma^{\varepsilon,\delta}_v \sqrt{\dot{\hat r}^{\varepsilon,\delta}_v} - \hat\sigma_v \sqrt{\dot{\hat r}^{\varepsilon,\delta}_v} \big|^2 dv \bigg]
        \\*
        &\quad + 4\sup_{u\in [0,u_*]} \E^{\breve\P}\bigg[\int_0^1 \bigg| \1_{[0,\bar s_{s^{\delta}_u(\omega)}(\omega)]}(v) \hat\sigma_{\bar s^{\delta}_{s_v(\omega)}(\omega)} \1_{\{\dot{\bar s}_{s_v(\omega)}(\omega) \not= 0\}} \frac {\sqrt{\dot{\hat r}^{\varepsilon,\delta}_{\bar s^{\delta}_{s_v(\omega)}(\omega)}\dot{\bar s}^{\delta}_{s_v(\omega)}(\omega)}} {\sqrt{\dot{\bar s}_{s_v(\omega)}(\omega)}}\\*
        &\quad- \1_{[0,u]}(v) \hat \sigma_v \sqrt{\dot{\hat r}_v}\bigg|^2 dv \bigg]+ 4 \E^{\breve\P}\bigg[\int_0^1 | \hat\sigma_{\bar s^{\delta}_v(\omega)} |^2 \dot{\hat r}^{\varepsilon,\delta}_{\bar s^{\delta}_v(\omega)} \dot{\bar s}^{\delta}_v(\omega) \1_{\{\dot{\bar s}_v(\omega) = 0\}} dv \bigg].
    \end{align}
    While the first term can again be controlled using standard Lipschitz arguments, we also note that the last term vanishes for $\delta\to 0$ by dominated convergence using \eqref{eq lemma a.2 convergence dot bar s delta dot bar s}. Thus we will now focus on the middle term. Noting that $\dot{\bar s}_{s_v(\omega)}(\omega) \not= 0$, $dv\otimes\breve\P$-a.s., we can rewrite this term as follows
    \begin{align}
        &4\sup_{u\in [0,u_*]} \E^{\breve\P}\bigg[\int_0^1 \bigg| \1_{[0,\bar s_{s^{\delta}_u(\omega)}(\omega)]}(v) \hat\sigma_{\bar s^{\delta}_{s_v(\omega)}(\omega)} \1_{\{\dot{\bar s}_{s_v(\omega)}(\omega) \not= 0\}} \frac {\sqrt{\dot{\hat r}^{\varepsilon,\delta}_{\bar s^{\delta}_{s_v(\omega)}(\omega)} \dot{\bar s}^{\delta}_{s_v(\omega)}(\omega)}} {\sqrt{\dot{\bar s}_{s_v(\omega)}(\omega)}}\\*
        &\qquad - \1_{[0,u]}(v) \hat \sigma_v \sqrt{\dot{\hat r}_v}\bigg|^2 dv \bigg]\\
        &= 4\sup_{u\in [0,u_*]} \E^{\breve\P}\bigg[\int_0^1 \1_{\{\dot{\bar s}_{s_v(\omega)}(\omega) \not= 0\}}\bigg| \1_{[0,\bar s_{s^{\delta}_u(\omega)}(\omega)]}(v) \hat\sigma_{\bar s^{\delta}_{s_v(\omega)}(\omega)} \sqrt{\dot{\hat r}^{\varepsilon,\delta}_{\bar s^{\delta}_{s_v(\omega)}(\omega)} \dot{\bar s}^{\delta}_{s_v(\omega)}(\omega)}\\*
        &\qquad - \1_{[0,u]}(v) \hat \sigma_v \sqrt{\dot{\hat r}_v \dot{\bar s}_{s_v(\omega)}(\omega)}\bigg|^2 ds_v(\omega) \bigg]\\
        &= 4\sup_{u\in [0,u_*]} \E^{\breve\P}\bigg[\int_0^1 \1_{\{\dot{\bar s}_v(\omega) \not= 0\}} \bigg| \1_{[0,s^{\delta}_u(\omega)]}(v) \hat\sigma_{\bar s^{\delta}_v(\omega)} \sqrt{\dot{\hat r}^{\varepsilon,\delta}_{\bar s^{\delta}_v(\omega)} \dot{\bar s}^{\delta}_v(\omega)}\\*
        &\qquad - \1_{[0,s_u(\omega)]}(v) \hat\sigma_{\bar s_v(\omega)} \sqrt{\dot{\hat r}_{\bar s_v(\omega)} \dot{\bar s}_v(\omega)} \bigg|^2 dv \bigg]\\
        &\leq 8 \E^{\breve\P}\bigg[\int_0^1 \1_{\{\dot{\bar s}_v(\omega) \not= 0\}} \bigg| \hat\sigma_{\bar s^{\delta}_v(\omega)} \sqrt{\dot{\hat r}^{\varepsilon,\delta}_{\bar s^\delta_v(\omega)}\dot{\bar s}^\delta_v(\omega)} - \hat\sigma_{\bar s_v(\omega)} \sqrt{\dot{\hat r}_{\bar s_v(\omega)}\dot{\bar s}_v(\omega)} \bigg|^2 dv \bigg]\\*
        &\qquad + 16 \sup_{u\in [0,u_*]} \E^{\breve\P}\bigg[\int_{s^{\delta}_u(\omega)\land s_u(\omega)}^{s^{\delta}_u(\omega)\lor s_u(\omega)} \1_{\{\dot{\bar s}_v(\omega) \not= 0\}} \Big(| \hat\sigma_{\bar s^{\delta}_v(\omega)} |^2 \dot{\hat r}^{\varepsilon,\delta}_{\bar s^{\delta}_v(\omega)} \dot{\bar s}^{\delta}_v(\omega)\\*
        &\qquad+ |\hat \sigma_{\bar s_v(\omega)} |^2 \dot{\hat r}_{\bar s_v(\omega)} \dot{\bar s}_v(\omega) \Big) dv\bigg].
    \end{align}
    For the first term, we note that $\dot{\bar s}_v(\omega) \not= 0$ implies $dv\otimes\breve\P$-a.s. that $\bar s_v(\omega) \in J^c_{[0,1]}(\hat\xi(\omega))$, which implies due to continuity of $\sigma$ that $dv\otimes\breve\P$-a.s, $\hat\sigma$ is also continuous in $\bar s_v(\omega)$. This allows us to use dominated convergence to show that the first term vanishes for $\delta,\varepsilon\to 0$ using Lipschitz continuity of $\sigma$ and \eqref{eq lemma a.2 convergence dot bar r epsilon,delta dot bar r}. Regarding the second term, we note that since $\hat r$ and $\hat r^{\varepsilon,\delta}$ are Lipschitz continuous uniform in $\varepsilon,\delta$ for all sufficiently small $\delta$,
    \begin{align}
        &\sup_{u\in [0,u_*]} \E^{\breve\P}\bigg[\int_{s^{\delta}_u(\omega)\land s_u(\omega)}^{s^{\delta}_u(\omega)\lor s_u(\omega)} \1_{\{\dot{\bar s}_v(\omega) \not= 0\}} \Big(| \hat\sigma_{\bar s^{\delta}_v(\omega)} |^2 \dot{\hat r}^{\varepsilon,\delta}_{\bar s^{\delta}_v(\omega)} \dot{\bar s}^{\delta}_v(\omega) + |\hat\sigma_{\bar s_v(\omega)} |^2 \dot{\hat r}_{\bar s_v(\omega)} \dot{\bar s}_v(\omega) \Big) dv\bigg]\\
        &\leq C \E^{\breve\P}\bigg[\sup_{v \in [0,1]} \abs{\hat\sigma_v}^2 \sup_{u\in [0,u_*]} \int_{s^{\delta}_u(\omega)\land s_u(\omega)}^{s^{\delta}_u(\omega)\lor s_u(\omega)} (\dot{\bar s}^{\delta}_v(\omega) + \dot{\bar s}_v(\omega)) dv\bigg]\\
        &\leq C \E^{\breve\P}\Big[\sup_{v \in [0,1]} \abs{\hat\sigma_v}^2 \sup_{u\in [0,u_*]} \abs{\bar s^{\delta}_{s_u(\omega)}(\omega) - \bar s_{s^{\delta}_u(\omega)}(\omega)}\Big]\\
        &= C \E^{\breve\P}\Big[\sup_{v \in [0,1]} \abs{\hat\sigma_v}^2 \sup_{u\in [0,u_*]} \bigl\vert\bar s^{\delta}_{s_u(\omega)}(\omega) - \bar s_{s_u(\omega)}(\omega) + \bar s^{\delta}_{s^{\delta}_u(\omega)}(\omega) - \bar s_{s^{\delta}_u(\omega)}(\omega)\bigr\vert\Big]\\
        &\leq 2 C \E^{\breve\P}\Big[\sup_{v\in [0,1]} \abs{\hat\sigma_v}^2 \Big] \sup_{u \in [0,1]}\esssup_{\omega\in\bar\Omega} |\bar s^{\delta}_u(\omega) - \bar s_u(\omega)|.
    \end{align}
    Therefore this term also vanishes as $\delta\to 0$.
    
    Now we turn to the last term $I_3$ and note that
    \begin{align}
        I_3 &= \sup_{u\in [0,u_*]} \E^{\breve\P}\bigg[\bigg|\int_0^{s^{\delta}_u(\omega)} \gamma(\hat r^{\varepsilon,\delta}_{\bar s^\delta_v(\omega)}) d\bar\xi_v - \int_0^{s_u(\omega)} \gamma(\hat r_{\bar s_v(\omega)}) d\bar\xi_v\bigg|^2 \bigg]\\
        &\leq 2 \sup_{u\in [0,u_*]} \E^{\breve\P}\bigg[\bigg|\int_0^{s^{\delta}_u(\omega)} (\gamma(\hat r^{\varepsilon,\delta}_{\bar s^\delta_v(\omega)}) - \gamma(\hat r_{\bar s_v(\omega)})) d\bar\xi_v\bigg|^2 \bigg]\\*
        &\qquad+ \sup_{u\in [0,u_*]} \E^{\breve\P}\bigg[\bigg|\int_{s^{\delta}_u(\omega)}^{s_u(\omega)} \gamma(\hat r_{\bar s_v(\omega)}) d\bar\xi_v \bigg|^2\bigg]\\
        &\leq 2\E^{\breve\P}\big[\abs{\bar\xi_1 - \bar\xi_0}^2 \big] \sup_{u \in [0,1]}\esssup_{\omega\in\bar\Omega} |\gamma(\hat r^{\varepsilon,\delta}_{\bar s^\delta_u(\omega)}) - \gamma(\hat r_{\bar s_u(\omega)})|^2\\*
        &\qquad+ \norm{\gamma}_{{\infty}}^2 \sup_{u\in [0,u_*]} \E^{\breve\P}\big[\abs{\bar \xi_{s_u(\omega)} - \bar \xi_{s^{\delta}_u(\omega)}}^2 \big]\\
        &\leq 2\E^{\breve\P}\big[\abs{\bar\xi_1 - \bar\xi_0}^2 \big] \sup_{u \in [0,1]}\esssup_{\omega\in\bar\Omega} |\gamma(\hat r^{\varepsilon,\delta}_{\bar s^\delta_u(\omega)}) - \gamma(\hat r_{\bar s_u(\omega)})|^2 + \norm{\gamma}_{{\infty}}^2 \sup_{u\in [0,1]} \E^{\breve\P}\big[\abs{\hat \xi_u - \hat \xi^{\delta}_u}^2 \big].
    \end{align}
    The second part goes to 0 for $\delta,\varepsilon\to 0$ by \eqref{eq lemma a.2 estimate xi hat epsilon,delta xi hat} and the first part vanishes for $\delta\to 0$ since $\gamma$ is uniformly continuous on $[0,1]$ together with \eqref{eq lemma a.2 estimate bar r epsilon,delta bar r}.
    
    Combining the estimates for $I_1$, $I_2$ and $I_3$, we see from an application of Gronwall's Lemma that
    \[
    \sup_{u \in [0,1]}\c W_2(\hat m_u^{\varepsilon,\delta},\hat m_u) \leq \sup_{u\in [0,1]} \E^{\breve\P}\big[\abs{\hat X^{\varepsilon,\delta}_u - \hat X_u}^2\big] + \sup_{u\in [0,1]} \E^{\breve\P}\big[\abs{\hat \xi^\delta_u - \hat \xi_u}^2\big] \to 0\qquad\text{as }\delta,\varepsilon\to 0.
    \]
    
    
    Turning to $\bar X^{\varepsilon,\delta}$, we again first introduce the following short-hand notations for $h\in \{b,\sigma\}$,
    \[
    \bar h_v \coloneqq h(\hat r_{\bar s_v(\omega)},\hat m_{\bar s_v(\omega)},\bar X_v,\bar \xi_v),\qquad \bar h^{\varepsilon,\delta}_v \coloneqq h(\hat r^{\varepsilon,\delta}_{\bar s^\delta_v(\omega)}, \hat m^{\varepsilon,\delta}_{\bar s^\delta_v(\omega)}, \bar X^{\varepsilon,\delta}_v,\bar\xi_v).
    \]
    Now we observe similar to before, by \eqref{eq lemma a.2 sde bar X}, \eqref{eq lemma a.2 sde bar X epsilon,delta} that for all $u_*\in [0,1]$,
    \[
        \E^{\breve\P}\Big[\sup_{u\in [0,u_*]} \abs{\bar X^{\varepsilon,\delta}_u - \bar X_u}^2 \Big] \leq 4(J_1 + J_2 + J_3),
    \]
    where
    \begin{equation}
    \begin{split}
        J_1 &\coloneqq \E^{\breve\P}\bigg[\sup_{u\in [0,u_*]} \bigg|\int_0^u \bar b^{\varepsilon,\delta}_v d\hat r^{\varepsilon,\delta}_{\bar s^\delta_v(\omega)}
        - \int_0^u \bar b_v d\hat r_{\bar s_v(\omega)} \bigg|^2 \bigg],\\
        J_2 &\coloneqq \E^{\breve\P}\bigg[\sup_{u\in [0,u_*]} \bigg|\int_0^u \bar\sigma^{\varepsilon,\delta}_v \sqrt{\dot{\hat r}^{\varepsilon,\delta}_{\bar s^\delta_v(\omega)}\dot{\bar s}^\delta_v(\omega)} d\bar B_v
        - \int_0^u \bar\sigma_v \sqrt{\dot{\hat r}_{\bar s_v(\omega)}\dot{\bar s}_v(\omega)} d\bar B_v \bigg|^2 \bigg],\\
        J_3 &\coloneqq \E^{\breve\P}\bigg[\sup_{u\in [0,u_*]} \bigg|\int_0^u \gamma(\hat r^{\varepsilon,\delta}_{\bar s^\delta_v(\omega)}) d\bar \xi_v - \int_0^u \gamma(\hat r_{\bar s_v(\omega)}) d\bar\xi_v \bigg|^2 \bigg].
    \end{split}
    \end{equation}
    
    We can bound the first term $J_1$ as follows
    \begin{equation}
    \begin{split}
        J_1 &\leq 2\E^{\breve\P}\bigg[\sup_{u\in [0,u_*]} \bigg|\int_0^u \bar b^{\varepsilon,\delta}_v - \bar b_v d\hat r^{\varepsilon,\delta}_{\bar s^\delta_v(\omega)} \bigg|^2 \bigg]
        + 2 \E^{\breve\P}\bigg[\sup_{u\in [0,u_*]} \bigg|\int_0^u \bar b_v (d\hat r^{\varepsilon,\delta}_{\bar s^\delta_v(\omega)} - d\hat r_{\bar s_v(\omega)}) \bigg| ^2 \bigg]\\
        &\leq 2 \E^{\breve\P}\bigg[\int_0^{u_*} |\bar b^{\varepsilon,\delta}_v - \bar b_v|^2 |\dot{\hat r}^{\varepsilon,\delta}_{\bar s^\delta_v(\omega)}\dot{\bar s}^\delta_v(\omega)|^2 dv \bigg]\\
        &\qquad+ 2 \E^{\breve\P}\bigg[\int_0^{u_*} \abs{\bar b_v}^2 \abs{\dot{\hat r}^{\varepsilon,\delta}_{\bar s^\delta_v(\omega)}\dot{\bar s}^\delta_v(\omega) - \dot{\hat r}_{\bar s_v(\omega)}\dot{\bar s}_v(\omega)}^2 dv \bigg].
    \end{split}
    \end{equation}
    Similar to $I_1$, we can bound the first term using standard arguments since $b$ is Lipschitz and $\dot{\hat r}^{\varepsilon,\delta}_{\bar s^\delta} \dot{\bar s}^\delta$ is uniformly bounded for all sufficiently small $\delta$. Further the second term vanishes as $\delta,\varepsilon\to 0$ by dominated convergence using \eqref{eq lemma a.2 convergence dot bar r epsilon,delta dot bar r}.
    
    Regarding the second term $J_2$ we note that due to Doob's martingale inequality
    \begin{align}
        J_2 &\leq 4 \E^{\breve\P}\bigg[\bigg|\int_0^{u_*} \big(\bar\sigma^{\varepsilon,\delta}_v \sqrt{\dot{\hat r}^{\varepsilon,\delta}_{\bar s^\delta_v(\omega)}\dot{\bar s}^\delta_v(\omega)} - \bar\sigma_v \sqrt{\dot{\hat r}_{\bar s_v(\omega)}\dot{\bar s}_v(\omega)}\big) d\bar B_v \bigg|^2 \bigg]\\
        &= 4\E^{\breve\P}\bigg[\int_0^{u_*} \big|\bar\sigma^{\varepsilon,\delta}_v \sqrt{\dot{\hat r}^{\varepsilon,\delta}_{\bar s^\delta_v(\omega)}\dot{\bar s}^\delta_v(\omega)} - \bar\sigma_v \sqrt{\dot{\hat r}_{\bar s_v(\omega)}\dot{\bar s}_v(\omega)}\big|^2 dv\bigg]\\
        &\leq 8\E^{\breve\P}\bigg[\int_0^{u_*} |\bar\sigma^{\varepsilon,\delta}_v - \bar\sigma_v |^2 \dot{\hat r}^{\varepsilon,\delta}_{\bar s^\delta_v(\omega)}\dot{\bar s}^\delta_v(\omega) dv\bigg]\\*
        &\qquad+ 8\E^{\breve\P}\bigg[\int_0^{u_*} \abs{\bar\sigma_v}^2 \big|\sqrt{\dot{\hat r}_{\bar s_v(\omega)}\dot{\bar s}_v(\omega)} - \sqrt{\dot{\hat r}^{\varepsilon,\delta}_{\bar s^\delta_v(\omega)}\dot{\bar s}^\delta_v(\omega)}\big|^2 dv \bigg].
    \end{align}
    Again controlling the first term is standard for all sufficiently small $\delta$ and the second term vanishes for $\delta,\varepsilon\to 0$ by dominated convergence using \eqref{eq lemma a.2 convergence dot bar r epsilon,delta dot bar r}.
    
    We can estimate the last term  $J_3$ as follows
    \begin{equation}
    \begin{split}
        J_3 &= \E^{\breve\P}\bigg[\sup_{u \in [0,u_*]} \bigg(\int_0^u \gamma(\hat r^{\varepsilon,\delta}_{\bar s^\delta_v(\omega)}) - \gamma(\hat r_{\bar s_v(\omega)}) d\bar \xi_v \bigg)^2 \bigg]\\
        &\leq \E^{\breve\P}\big[\abs{\bar \xi_1 - \bar \xi_0}^2 \big] \sup_{u \in [0,u_*]}\esssup_{\omega\in\bar\Omega} |\gamma(\hat r^{\varepsilon,\delta}_{\bar s^\delta_u(\omega)}) - \gamma(\hat r_{\bar s_u(\omega)}) |^2,
    \end{split}
    \end{equation}
    which goes to 0 for $\delta\to 0$ by \eqref{eq lemma a.2 estimate bar r epsilon,delta bar r} since $\gamma$ is uniformly continuous on $[0,1]$.
    
    Finally, by combining the previous estimates and applying Gronwall's Lemma, we obtain
    \[
    \E^{\breve\P}\Big[\sup_{u\in [0,1]} \abs{\bar X^{\varepsilon,\delta}_u - \bar X_u}^2 \Big] \to 0\qquad\text{as }\delta\to 0.
    \]
    Together with our previous results, this yields, at least along a subsequence,
    \[
    \breve\P_{(\bar X^{\varepsilon,\delta},\bar\xi,\bar s^\delta)} \to \breve\P_{(\bar X,\bar\xi,\bar s)} = \bar\P\qquad\text{in }\c P_2(C([0,1];\R^d\times\R^l\times [0,1])).
    \]
    \end{enumerate}
\end{proof}

\section{{The reward functional}} \label{section reward functional}

{Building upon the two-layer parametrisation framework developed in Section \ref{section parametrisations}, this section derives a continuous reward functional $J_2$ on the space of parametrisations. We then establish the equivalence between the original definition of the reward functional $J$ in \eqref{eq reward function general definition} and the supremum of $J_2$ over all possible parametrisations. Finally, we use this equivalence to obtain a more explicit representation of $J$ in terms of minimal jump costs, extending the intuition from the non-mean-field case discussed in Section \ref{subsection motivation parametrisations}.}

\subsection{The reward functional for two-layer parametrisations}\label{subsection reward functional for parametrisations}
We now introduce a reward function on the set of two-layer parametrisations. This function turns out to be continuous. Subsequently, we derive an alternative representation of our original reward functional in the spirit of \eqref{eq reward functional motivation supremum parametrisation}.

\begin{definition}\label{definition reward functional for parametrisations}
    We define the reward of a given two-layer parametrisation $((\hat m,\hat r),\bar \P)$ as
\begin{equation}
\begin{split}
    &J_2((\hat m,\hat r),\bar\P)\\
    &\coloneqq \E^{\bar\P}\bigg[\int_0^1 f(\hat r_{\bar s_u},\hat m_{\bar s_u},\bar X_u,\bar\xi_u) d\hat r_{\bar s_u} + g(\hat m_{\bar s_1},\bar X_1,\bar\xi_1) - \int_0^1 c(\hat r_{\bar s_u},\hat m_{\bar s_u},\bar X_u,\bar \xi_u)d\bar\xi_u \bigg].
\end{split}
\end{equation}
\end{definition}

It follows from the above definition that whenever $\P\in \c C(t,m)$ is a continuous control and $((\hat m,\hat r),\bar\P)$ a two-layer parametrisation of $\P$, then $J_{\c C}(t,m,\P) = J_2((\hat m,\hat r),\bar\P)$ since the above definition is invariant under reparametrisation of the time scale.

\begin{lemma}\label{lemma reward function continuous}
    The reward functional $J_2$ defined in Definition \ref{definition reward functional for parametrisations} is continuous with respect to convergence of two-layer parametrisations.
\end{lemma}

\begin{proof}
    Let $((\hat m^n,\hat r^n),\bar\P^n)_n$ and $((\hat m,\hat r),\bar\P))$ be two-layer parametrisations such that
    \[
    ((\hat m^n,\hat r^n),\bar\P^n)\to((\hat m,\hat r),\bar\P).
    \]
    Then, since $\bar\P^n\to\bar\P$ in $\c P_2$, we can apply Skorokhod's representation theorem to obtain processes $(\bar X^n,\bar\xi^n,\bar s^n)$ and $(\bar X,\bar\xi,\bar s)$ on a common probability space $(\check\Omega,\check{\c F},\check\P)$ such that
    \begin{equation}\label{eq proposition 3.5 final parametrisation convergence}
    (\bar X^n,\bar\xi^n,\bar s^n) \to (\bar X,\bar\xi,\bar s)\qquad\check\P\text{-a.s. and in }L^2.
    \end{equation}
    Due to assumption \ref{assumption g continuous}, this implies that
    \[
    \E^{\bar\P^n}[g(\hat m^n_{\bar s_1},\bar X_1,\bar\xi_1)] = \E^{\check\P}[g(\hat m^n_{\bar s^n_1},\bar X^n_1,\bar\xi^n_1)] \to \E^{\check\P}[g(\hat m_{\bar s_1},\bar X_1,\bar\xi_1)] = \E^{\bar\P}[g(\hat m_{\bar s_1},\bar X_1,\bar\xi_1)].
    \]
    To establish the convergence of the jump costs
    {
    \[
    \E^{\bar\P^n}\bigg[\int_0^1 c(\hat r_{\bar s_u},\hat m^n_{\bar s_u},\bar X_u,\bar \xi_u) d\bar \xi_u\bigg] \to  \E^{\bar\P}\bigg[\int_0^1 c(\hat r_{\bar s_u},\hat m_{\bar s_u},\bar X_u,\bar \xi_u) d\bar \xi_u\bigg],
    \]
    or, equivalently,
    \[
    \E^{\check\P}\bigg[\int_0^1 c(\hat r^n_{\bar s^n_u},\hat m^n_{\bar s^n_u},\bar X^n_u,\bar \xi^n_u) d\bar \xi^n_u\bigg]
    \to \E^{\check\P}\bigg[\int_0^1 c(\hat r_{\bar s_u},\hat m_{\bar s_u},\bar X_u,\bar \xi_u) d\bar \xi_u\bigg],
    \]
    }
    in view of \ref{assumption f,c continuous} and \ref{assumption c linear growth} it is enough to prove that for $\check\P$-almost all $\omega\in\check\Omega$
    \[
    \int_0^1 c(\hat r^n_{\bar s^n_u(\omega)},\hat m^n_{\bar s^n_u(\omega)},\bar X^n_u(\omega),\bar \xi^n_u(\omega)) d\bar \xi^n_u(\omega) \to \int_0^1 c(\hat r_{\bar s_u(\omega)},\hat m_{\bar s_u(\omega)},\bar X_u(\omega),\bar \xi_u(\omega)) d\bar \xi_u(\omega).
    \]
    To this end, we note that
    \[
        \int_0^1 c(\hat r_{\bar s_u(\omega)},\hat m_{\bar s_u(\omega)},\bar X_u(\omega),\bar \xi_u(\omega)) d\bar \xi^n_u(\omega) \to \int_0^1 c(\hat r_{\bar s_u(\omega)},\hat m_{\bar s_u(\omega)},\bar X_u(\omega),\bar \xi_u(\omega)) d\bar \xi_u(\omega)
    \]
    holds by the Portmanteau theorem. At the same time,
    \begin{equation}
    \begin{split}
        &\abs*{\int_0^1 c(\hat r^n_{\bar s^n_u(\omega)},\hat m^n_{\bar s^n_u(\omega)},\bar X^n_u(\omega),\bar \xi^n_u(\omega)) d\bar \xi^n_u(\omega) - \int_0^1 c(\hat r_{\bar s_u(\omega)},\hat m_{\bar s_u(\omega)},\bar X_u(\omega),\bar \xi_u(\omega)) d\bar \xi^n_u(\omega)}\\
        &\leq \int_0^1 \abs*{c(\hat r^n_{\bar s^n_u(\omega)},\hat m^n_{\bar s^n_u(\omega)},\bar X^n_u(\omega),\bar \xi^n_u(\omega)) - c(\hat r_{\bar s_u(\omega)},\hat m_{\bar s_u(\omega)},\bar X_u(\omega),\bar \xi_u(\omega))} d\bar \xi^n_u(\omega)\\
        &\leq \abs{\bar \xi^n_1 - \bar \xi^n_{0-}} \sup_{u\in [0,1]} \abs*{c(\hat r^n_{\bar s^n_u(\omega)},\hat m^n_{\bar s^n_u(\omega)},\bar X^n_u(\omega),\bar \xi^n_u(\omega)) - c(\hat r_{\bar s_u(\omega)},\hat m_{\bar s_u(\omega)},\bar X_u(\omega),\bar \xi_u(\omega))}.
    \end{split}
    \end{equation}
    Since the first factor is uniformly bounded the desired result follows from the fact that the second term vanishes since $c$ is uniformly continuous, due to assumption \ref{assumption f,c continuous}, together with the convergence \eqref{eq proposition 3.5 final parametrisation convergence}.
    
    Finally, a similar calculation using \ref{assumption f,g quadratic growth} in place of \ref{assumption c linear growth} shows that also
    \[
    \E^{\bar\P^n}\bigg[\int_0^1 f(\hat r_{\bar s_u},\hat m^n_{\bar s_u},\bar X_u,\bar \xi_u) d\hat r_{\bar s_u}\bigg] \to \E^{\bar\P}\bigg[\int_0^1 f(\hat r_{\bar s_u},\hat m_{\bar s_u},\bar X_u,\bar \xi_u) d\hat r_{\bar s_u}\bigg].
    \]
\end{proof}

The next theorem establishes an alternative representation of our reward function in terms of parametrisations. It shows that it is enough to approximate singular controls by continuous controls of bounded velocity and that the reward functional allows for a representation as a supremum over all parametrisations as motivated in \eqref{eq reward functional motivation supremum parametrisation}.

\begin{theorem}\label{theorem representation of J with parametrisations}
    Let $\P\in \c P(t,m)$ be a singular control, {then}
    \[
        J(t,m,\P) = {\sup_{\substack{\P^n\to \P\text{ in }\c P_2(D^0)\\(\P^n)_n\subseteq\c L(t,m)}} \limsup_{n\to\infty}} \ J_{{C}}(t,m,\P^n)
        = \sup_{\substack{((\hat m,\hat r),\bar\P)\text{ is a two-layer}\\\text{parametrisation of }\P\text{ on }[t,T]}} J_2((\hat m,\hat r),\bar\P).
    \]
    {In particular,
    \[
    V(t,m) = \sup_{\P\in\c P(t,m)} J(t,m,\P) = \sup_{\P\in\c L(t,m)} J(t,m,\P).
    \]
    }
\end{theorem}

\begin{proof}
    \begin{enumerate}[wide]
        \item \label{proof theorem representation of J with parametrisations part 1} Let $(\P^n)_n\subseteq\c C(t,m)$ be a sequence of continuous controls that converges to $\P\in\c P(t,m)$. By Theorem \ref{theorem construct limit parametrisation from convergent sequence} there exists a subsequence $(\P^{n_k})_{n_k}$ with two-layer parametrisations $((\hat m^{n_k},\hat r^{n_k}),\bar\P^{n_k})$ and a two-layer parametrisation $((\hat m,\hat r),\bar\P)$ of $\P$ such that
        \[
        ((\hat m^{n_k},\hat r^{n_k}),\bar\P^{n_k}) \to ((\hat m,\hat r),\bar\P).
        \]
        By the previous Lemma \ref{lemma reward function continuous}, this implies that
        \[
        J_2((\hat m^{n_k},\hat r^{n_k}),\bar\P^{n_k}) \to J_2((\hat m,\hat r),\bar\P).
        \]
        Now, since $(\P^{n_k})_k$ are continuous controls, we obtain from the definition that
        \[
        J_{{C}}(t,m,\P^{n_k}) = J_2((\hat m^{n_k},\hat r^{n_k}),\bar\P^{n_k}).
        \]
        Since the sequence $(\P^n)_n$ was arbitrary, we get
        \[
        J(t,m,\P) = {\sup_{\substack{\P^n\to \P\text{ in }\c P_2(D^0)\\(\P^n)_n\subseteq\c C(t,m)}} \limsup_{n\to\infty}} \ J_{{C}}(t,m,\P^n)
        \leq \sup_{\substack{((\hat m,\hat r),\bar\P)\text{ is a two-layer}\\\text{parametrisation of }\P\text{ on }[t,T]}} J_2((\hat m,\hat r),\bar\P).
        \]
        
        \item \label{proof theorem representation of J with parametrisations part 2} By {Lemma \ref{lemma existence two layer parametrisations}} there exists a two-layer parametrisation $((\hat m,\hat r),\bar\P)$ of $\P\in \c P(t,m)$. By Theorem \ref{theorem two layer parametrisation bounded velocity approximation}, there exists a sequence of continuous bounded velocity controls $(\P^n)_n\subseteq \c L(t,m)$ with corresponding two-layer parametrisations $((\hat m^n,\hat r^n),\bar\P^n)$ such that
        \[
        ((\hat m^n,\hat r^n),\bar\P^n) \to ((\hat m,\hat r),\bar\P).
        \]
        By the previous Lemma \ref{lemma reward function continuous}, this implies that
        \[
        J_2((\hat m^n,\hat r^n),\bar\P^n) \to J_2((\hat m,\hat r),\bar\P).
        \]
        Now, since $(\P^n)_n$ are continuous controls, we obtain from the definition that
        \[
        J_{{C}}(t,m,\P^n) = J_2((\hat m^n,\hat r^n),\bar\P^n),
        \]
        {which implies that
        \[
        J_2((\hat m,\hat r),\bar\P) = \lim_{n\to\infty} J_2((\hat m^n,\hat r^n),\bar\P^n) = \lim_{n\to\infty} J_{{C}}(t,m,\P^n).
        \]
        }
        Since the two-layer parametrisation $((\hat m,\hat r),\bar\P)$ was arbitrary, this implies
        \begin{equation}
        \begin{split}
        \sup_{\substack{((\hat m,\hat r),\bar\P)\text{ is a two-layer}\\\text{parametrisation of }\P\text{ on }[t,T]}} J_2((\hat m,\hat r),\bar\P)
        &\leq {\sup_{\substack{\P^n\to \P\text{ in }\c P_2(D^0)\\(\P^n)_n\subseteq\c L(t,m)}} \limsup_{n\to\infty}} \ J_{{C}}(t,m,\P^n)\\
        &\leq  { \sup_{\substack{\P^n\to \P\text{ in }\c P_2(D^0)\\(\P^n)_n\subseteq\c C(t,m)}} \limsup_{n\to\infty} \ J_{{C}}(t,m,\P^n) =} J(t,m,\P).
        \end{split}
        \end{equation}
        \item {By Parts \eqref{proof theorem representation of J with parametrisations part 1} and \eqref{proof theorem representation of J with parametrisations part 2}, we have for every admissible control $\P\in\c P(t,m)$,
        \[
        J(t,m,\P) = \sup_{\substack{\P^n\to \P\text{ in }\c P_2(D^0)\\(\P^n)_n\subseteq\c L(t,m)}} \limsup_{n\to\infty} J(t,m,\P^n)  \leq \sup_{\P\in\c L(t,m)} J(t,m,\P).
        \]
        Since $\c L(t,m)\subseteq\c P(t,m)$, this implies that
        \[
        V(t,m) = \sup_{\P\in\c P(t,m)} J(t,m,\P) = \sup_{\P\in\c L(t,m)} J(t,m,\P).
        \]
        }
    \end{enumerate}
\end{proof}

\subsection{Explicit representation of the reward functional}\label{subsection reward functional jump representation}

The representation of the reward function in terms of parametrisations {derived in Section \ref{subsection reward functional for parametrisations}} allows us to {now} obtain a more explicit representation in terms of minimal jump costs akin to the one dimensional case.

When considering different approximating sequences of a singular control $\P \in \c P(t,m)$, we are essentially thinking about different ways of interpolating discontinuities of the state-control process under $\P$. Within our mean-field setting there are two kinds of discontinuities of the control process $\xi$ (and thus of the state-control process $(X,\xi)$) that require different treatment.

\subsubsection{Discontinuities of the first kind}

The first kind of discontinuity occurs if $\xi_{u-}(\omega) \not= \xi_u(\omega)$ while $m_{u-} = m_u$. These discontinuities materialise only on a pathwise (or particlewise, adopting the point of view of e.g.\@ \cite{talbi2021dynamic}) level, {meaning they do not affect the overall measure flow. Thus, it is} natural to approximate discontinuities of the {first} kind {independently of other particles by keeping the measure flow fixed.}
 
 In the one-dimensional case $\xi \in \R$ it is thus natural to follow  \cite{min_singular_1987,zhu_generalized_1992,taksar_infinite-dimensional_1997,dufour_singular_2004,de2018stochastic} and to consider the following jump cost (resulting from linear interpolation):
\begin{equation}
\begin{split}
&C_{D^0}\big(u,m_u,X_{u-}(\omega),\xi_{u-}(\omega),\xi_u(\omega)\big)\\
&=  \int_0^{\xi_u(\omega)-\xi_{u-}(\omega)} c\big(u,m_u,X_{u-}(\omega) + \gamma(u)\zeta,\xi_{u-}(\omega)+\zeta\big) d\zeta.
\end{split}
\end{equation}

In the multi-dimensional case there might be multiple paths interpolating such a jump. We are interested in minimal jump costs.

\begin{definition}\label{definition c xi}
    Let $(t,m) \in [0,T]\times \c P_2$, $x\in \R^d$ and $\xi\leq \xi'$ component-wise for $\xi,\xi'\in\R^l$. Then we define
    \[
    C_{D^0}(t,m,x,\xi,\xi') \coloneqq \inf_{\zeta\in\Xi(\xi,\xi')} \int_0^1 c\big(t,m,x+\gamma(t)(\zeta_\lambda-\xi), \zeta_\lambda\big) d\zeta_\lambda,
    \]
    where {the integral is interpreted as the sum of component-wise Riemann-Stieltjes integrals and} 
    $\Xi(\xi,\xi')\subseteq C([0,1];\R^l)$ denotes the set of all continuous, non-decreasing paths
    \[
    \zeta: [0,1] \ni \lambda \mapsto \zeta_\lambda \in \R^l,\qquad \zeta_0 = \xi\text{ and }\zeta_1 = \xi'.
    \]
\end{definition}

The function $C_{D^0}$ yields the most cost efficient way of interpolating discontinuities on a pathwise level. To be able to incorporate the minimal jump costs into our two-layer parametrisations we require this function to be measurable. In what follows we establish the desired measurability and prove that the infimum is attained. In particular, we can choose minimising interpolating paths in a measurable way. We start with the following lemma.

\begin{lemma}\label{lemma measurable selection omega-wise jumps}
    We denote the domain of $C_{D^0}$ by
    \[
    \c M \coloneqq \bigl\{(t,m,x,\xi,\xi') \in [0,T]\times \c P_2\times\R^d\times\R^l\times\R^l \bigm\vert \xi \leq \xi'\bigr\}.
    \]
    Then there exists a measurable selection
    \[
    \pi:\c M \to C([0,1];\R^l)
    \]
    with $\pi(t,m,x,\xi,\xi') \in \Xi(\xi,\xi')$, which minimises the jump costs in Definition \ref{definition c xi}. That is, 
    \begin{equation}
    \begin{split}
    &C_{D^0}(t,m,x,\xi,\xi')\\
    &= \int_0^1 c\big(t,m,x + \gamma(t)(\pi(t,m,x,\xi,\xi')_\lambda - \xi), \pi(t,m,x,\xi,\xi')_\lambda\big) d\pi(t,m,x,\xi,\xi')_\lambda,
    \end{split}
    \end{equation}
    for all $(t,m,x,\xi,\xi') \in \c M$. In particular, $C_{D^0}$ is measurable.
\end{lemma}

\begin{proof}
    For this proof it will be convenient to equip $\R^l$ with the $\ell_1$-norm. This does not change the statement since the generated topology and $\sigma$-algebra do not change. 
        
    The key idea is that given $\xi\leq \xi'$, instead of minimising over the non-compact set $\Xi(\xi,\xi')$ in the definition of $C_{D^0}$, it is sufficient to minimise over a compact subset instead. To this end, we first introduce the subset of non-decreasing paths ``of constant speed'' (constant $L^1$-norm along the path) from $\xi$ to $\xi'$
    \begin{equation*}
    \begin{split}
    \Gamma(\xi,\xi') & ~ \coloneqq \bigl\{ h \in \Xi(\xi,\xi') \bigm\vert 
    |h_\lambda {-\xi}|_1 = \lambda |\xi' - \xi|_1 
    \text{ for all }\lambda\in [0,1] \bigr\} \\ & ~ \subseteq \Xi(\xi,\xi') \\ &~ \subseteq C([0,1];\R^l).
    \end{split}     
    \end{equation*}
    The set $\Gamma(\xi,\xi')$ is compact, due to the Theorem by Arzelà-Ascoli, since for all $\lambda,\rho\in [0,1]$,
    \[
    {\sup_{\lambda\in [0,1]} |h_\lambda|_1} \leq |\xi|_1 + |\xi'|_1 \quad \mbox{and} \quad |h_\lambda-h_\rho|_1 = \big| |h_\lambda-h_0|_1 - |h_\rho-h_0|_1 \big| = |\lambda-\rho| |\xi' - \xi|_1. 
    \]
    We can reparametrise every curve $h\in \Xi(\xi,\xi')$ to be in $\Gamma(\xi,\xi')$ by {defining a function $\phi:[0,1]\to [0,1]$ by $\phi(\lambda) \coloneqq \frac{|h_\lambda-\xi|_1}{|\xi-\xi'|_1}$} and then considering the path ${\bar h=h \circ \phi^{-1}}$. 
    {Since $$|\bar h_\lambda -\xi|_1 = |h_{\phi^{-1}(\lambda)}-\xi|_1 = \phi(\phi^{-1}(\lambda)) |\xi'-\xi|_1 = \lambda |\xi'-\xi|_1$$} 
    {we see that $\bar h \in \Gamma(\xi,\xi')$.} For such a reparametrisations it holds that
    \[
    \int_0^1 c(t,m,x+{\gamma(t)(h_\lambda-\xi)},h_\lambda) dh_\lambda = \int_0^1 c(t,m,x+{\gamma(t)({\bar h}_\lambda-\xi)},{\bar h}_\lambda) d {\bar h}_\lambda,
    \]
    and hence it follows that
    \[
    C_{D^0}(t,m,x,\xi,\xi') = {\inf}_{h\in\Gamma(\xi,\xi')} \int_0^1 c(t,m,x+{\gamma(t)(h_\lambda-\xi)},h_\lambda) dh_\lambda,
    \]
    {where we now take the infimum over the compact set $\Gamma(\xi,\xi')$. To show that the infimum is indeed attained, we need to show that the function
    \[
    \Gamma(\xi,\xi')\ni h\mapsto \int_0^1 c(t,m,x+\gamma(t)(h_\lambda - \xi),h_\lambda)dh_\lambda
    \]
    is continuous. To this end, we note that each $h\in\Gamma(\xi,\xi')$ is absolutely continuous and thus admits a.e.\@ a derivative $h'$. By the ``constant speed property'' $|h'_\lambda|_1 = |\xi'-\xi|_1$ for all $\lambda\in [0,1]$. Thus for each $h\in\Gamma(\xi,\xi')$ we have that
    \[
    \int_0^1 c(t,m,x+\gamma(t)(h_\lambda-\xi),h_\lambda) dh_\lambda
    = \int_0^1 c(t,m,x+\gamma(t)(h_\lambda-\xi), h_\lambda) h'_\lambda d\lambda.
    \]\\
    Let us now assume that we are given a sequence of functions $(h^n)_n\subseteq \Gamma(\xi,\xi')$ with $h^n\to h\in \Gamma(\xi,\xi')$ in $C([0,T];\R^l)$. Then the local uniform continuity of $c$ (Assumption \ref{assumption f,c continuous}) implies that 
    \[
        c(t,m,x+\gamma(t)(h^n),h^n) \to c(t,m,x+\gamma(t)(h),h) \quad \mbox{in} \quad 
        C([0,1];\R^{1\times d}) 
    \] 
        and hence also in $L^1$. At the same time, by Banach–Alaoglu theorem, the sequence $((h^n)')_n$ converges weakly* to $h'$ in $L^\infty$. Together, this implies the convergence of the integrals
    \[
    \int_0^1 c(t,m,x+\gamma(t)(h^n_\lambda-\xi), h^n_\lambda) (h^n)'_\lambda d\lambda \to \int_0^1 c(t,m,x+\gamma(t)(h_\lambda-\xi), h_\lambda) h'_\lambda d\lambda,
    \]
    and hence the desired continuity. As a result, the minimum in the Definition \ref{definition c xi} of $C_{D^0}$ is attained since $\Gamma(\xi,\xi')$ is compact.}
    
    To show the existence of a measurable selection function of minimisers, we note that by introducing
    \[
    \bar h \coloneqq \frac{h-\xi}{\xi'-\xi},
    \]
    where the division should be understood component-wise,
    we can further rewrite our problem as a minimisation problem over
    \[
    \Gamma({0_{\R^l},1_{\R^l}}) {= \bigl\{ \bar h\in \Xi(0_{\R^l},1_{\R^l}) \bigm\vert |\bar h_\lambda|_1 = \lambda l \text{ for all }\lambda\in [0,1]\bigr\}},
    \]
    {where $0_{\R^l}$ and $1_{\R^l}$ denotes the constant 0- and 1-vector in $\R^l$, respectively. We obtain that}
    \begin{equation}
    \begin{split}
    &C_{D^0}(t,m,x,\xi,\xi')\\
    &= \min_{\bar h\in \Gamma({0_{\R^l},1_{\R^l}})} \int_0^1 c\big(t,m,x+\gamma(t) ((\xi'-\xi) \odot \bar h_\lambda), \xi + ((\xi'-\xi)\odot \bar h_\lambda)\big)  d\big((\xi'-\xi)\odot \bar h_\lambda\big),
    \end{split}
    \end{equation}
    where $\odot$ denotes the component-wise multiplication of two vectors. Since the function
    \[
    (t,m,x,\xi,\xi',\bar h)\mapsto \int_0^1 c\big(t,m,x+\gamma(t) ((\xi'-\xi) \odot \bar h_\lambda), \xi + ((\xi'-\xi)\odot \bar h_\lambda)\big)  d\big((\xi'-\xi)\odot \bar h_\lambda\big)
    \]
    is measurable, we can use the Measurable Maximum Theorem (see \cite[{Theorem} 18.19]{aliprantis2006border}) to obtain a measurable selection function $\bar\pi:\c M\to \Gamma({0_{\R^l},1_{\R^l}})$ selecting a minimiser. Using $\bar\pi$, we can construct our desired measurable selection function $\pi$ via
    \[
    \pi:\c M\to C([0,1];\R^l),\qquad \pi(t,m,x,\xi,\xi') \coloneqq \xi + (\xi'-\xi)\odot \bar\pi(t,m,x,\xi,\xi'),
    \]
    with $\pi(t,m,x,\xi,\xi') \in \Gamma(\xi,\xi') \subseteq \Xi(\xi,\xi')$ for all $(t,m,x,\xi,\xi')\in \c M$.
\end{proof}

\subsubsection{Discontinuities of the second kind}

The second kind of discontinuities occurs if $\xi_{u-}(\omega) \not= \xi_u(\omega)$ with $m_{u-} \not= m_u$. 
Discontinuities of the second kind occur if a non-$m_{u-}$-negligible number of particles jump at time $u$. In this case the jump costs do not only depend on {\sl how} the jumps are executed -- which would corresponds to the cost $C_{D^0}$ -- but also in which order. To this end, {we} need to interpolate the mapping $u\mapsto \xi_u$ on a distributional level. This leads to the following definition, for which we denote the jump and continuation sets of a path $h$ on an interval $[u-,v]$ by
\[
J^d_{[u,v]}(h) \coloneqq \{r\in [u,v]\mid h_{r-} \not= h_r \} \quad \mbox{and} \quad J^c_{[u,v]}(h) \coloneqq \{r\in [u,v]\mid h_{r-} = h_r\},
\]
respectively. Depending on whether $h$ is a deterministic path like $u\mapsto m_u$ or a random sample path like $u\mapsto \xi_u$, these sets will also be deterministic or random.

\begin{definition}\label{definition c l2}
    Let $(t,m)\in [0,T]\times \c P_2$ and let $\P\in \c P(t,m)$ be a singular control. Let  $(Y,\zeta)$ be the canonical process on $D^0([0,1];\R^d\times\R^l)$. 
    For any $u\in [t,T]$ we set
    \[
    m_{u-,u} \coloneqq \P_{(X_{u-},\xi_{u-},X_u,\xi_u)}
    \]
    and define $\Xi(u,m_{u-,u})$ as the set of all probability laws $\mu\in \c P_2(D^0([0,1];\R^d\times\R^l))$ such that
    \begin{enumerate}[label=(\roman*)]
        \item $\mu_{(Y_0,\zeta_0,Y_1,\zeta_1)} = m_{u-,u}$,
        \item $[0,1]\ni\lambda\mapsto \zeta_\lambda \in\R^l$ is non-decreasing and càdlàg $\mu$-a.s.,
        \item $[0,1]\ni \lambda\mapsto \mu_\lambda \coloneqq \mu_{(Y_\lambda,\zeta_\lambda)} \in \c P_2$ is continuous,
        \item $Y_\lambda - Y_0 = \gamma(u)(\zeta_\lambda-\zeta_0)$ for all $\lambda\in [0,1]$, $\mu$-a.s.
    \end{enumerate}
    We call $\Xi(u,m_{u-,u})$ the set of all interpolating paths for the jump $m_{u-,u}$ at time $u$ and define the costs for such a jump as follows:
    \begin{equation}
    \begin{split}
        &C_{\c P_2} (u,m_{u-,u})\\
        &\coloneqq \inf_{\mu\in \Xi(u,m_{u-,u})} \E^{\mu}\bigg[\int_{J^c_{[0,1]}(\zeta)} c(u,\mu_\lambda,Y_\lambda,\zeta_\lambda) d\zeta_\lambda + \sum_{J^d_{[0,1]}(\zeta)} C_{D^0}(u,\mu_\lambda,Y_{\lambda-},\zeta_{\lambda-},\zeta_\lambda) \bigg],
    \end{split}
    \end{equation}
    {where $\E^\mu$ denotes the expectation under the probability measure $\mu$ on $D^0([0,1];\R^d\times\R^l)$.}
\end{definition}

\subsubsection{Representation using minimal jump costs}

Having defined the functions $C_{D^0}$ and $C_{\c P_2}$, we are now ready to state and prove the following alternative representation of our reward function. The result generalises the formula for the one-dimensional, non-mean-field case \eqref{eq reward functional one-dimensional non-mean-field case}. 
{Its derivation relies crucially on the next two lemmas, which establish bounds for the reward $J_2$ of two-layer parametrisations. The first, Lemma \ref{lemma reward functional upper bound}, provides an upper bound. The idea is to take an arbitrary two-layer parametrisation and decompose its cost $J_2$ according to the discontinuities of first and second kind studied in the previous subsections. 
For each type of discontinuities, we then argue that the cost accrued by the specific paths within the parametrisation cannot be lower than the defined minimal costs $C_{\c P_2}$ and $C_{D^0}$, finally leading to the explicit upper bound below.}

\begin{lemma}\label{lemma reward functional upper bound}
{
    For all $(t,m)\in [0,T]\times\c P_2$ and $\P\in\c P(t,m)$, we have
    \begin{equation}
        \begin{split}
            &\sup_{\substack{((\hat m,\hat r),\bar\P)\text{ is a two-layer}\\\text{parametrisation of }\P\text{ on }[t,T]}} J_2((\hat m,\hat r),\bar\P)\\
            &\leq \E^\P\bigg[\int_t^T f(u,m_u,X_u,\xi_u) du + g(m_T,X_T,\xi_T) - \sum_{J^d_{[t,T]}(m)} C_{\c P_2}(u,m_{u-,u})\\
            &\quad-\sum_{J^c_{[t,T]}(m)\cap J^d_{[t,T]}(\xi)} C_{D^0}(u,m_u,X_{u-},\xi_{u-},\xi_u) - \int_{J^c_{[t,T]}(m)\cap J^c_{[t,T]}(\xi)} c(u,m_u,X_u,\xi_u) d\xi_u\bigg].
        \end{split}
    \end{equation}
}
\end{lemma}
\begin{proof}
    Let $\P\in \c P(t,m)$ be a singular control and $((\hat m,\hat r),\bar \P)$ be a two-layer parametrisation of $\P$. We recall that
        \begin{equation}
        \begin{split}
            &J_2((\hat m,\hat r),\bar\P)\\
            &= \E^{\bar\P}\bigg[\int_0^1 f(\hat r_{\bar s_u},\hat m_{\bar s_u},\bar X_u,\bar\xi_u) d\hat r_{\bar s_u} + g(\hat m_{\bar s_1},\bar X_1,\bar\xi_1) - \int_0^1 c(\hat r_{\bar s_u},\hat m_{\bar s_u},\bar X_u,\bar \xi_u)d\bar\xi_u \bigg].
        \end{split}
        \end{equation}
        The running and terminal payoff can be rewritten as an expectation of $(\tilde X,\tilde\xi)$, as
        \begin{equation}
        \label{eq theorem 3.15 spliting the rewards into the different layers}
        \begin{split}
            &\E^{\bar\P}\bigg[\int_0^1 f(\hat r_{\bar s_u},\hat m_{\bar s_u},\bar X_u,\bar\xi_u) d\hat r_{\bar s_u} + g(\hat m_{\bar s_1},\bar X_1,\bar\xi_1)\bigg]\\
            &= \E^{\bar\P}\bigg[\int_t^T f(\tau,m_\tau,\tilde X_\tau,\tilde \xi_\tau) d\tau + g(m_T,\tilde X_T,\tilde \xi_T)\bigg].
        \end{split}
        \end{equation}
        
        Rewriting the singular control term is slightly more involved as the time scales $\hat r$ and $\bar s$ may not be injective and therefore not invertible. Hence, one time point on the $(\hat X,\hat\xi)$-time scale may correspond to a whole time interval on the $(\bar X,\bar\xi)$-time scale.
        {We note that by the construction of $s$ as right-continuous inverse of $\bar s$, we have $\bar\P$-a.s.\@ that
        \[
        s_{\bar s_v-} \leq v \leq s_{\bar s_v},\quad\text{for all }v\in [0,1].
        \]
        Thus the monotonicity of $\bar\xi$ together with \eqref{eq subsubsection second layer inverse time change layer transformation} implies that, $\bar\P$-a.s.\@,
        \[
        \hat\xi_{\bar s_v-}=\bar\xi_{s_{\bar s_v-}}\leq\bar\xi_v \leq \bar\xi_{s_{\bar s_v}} = \hat\xi_{\bar s_v},\quad\text{for all }v\in [0,1],
        \]
        which in particular shows that $\bar\xi_v = \hat\xi_{\bar s_v}$ if $\bar s_v \in J^c_{[0,1]}(\hat\xi)$. This motivates us to rewrite the singular control cost by splitting the integral as follows:
        }
        
        \begin{equation}
        \begin{split}
            &\E^{\bar\P}\bigg[\int_0^1 c(\hat r_{\bar s_u},\hat m_{\bar s_u},\bar X_u,\bar \xi_u)d\bar\xi_u \bigg]\\
            &= {\E^{\bar\P}\bigg[\int_{\{u \mid \bar s_u \in J^c_{[0,1]}(\hat\xi)\}} c(\hat r_{\bar s_u},\hat m_{\bar s_u},\bar X_u,\bar \xi_u)d\bar\xi_u + \int_{\{u\mid \bar s_u\in J^d_{[0,1]}(\hat\xi)\}} c(\hat r_{\bar s_u},\hat m_{\bar s_u},\bar X_u,\bar \xi_u)d\bar\xi_u \bigg]}\\
            &= {\E^{\bar\P}\bigg[\int_{\{u \mid \bar s_u \in J^c_{[0,1]}(\hat\xi)\}} c(\hat r_{\bar s_u},\hat m_{\bar s_u},\hat X_{\bar s_u},\hat\xi_{\bar s_u})d\hat\xi_{\bar s_u}}\\
            &\quad {+ \int_{\{u \mid u\in [s_{v-},s_v], v\in J^d_{[0,1]}(\hat\xi)\}} }
            {c(\hat r_{\bar s_u},\hat m_{\bar s_u},\bar X_u,\bar \xi_u)d\bar\xi_u \bigg]}\\
            &= \E^{\bar\P}\bigg[\int_{J^c_{[0,1]}(\hat \xi)} c(\hat r_v,\hat m_v,\hat X_v,\hat\xi_v) d\hat\xi_v + \sum_{v\in J^d_{[0,1]}(\hat\xi)} \int_{s_{v-}}^{s_v} c(\hat r_{\bar s_u},\hat m_{\bar s_u},\bar X_u,\bar\xi_u) d\bar\xi_u\bigg].
        \end{split}
        \end{equation}
        Applying the same idea to the time change introduced by $\hat r$, we arrive at
        \begin{equation}
        \label{eq theorem 3.15 spliting the costs into the different layers}
        \begin{split}
            &\E^{\bar\P}\bigg[\int_0^1 c(\hat r_{\bar s_u},\hat m_{\bar s_u},\bar X_u,\bar \xi_u)d\bar\xi_u \bigg]\\
            &= \E^{\bar\P}\bigg[ \int_{J^c_{[t,T]}(m) \cap J^c_{[t,T]}(\tilde\xi)} c(\tau,m_\tau,\tilde X_\tau, \tilde \xi_\tau) d\tilde\xi_\tau + \sum_{\tau\in J^c_{[t,T]}(m)\cap J^d_{[t,T]}(\tilde\xi)}\int_{s_{r_\tau-}}^{s_{r_\tau}} c(\tau,m_\tau,\bar X_u,\bar \xi_u) d\bar\xi_u\\
            &\quad + \sum_{{z}\in J^d_{[t,T]}(m)} \bigg(\int_{J^c_{[r_{{z}-},r_{z}]}(\hat\xi)} c({z},\hat m_v,\hat X_v,\hat \xi_v) d\hat\xi_v + \sum_{v\in J^d_{[r_{{z}-},r_{z}]}(\hat\xi)} \int_{s_{v-}}^{s_v} c({z},\hat m_v,\bar X_u,\bar\xi_u) d\bar\xi_u \bigg)\bigg].
        \end{split}
        \end{equation}
        {We emphasize that that in the equation above, the time points $\tau \in J^c_{[t,T]}(m)\cap J^c_{[t,T]}(\tilde\xi)$ and $\tau \in J^c_{[t,T]}(m)\cap J^d_{[t,T]}(\tilde\xi)$ are random, while the jump times $z\in J^d_{[t,T]}(m)$ are deterministic.}
        
        The resulting representation is already similar to the desired form \eqref{eq reward functional alternative form}. In fact, the first term above is equal to the the last summand in \eqref{eq reward functional alternative form}:
        \[
            \E^{\bar\P}\bigg[\int_{J^c_{[t,T]}(m) \cap J^c_{[t,T]}(\tilde\xi)} c(\tau,m_\tau,\tilde X_\tau, \tilde \xi_\tau) d\tilde\xi_\tau\bigg]
            = \E^\P\bigg[\int_{J^c_{[t,T]}(m)\cap J^c_{[t,T]}(\xi)} c(u,m_u,X_u,\xi_u) d\xi_u \bigg].
        \]
        
        It remains to consider the remaining two terms; we start with the latter one. For any $v\in J^d_{[0,1]}(\hat\xi)$ the path 
        \[
        [s_{v-},s_v]\ni u \mapsto (\bar X_u,\bar\xi_u)\in\R^d\times\R^l
        \]
        interpolates continuously from $(\hat X_{v-},\hat\xi_{v-}) = (\bar X_{s_{v-}},\bar\xi_{s_{v-}})$ to $(\hat X_v,\hat\xi_v) = (\bar X_{s_v},\bar\xi_{s_v})$ as defined in Definition \ref{definition c xi}. Thus, for all ${z} \in [0,1]$, $v\in J^d_{[0,1]}(\hat\xi)$, we have that
        \[
        \int_{s_{v-}}^{s_v} c({z},\hat m_v, \bar X_u,\bar\xi_u) d\bar\xi_u \geq C_{D^0}({z},\hat m_v, \bar X_{s_{v-}},\bar\xi_{s_{v-}},\bar\xi_{s_v}) = C_{D^0}({z},\hat m_v,\hat X_{v-},\hat\xi_{v-},\hat\xi_v).
        \]
        This implies that
        \begin{equation}
        \begin{split}
            &\E^{\bar\P}\bigg[ \sum_{{z}\in J^d_{[t,T]}(m)} \bigg(\int_{J^c_{[r_{{z}-},r_{z}]}(\hat\xi)} c({z},\hat m_v,\hat X_v,\hat \xi_v) d\hat\xi_v + \sum_{v\in J^d_{[r_{{z}-},r_{z}]}(\hat\xi)} \int_{s_{v-}}^{s_v} c({z},\hat m_v,\bar X_u,\bar\xi_u) d\bar\xi_u \bigg) \bigg]\\
            &\geq \sum_{{z}\in J^d_{[t,T]}(m)} \E^{\bar\P}\bigg[\int_{J^c_{[r_{{z}-},r_{z}]}(\hat\xi)} c({z},\hat m_v,\hat X_v,\hat \xi_v) d\hat\xi_v + \sum_{v\in J^d_{[r_{{z}-},r_{z}]}(\hat\xi)} C_{D^0}({z},\hat m_v,\hat X_{v-},\hat\xi_{v-},\hat\xi_v)\bigg].
        \end{split}
        \end{equation}
        Moreover, for all ${z}\in J^d_{[t,T]}(m)$, the path 
        \[
        [r_{{z}-},r_{z}] \ni v \mapsto (\hat X_v,\hat\xi_v)\in \R^d\times\R^l
        \]
        interpolates between $(\tilde X_{{z}-},\tilde\xi_{{z}-})$ and $(\tilde X_{z},\tilde\xi_{z})$ in a càdlàg fashion, such that $\hat\xi$ is non-decreasing. Furthermore, by the definition of $\c P_2$-parametrisations the measure flow $u\mapsto \hat m_u = \bar\P_{(\hat X_u,\hat\xi_u)} \in \c P_2$ is continuous. Thus, the path $\bar\P_{(\hat X,\hat\xi)_{[r_{{z}-},r_{z}]}}$ interpolates the jump $m_{{z}-,{z}} = \bar\P_{(\tilde X_{{z}-},\tilde\xi_{{z}-},\tilde X_{z},\tilde \xi_{z})}$ in the sense of Definition \ref{definition c l2} and we have that
        \begin{equation}
        \begin{split}
            &\sum_{{z}\in J^d_{[t,T]}(m)} \E^{\bar\P}\bigg[\int_{J^c_{[r_{{z}-},r_{z}]}(\hat\xi)} c({z},\hat m_v,\hat X_v,\hat \xi_v) d\hat\xi_v + \sum_{v\in J^d_{[r_{{z}-},r_{z}]}(\hat\xi)} C_{D^0}({z},\hat m_v,\hat X_{v-},\hat\xi_{v-},\hat\xi_v)\bigg]\\
            &\geq \sum_{{z}\in J^d_{[t,T]}(m)} C_{\c P_2}({z},m_{{z}-,{z}}).
        \end{split}
        \end{equation}
        
        Let us now consider the last term in \eqref{eq theorem 3.15 spliting the costs into the different layers}. As before that for $\tau\in J^c_{[t,T]}(m)\cap J^d_{[t,T]}(\tilde\xi)$ the path 
        \[
        [s_{r_\tau-},s_{r_\tau}] \ni v \mapsto (\bar X_v,\bar\xi_v)\in \R^d\times\R^l
        \]
        is an interpolating path from $(\bar X_{s_{r_\tau-}},\bar\xi_{s_{r_\tau-}}) = (\tilde X_{\tau-},\tilde\xi_{\tau-})$ to $(\bar X_{s_{r_\tau}},\bar\xi_{s_{r_\tau}}) = (\tilde X_\tau,\tilde\xi_\tau)$ in the sense of Definition \ref{definition c xi}. Thus we again obtain for all $\tau\in J^c_{[t,T]}(m)\cap J^d_{[t,T]}(\tilde\xi)$ that
        \begin{equation}
        \begin{split}
            \int_{s_{r_\tau-}}^{s_{r_\tau}} c(\tau,m_\tau,\bar X_u,\bar \xi_u) d\bar\xi_u \geq C_{D^0}(\tau,m_\tau,\tilde X_{\tau-},\tilde \xi_{\tau-},\tilde\xi_\tau),
        \end{split}
        \end{equation}
        which implies that
        \begin{equation}
        \begin{split}
            &\E^{\bar\P}\bigg[\sum_{\tau\in J^c_{[t,T]}(m)\cap J^d_{[t,T]}(\tilde\xi)}\int_{s_{r_\tau-}}^{s_{r_\tau}} c(\tau,m_\tau,\bar X_u,\bar \xi_u) d\bar\xi_u \bigg]\\
            &\geq \E^{\bar\P}\bigg[\sum_{\tau\in J^c_{[t,T]}(m)\cap J^d_{[t,T]}(\tilde\xi)} C_{D^0}(\tau,m_\tau,\tilde X_{\tau-},\tilde \xi_{\tau-},\tilde\xi_\tau)\bigg]\\
            &= \E^{\P}\bigg[\sum_{J^c_{[t,T]}(m)\cap J^d_{[t,T]}(\xi)} C_{D^0}(u,m_u,X_{u-},\xi_{u-},\xi_u)\bigg].
        \end{split}
        \end{equation}
    
        Summarising our estimates, we obtain that
        \begin{equation}
        \begin{split}
            &J_2((\hat m,\hat r),\bar\P)\\
            &\leq \E^\P\bigg[\int_t^T f(u,m_u,X_u,\xi_u) du + g(m_T,X_T,\xi_T) - \sum_{J^d_{[t,T]}(m)} C_{\c P_2}(u,m_{u-,u})\\
            &\ -\sum_{J^c_{[t,T]}(m)\cap J^d_{[t,T]}(\xi)} C_{D^0}(u,m_u,X_{u-},\xi_{u-},\xi_u) - \int_{J^c_{[t,T]}(m)\cap J^c_{[t,T]}(\xi)} c(u,m_u,X_u,\xi_u) d\xi_u\bigg].
        \end{split}
        \end{equation}
        Since the the two-layer parametrisation $((\hat m,\hat r),\bar\P)$ was arbitrary, we obtain {the claim by taking the supremum over all two-layer parametrisations of $\P$}.
\end{proof}

{Lemma \ref{lemma reward functional upper bound} provides the first inequality needed for Theorem \ref{theorem reward functional alternative form}. We now establish the converse inequality in Lemma \ref{lemma reward functional lower bound} by constructing specific $\varepsilon$-optimal parametrisations. This involves carefully inserting $\varepsilon$-optimal interpolating paths (related to $C_{\c P_2}$) for distributional jumps and optimal interpolating paths (related to $C_{D^0}$) for pathwise jumps into the two-layer framework.}

\begin{lemma}\label{lemma reward functional lower bound}
{
    For all $(t,m)\in [0,T]\times\c P_2$ and $\P\in\c P(t,m)$, we have
    \begin{equation}
        \begin{split}
            &\sup_{\substack{((\hat m,\hat r),\bar\P)\text{ is a two-layer}\\\text{parametrisation of }\P\text{ on }[t,T]}} J_2((\hat m,\hat r),\bar\P)\\
            &\geq \E^\P\bigg[\int_t^T f(u,m_u,X_u,\xi_u) du + g(m_T,X_T,\xi_T) - \sum_{J^d_{[t,T]}(m)} C_{\c P_2}(u,m_{u-,u})\\
            &\quad-\sum_{J^c_{[t,T]}(m)\cap J^d_{[t,T]}(\xi)} C_{D^0}(u,m_u,X_{u-},\xi_{u-},\xi_u) - \int_{J^c_{[t,T]}(m)\cap J^c_{[t,T]}(\xi)} c(u,m_u,X_u,\xi_u) d\xi_u\bigg].
        \end{split}
    \end{equation}
}
\end{lemma}
\begin{proof}
    {Let} $\P\in \c P(t,m)$ be a singular control. {To prove the claim, it} is sufficient to construct for each $\varepsilon > 0$ a two-layer parametrisation $((\hat m^\varepsilon,\hat r^\varepsilon),\bar\P^\varepsilon)$ of $\P$ satisfying
        \begin{equation}
        \begin{split}
        &\E^\P\bigg[\int_t^T f(u,m_u,X_u,\xi_u) du + g(m_T,X_T,\xi_T) - \sum_{J^d_{[t,T]}(m)} C_{\c P_2}(u,m_{u-,u})\\
        &\quad-\sum_{J^c_{[t,T]}(m)\cap J^d_{[t,T]}(\xi)} C_{D^0}(u,m_u,X_{u-},\xi_{u-},\xi_u) - \int_{J^c_{[t,T]}(m)\cap J^c_{[t,T]}(\xi)} c(u,m_u,X_u,\xi_u) d\xi_u\bigg]\\
        &\leq J_2((\hat m^\varepsilon,\hat r^\varepsilon),\bar\P^\varepsilon) + \varepsilon.
        \end{split}
        \end{equation}
        
        Let us hence fix $\varepsilon > 0$ and omit any dependence of parametrisations on $\varepsilon$ what follows. We start with a weak solution $\tilde\P$ to the SDE \eqref{eq diffusion x} on the canonical space $(\tilde\Omega,\tilde{\c F})$ equipped with the canonical process $(\tilde X,\tilde\xi,\tilde W)$ corresponding to the singular control $\P$. That is, $\tilde\P_{(\tilde X,\tilde\xi)} = \P_{(X,\xi)}$. We use the weak solution to construct a two-layer parametrisation in two steps. In the first step, we construct the $\c P_2$-layer, obtaining a $\c P_2$-parametrisation of $\P$. Subsequently, we construct the $\omega$-wise layer. 

        \begin{enumerate}[wide]
        \item[\textbf{ Step 1. \emph{The $\c P_2$-layer.}}]        
        
        We construct a $\c P_2$-parametrisation from the weak solution $\tilde\P$ to the SDE \eqref{eq diffusion x}. The parametrisation differs from the weak solution as an additional time change is applied to ensure the continuity of the measure flow $u\mapsto m_u$. Thus, we focus on the {\sl deterministic} jump times $u\in J^d_{[t,T]}(m)$ of the measure flow. Our idea is to insert suitable $\c P_2$-continuous interpolations at the jump times using a suitable time change. 
        
        The set of jump times $J^d_{[t,T]}(m)$, and thus the number of jumps we need to interpolate is countable since the process $(\tilde X,\tilde\xi)$ is càdlàg. Let $({{t}}_n)_{n\in\N}$ be a deterministic enumeration of $J^d_{[t,T]}(m)$,\footnote{{Since $(\tilde X,\tilde\xi)$ is càdlàg, by \cite[Lemma 3.7.7, p. 131]{ethier_markov_1986} the set
        \[
        J^d_{[t,T]}(m) = \{s \in [t,T] \mid m_{s-} \not= m_s\} \subseteq \{s \in [t,T] \mid \tilde\P((X_{s-},\xi_{s-}) = (X_s,\xi_s)) < 1\}
        \]
        is at most countable.}} say by ordering the jumps by jump size in a decreasing order. For each $n\in\N$ we consider the jump from $m_{{{t}}_n-}$ to $m_{{{t}}_n}$.
        {We first note that by \eqref{eq diffusion x},
        \[
        \tilde X_{{{t}}_n} = \tilde X_{{{t}}_n-} + \gamma({{t}}_n) (\tilde\xi_{{{t}}_n} - \tilde\xi_{{{t}}_n-}).
        \]
        Thus, after defining the (pathwise) linear interpolating processes
        \[
        \tilde Y_\lambda^n \coloneqq \tilde X_{{{t}}_n-} + \lambda (\tilde X_{{{t}}_n} - \tilde X_{{{t}}_n-}),\qquad \tilde \zeta_\lambda^n \coloneqq \tilde\xi_{{{t}}_n-} + \lambda (\tilde\xi_{{{t}}_n} - \tilde\xi_{{{t}}_n-}),\qquad \lambda\in [0,1],
        \]
        their joint law $\tilde\P_{(\tilde Y^n,\tilde\zeta^n)}$ is an interpolating path for the jump from $m_{t_{n-}}$ to $m_{t_n}$ in the sense of Definition \ref{definition c l2}.\footnote{{We can also construct a $\c P_2$-parametrisation of $\P$ using these linear interpolations; this is the construction outlined in the proof sketch of Lemma \ref{lemma existence W2 parametrisations}. The construction here is more general, using non-linear, $\varepsilon$-optimal interpolations.}} 
        In particular, the minimal jump costs $C_{\c P_2}({{t}}_n, m_{{{t}}_n-,{{t}}_n})$ in Definition \ref{definition c l2} are well-defined and we can find a $(2^{-(n+1)}\varepsilon)$-optimal interpolating path, by which we mean
        }
        an interpolating path $\mu^n$ on the space $D^0([0,1];\R^d\times\R^l)$ equipped with the canonical process $(Y,\zeta)$, such that
        \[
        \mu^n_{(Y_0,\zeta_0,Y_1,\zeta_1)} = m_{{{t}}_n-,{{t}}_n} = \tilde\P_{(\tilde X_{{{t}}_n-},\tilde\xi_{{{t}}_n-},\tilde X_{{{t}}_n},\tilde\xi_{{{t}}_n})},
        \]
        and, using the notation $\mu^n_\lambda \coloneqq \mu^n_{(Y_\lambda,\zeta_\lambda)}$,
        \begin{equation}
        \label{eq theorem 3.15 l2 epsilon optimal interpolation}
        \begin{split}
        &C_{\c P_2}({{t}}_n,m_{{{t}}_n-,{{t}}_n})\\
        &\geq \E^{\mu^n}\bigg[\int_{J^c_{[0,1]}(\zeta)} c({{t}}_n,\mu^n_\lambda,Y_\lambda,\zeta_\lambda) d\zeta_\lambda + \sum_{J^d_{[0,1]}(\zeta)} C_{D^0}({{t}}_n,\mu^n_\lambda,Y_{\lambda-},\zeta_{\lambda-},\zeta_\lambda) \bigg] - 2^{-(n+1)} \varepsilon.
        \end{split}
        \end{equation}
        
        Our next goal is to integrate the measure flows $(\mu_n)_n$ into our process $(\tilde X,\tilde\xi)$. 
        {Since the interpolating measure flows $(\mu_n)_n$ are each defined on separate probability spaces, we first construct a joint probability space supporting both the process $(\tilde X,\tilde\xi)$ and all the interpolating paths.} 
        To this end, we use the fact that $D^0([0,1];\R^d\times\R^l)$ is a complete separable metric space and construct the regular conditional probability of $(Y,\zeta)$ under $\mu^n$ given the random variables $(Y_0, \zeta_0, Y_1, \zeta_1)$ as the stochastic kernel
        \[
        K^n:(\R^d\times \R^l\times\R^d\times\R^l)\times \c B(D^0([0,1];\R^d\times\R^l))\to [0,1]
        \]
        such that
        \[
        \mu^n = \mu^n_{(Y_0,\zeta_0,Y_1,\zeta_1)} K^n = \tilde\P_{(\tilde X_{{{t}}_n-},\tilde\xi_{{{t}}_n-},\tilde X_{{{t}}_n},\tilde\xi_{{{t}}_n})} K^n.
        \]
        Next, {we} introduce the stochastic kernels $\tilde K^n$ from $(\tilde\Omega,\tilde{\c F})$ to the space $(D^0([0,1];\R^d\times\R^l),\allowbreak \c B(D^0([0,1];\R^d\times\R^l)))$ given as
        {
        \begin{equation}\label{eq theorem 3.15 construction tilde Kn}
        \tilde K^n(\tilde\omega,\,\cdot\,) \coloneqq K^n \bigl((\tilde X_{{{t}}_n-}(\tilde\omega),\tilde\xi_{{{t}}_n-}(\tilde\omega),\tilde X_{{{t}}_n}(\tilde\omega),\tilde\xi_{{{t}}_n}(\tilde\omega)), \,\cdot\, \bigr),\quad \text{for all }\tilde\omega\in\tilde\Omega.
        \end{equation}
        For each $\tilde\omega\in\tilde\Omega$ these kernels select paths in $D^0([0,1];\R^d\times\R^l)$ interpolating the jump from $(\tilde X_{{{t}}_n-}(\tilde\omega),\tilde\xi_{{{t}}_n-}(\tilde\omega))$ to $(\tilde X_{{{t}}_n}(\tilde\omega),\tilde\xi_{{{t}}_n}(\tilde\omega))$ consistent with the interpolating measure flow $\mu^n$.}

        {
        Using the probability measure $\tilde\P$ and the kernels $(\tilde K^n)_n$, we now define our joint probability space.} By the Ionescu-Tulcea theorem, there exists a measure $\breve\P$ on
        \[
        (\breve\Omega,\breve{\c F}) \coloneqq \bigg(\tilde\Omega \times \prod_{n=0}^\infty D^0([0,1];\R^d\times\R^l), \tilde{\c F}\otimes \bigotimes_{n=0}^\infty \c B\big(D^0([0,1];\R^d\times\R^l)\big)\bigg),
        \]
        such that for all $N\in\N$ and all $A\in \tilde{\c F} \otimes \bigotimes_{n=0}^N \c B\big(D^0([0,1];\R^d\times\R^l)\big)$,
        \begin{equation}\label{eq theorem 3.15 construction bar P epsilon}
        \bigg(\tilde\P\otimes \bigotimes_{n=0}^N \tilde K^n \bigg) (A) = \breve\P\bigg(A\times \prod_{n=N+1}^\infty D^0([0,1];\R^d\times\R^l) \bigg).
        \end{equation}
        
        We assume w.l.o.g.\@ that $(\breve\Omega,\breve{\c F},\breve\P)$ is complete; else we consider the completion. With a slight abuse of notation, we use $(\tilde X,\tilde\xi,\tilde W),(Y^n,\zeta^n)_n$ to denote the canonical process on the space $(\breve\Omega,\breve{\c F})$. Then, 
        \begin{equation}\label{eq theorem 3.15 definition canonical extensions}
        (\tilde X,\tilde \xi,\tilde W)(\breve\omega) \coloneqq (\tilde X,\tilde\xi,\tilde W)(\tilde\omega),\qquad (Y^n,\zeta^n)(\breve\omega) \coloneqq (Y,\zeta)(\omega^n),
        \end{equation}
        where $\big(\tilde\omega,(\omega^n)_{n\in\N}\big) \coloneqq \breve\omega\in\breve\Omega$.
        Furthermore, by construction of our joint probability space using the kernels $\tilde K^n$, we have for all $n\in\N$ that
        \[
        (\tilde X_{{{t}}_n-},\tilde\xi_{{{t}}_n-}) = (Y^n_0,\zeta^n_0),\qquad (\tilde X_{{{t}}_n},\tilde \xi_{{{t}}_n}) = (Y^n_1,\zeta^n_1)\qquad\breve\P\text{-a.s}.
        \]
        
        {
        We remark that while the measure flow $u\mapsto m_u$ is discontinuous at the (deterministic) times $({{t}}_n)_n$, not every particle needs to jump at ${{t}}_n$. In particular, for all $\breve\omega\in\breve\Omega$ with $(\tilde X_{{{t}}_n-},\tilde\xi_{{{t}}_n-}) = (\tilde X_{{{t}}_n},\xi_{{{t}}_n})$, the interpolating path $\lambda\mapsto (Y^n_\lambda, \zeta^n_\lambda) \coloneqq (\tilde X_{{{t}}_n-},\tilde\xi_{{{t}}_n-})$ is constant.
        }
        
        On the space $(\breve\Omega,\breve{\c F},\breve\P)$ we can now define our $\c P_2$-layer process $(\hat X,\hat\xi)$ by inserting the interpolations for each jump. To this end, we define for all $v\in [t-,T]$
        \[
        r_v \coloneqq  \frac{(v-t) + \norm{{\tilde\xi}_v - {\tilde\xi}_{t-}}_{L^2(\breve\P{)}}}{(T-t) + \norm{{\tilde\xi}_T-{\tilde\xi}_{t-}}_{L^2(\breve\P{)}}}
        \]
        and let $\hat r_u \coloneqq \inf \{v\in {[t,T]}\mid r_v > u\}\land {T}$ be its inverse. 
        {By construction, $r_{v-} = r_v$ if and only if $v \in J^c_{[t,T]}(m)$. We
        recall that $J^d_{[t,T]}(m) = (t_n)_{n \in \mathbb{N}}$ and define the process $(\hat X,\hat\xi)$ as follows:}
        \begin{equation}\label{eq theorem 3.15 definition hat xi}
        \hat\xi_u \coloneqq \tilde\xi_{\hat r_u -} + \sum_{n=0}^\infty \1_{[r_{{{t}}_n-},r_{{{t}}_n}]}(u) \Big(\zeta^n_{\frac{u - r_{{{t}}_n-}}{r_{{{t}}_n} - r_{{{t}}_n-}}} -  \tilde\xi_{{{t}}_n-}\Big),\qquad\text{for }u \in [0,1],
        \end{equation}
        and
        \begin{equation}\label{eq theorem 3.15 definition hat x}
        \begin{split}
        \hat X_u &\coloneqq \tilde X_{\hat r_u -} + \sum_{n=0}^\infty \1_{[r_{{{t}}_n-},r_{{{t}}_n}]}(u) \Big(Y^n_{\frac{u - r_{{{t}}_n-}}{r_{{{t}}_n} - r_{{{t}}_n-}}} -  \tilde X_{{{t}}_n-}\Big)\\
        &= {\tilde X_{\hat r_u -} + \sum_{n=0}^\infty \1_{[r_{{{t}}_n-},r_{{{t}}_n}]}(u) \, \gamma(t_n)\Big(\zeta^n_{\frac{u - r_{{{t}}_n-}}{r_{{{t}}_n} - r_{{{t}}_n-}}} -  \tilde \xi_{{{t}}_n-}\Big)},\qquad\text{for }u \in [0,1],
        \end{split}
        \end{equation}
        {recalling that $Y^n_u = Y^n_0 + \gamma(t_n) (\zeta^n_u - \zeta^n_0)$ by Definition \ref{definition c l2}.
        This construction ensures that $(\hat{X}_{r_v}, \hat{\xi}_{r_v}) = (\tilde{X}_v, \tilde{\xi}_v)$ for $v \in J^c_{[t,T]}(m)$ (i.e., when there is no jump in the measure flow, the processes coincide).  For $u \in [r_{t_n-}, r_{t_n}]$, the processes $(\hat{X}, \hat{\xi})$ follow the interpolated paths $(Y^n, \zeta^n)$, where, crucially, $(Y^n_0, \zeta^n_0) = (\tilde{X}_{t_n-}, \tilde{\xi}_{t_n-})$ and $(Y^n_1, \zeta^n_1) = (\tilde{X}_{t_n}, \tilde{\xi}_{t_n})$.}

        {
        From \eqref{eq theorem 3.15 definition hat xi} and \eqref{eq theorem 3.15 definition hat x}, we observe that for all $u\in [0,1]$,
        \begin{equation}
        \label{eq theorem 3.15 hat x hat xi}
            \hat X_u = \tilde X_{\hat r_u -} + \1_{\{r_{u-} \not= r_u\}} \gamma(\hat r_u) (\hat\xi_u-\tilde\xi_{\hat r_u -})
            = \tilde X_{\hat r_u -} + \gamma(\hat r_u) (\hat\xi_u-\tilde\xi_{\hat r_u -}).
        \end{equation}}
        Since $(\tilde X,\tilde\xi,\tilde W)$ satisfies the SDE \eqref{eq diffusion x}, {with \eqref{eq theorem 3.15 hat x hat xi} we see} that $(\hat X,\hat\xi,\hat r,\tilde W)$ satisfies the time changed SDE \eqref{eq weak solution W2 time changed SDE} {on $(\breve\Omega,\breve{\c F},\breve\P)$}.
        To show that $(\breve\P_{(\hat X,\hat\xi,\breve W)},\hat r)$ is a weak solution to \eqref{eq weak solution W2 time changed SDE}, it remains to verify that $\tilde W$ is also a $\hat{\b F}$-Brownian motion, where $\hat{\c F}_u \coloneqq \c F^{\tilde W}_u\lor \c F^{\hat X,\hat\xi}_{r_u}$. Again we use that $\tilde\P$ is already a weak solution to the SDE \eqref{eq diffusion x}, since this implies that $\tilde W$ is an $\b F^{\tilde W,\tilde X,\tilde\xi}$-Brownian motion, which implies that for all $t\leq u\leq v\leq T$,
        \begin{equation}\label{eq theorem 3.15 independence tilde W}
        \tilde W_v - \tilde W_u \indep (\tilde W_l,\tilde X_l,\tilde\xi_l)_{t-\leq l\leq u}.
        \end{equation}
        Since we constructed the measure $\breve\P$ using the kernels $\tilde K^n$, each of which each only depending on $m_{{{t}}_{n-},{{t}}_n}$, see \eqref{eq theorem 3.15 construction tilde Kn}, we see that the canonical processes $(Y^n,\zeta^n)$ from \eqref{eq theorem 3.15 definition canonical extensions} satisfy for all $n\in\N$ and ${{t}}_n\leq u\leq T$ that
        \begin{equation}\label{eq theorem 3.15 independence hat W}
        \tilde W_u - \tilde W_{{{t}}_n} \indep Y^n,\zeta^n.
        \end{equation}
        Finally, the relations \eqref{eq theorem 3.15 independence tilde W}, \eqref{eq theorem 3.15 independence hat W} together imply that the processes constructed in \eqref{eq theorem 3.15 definition hat xi} and \eqref{eq theorem 3.15 definition hat x} satisfy for all $t\leq u\leq v\leq T$ that
        \[
        \tilde W_v - \tilde W_u \indep (\tilde X_{\hat r_l},\tilde\xi_{\hat r_l})_{0-\leq l\leq r_u},(Y^n,\zeta^n)_{n\in\N,{{t}}_n\leq u},
        \]
        and thus
        \[
        \tilde W_v - \tilde W_u \indep (\hat X_l,\hat\xi_l)_{0-\leq l\leq r_u}.
        \]
        Therefore $\tilde W$ is an $\hat{\b F}$-Brownian motion under $\breve\P$, where $\hat{\c F}_u \coloneqq \c F^{\tilde W}_u\lor \c F^{\hat X,\hat\xi}_{r_u}$, and thus $(\breve\P_{(\hat X,\hat\xi,\tilde W)},\hat r)$ is a weak solution to \eqref{eq weak solution W2 time changed SDE}.
        This implies that $(\breve\P_{(\hat X,\hat\xi)},\hat r)$ is indeed a $\c P_2$-parametrisation of $\P$.
        
        \item[\textbf{ Step 2. \emph{The pathwise layer.}}]
        
        Now that we have constructed the $\c P_2$-layer, we will continue with constructing the $\breve\omega$-wise layer.
        Our goal is to insert an optimally interpolating paths into the first-layer process
        \[
        [0,1]\times\breve\Omega \ni (u,\breve\omega) \mapsto (\hat X_u(\breve\omega),\hat\xi_u(\breve\omega),\hat r_u) \in \R^d\times\R^l\times [t,T].
        \]
        As these interpolations take place on an $\breve\omega$-wise level, extra caution is required to ensure that the resulting process remains measurable.
        
        We start by defining the new time scale $\bar s$.
        For this, we introduce for each $\breve\omega\in\breve\Omega$ {the strictly monotone} mapping
        \begin{equation}
        \label{eq theorem 3.15 definition of inverse time change s}
        s_v(\breve\omega) \coloneqq \frac{v + \arctan(\Var(\hat\xi(\breve\omega), [0,v]))}{1+\frac\pi 2},\qquad \text{for }u\in[0,1],
        \end{equation}
        where $\Var$ denotes the total variation. The process $s$ is measurable and, more importantly, $\b F^{\hat\xi}$-adapted since $s_v$ only depends on information about the process $\hat\xi$ up to the current time $v$. This will guarantee that the original Brownian motion $\tilde\Omega$ remains a Brownian motion for this second layer after inserting the interpolations. 
        We also introduce the inverse, for $u\in [0,1]$,
        \[
        \bar s_u(\breve\omega) \coloneqq \inf\{v\in [0,1]\mid s_v(\breve\omega) > u\} \land 1 = \inf\{v\in [0,1]\cap \Q\mid s_v(\breve\omega) > u \} \land 1,
        \]
        which is then also measurable as a pointwise limit of measurable functions and also $(\c F^{s}_{\bar s_u})_u$-adapted,%
        \footnote{{$(\c F^s_u)_u$ denotes the natural filtration generated by the process $s$ defined in \eqref{eq theorem 3.15 definition of inverse time change s}.}} 
        since $s$ is càdlàg, which implies that $\bar s$ is also $(\c F^{\hat\xi}_{\bar s_u})_u$-adapted.

        It remains to construct the process $(\bar X,\bar\xi)$. To this end, we note that the set of discontinuities
        \[
        \{(u,\breve\omega)\in [0,1]\times\breve\Omega\mid \hat\xi_{u-}(\breve\omega) \not= \hat\xi_u(\breve\omega) \} = \{(u,\breve\omega)\in [0,1]\times\breve\Omega \mid s_{u-}(\breve\omega) \not= s_u(\breve\omega)\}
        \]
        is measurable and since $\hat\xi$ is $\breve\P$-a.s.\@ càdlàg, the set of jump times
        \[
        \{u\in [0,1] \mid \hat\xi_{u-}(\breve\omega) \not= \hat\xi_u(\breve\omega)\} = \{u\in [0,1] \mid s_{u-}(\breve\omega) \not= s_u(\breve\omega)\}
        \]
        is countable for $\breve\P$-almost every $\breve\omega\in\breve\Omega$. Since $(\breve\Omega,\breve{\c F},\breve\P)$ is complete we can enumerate the jump times $(\hat\tau_n(\breve\omega))_{n\in\N}$ in a measurable fashion, by which we mean that $\hat\tau_n:\breve\Omega\to [0,1]$ is measurable for each $n\in\N$; see \cite[Proposition 1.32]{jacod2013limit} for the case that $\hat\xi(\breve\omega)$ is càdlàg, while defining $\hat\tau_n(\breve\omega) \equiv 0$ on the $\breve\P$-negligible set where $\hat\xi(\breve\omega)$ may not be càdlàg.
        
        For each $n\in\N$ we now interpolate the jump from $\hat\xi_{\hat\tau_n(\breve\omega)-}(\breve\omega)$ to $\hat\xi_{\hat\tau_n(\breve\omega)}(\breve\omega)$ optimally for each $\breve\omega$ such that the interpolation is still measurable. This is where Lemma \ref{lemma measurable selection omega-wise jumps} is used. Using the measurable selection function $\pi:\c M\to C([0,1];\R^l)$ introduced therein, we construct $\bar \xi$ by inserting the interpolating paths, suitably adjusted for the time interval, into $\hat \xi$ as follows, for $u\in [0,1]$,
        \[
            \bar\xi_u \coloneqq \hat\xi_{\bar s_u-} + \sum_{n=0}^\infty \1_{[s_{\hat\tau_n-},s_{\hat\tau_n}]}(u) \Big(\pi(\hat r_{\hat\tau_n}, \hat m_{\hat\tau_n},\hat X_{\hat\tau_n-},\hat\xi_{\hat\tau_n-},\hat\xi_{\hat\tau_n})_{\frac{u - s_{\hat\tau_n-}}{s_{\hat\tau_n} - s_{\hat\tau_n-}}} - \hat\xi_{\hat\tau_n-}\Big),
        \]
        and correspondingly $\bar X$ as, for $u\in [0,1]$,
        \[
            \bar X_u \coloneqq \hat X_{\bar s_u-} + \sum_{n=0}^\infty \1_{[s_{\hat\tau_n-},s_{\hat\tau_n}]}(u) \, \gamma(\hat r_{\hat\tau_n}) \Big(\pi(\hat r_{\hat\tau_n}, \hat m_{\hat\tau_n},\hat X_{\hat\tau_n-},\hat\xi_{\hat\tau_n-},\hat\xi_{\hat\tau_n})_{\frac{u - s_{\hat\tau_n-}}{s_{\hat\tau_n} - s_{\hat\tau_n-}}} - \hat\xi_{\hat\tau_n-}\Big),
        \]
        {which makes sure that $\bar X$ and $\bar\xi$ satisfy \eqref{eq two layer parametrisations X hat, X bar given xi hat, xi bar}.}

        {This construction guarantees that for $u \in [s_{\hat{\tau}_n-}, s_{\hat{\tau}_n}]$, the process $\bar{\xi}$ follows the path selected by $\pi$, that means
        \[
        \bar\xi_u = \pi^n_{\frac{u-s_{\hat{\tau}_n-}}{s_{\hat{\tau}_n}-s_{\hat{\tau}_n-}}} \coloneqq \pi(\hat r_{\hat\tau_n},\hat m_{\hat\tau_n},\hat X_{\hat\tau_n-},\hat\xi_{\hat\tau_n-},\hat\xi_{\hat\tau_n})_{\frac{u-s_{\hat{\tau}_n-}}{s_{\hat{\tau}_n}-s_{\hat{\tau}_n-}}},
        \]
        which, by Lemma \ref{lemma measurable selection omega-wise jumps}, is an optimal interpolation between $\hat{\xi}_{\hat{\tau}_n-}$ and $\hat{\xi}_{\hat{\tau}_n}$, which means that}
        \begin{equation}
        \label{eq theorem 3.15 omega optimal interpolation}
        \begin{split}
        C_{D^0}(\hat r_{\hat\tau_n},\hat m_{\hat\tau_n},\hat X_{\hat\tau_n-},\hat\xi_{\hat\tau_n-},\hat\xi_{\hat\tau_n})
        &= {\int_0^1 c\big(\hat r_{\hat\tau_n},\hat m_{\hat\tau_n},\hat X_{\hat\tau_n-} + \gamma(\hat r_{\hat\tau_n})(\pi^n_\lambda - \xi), \pi^n_\lambda\big) d\pi^n_\lambda}\\
        &= \int_{s_{\hat\tau_n-}}^{s_{\hat\tau_n}} c(\hat r_{\hat\tau_n},\hat m_{\hat\tau_n},\bar X_u,\bar\xi_u) d\bar\xi_u,
        \end{split}
        \end{equation}
        
        We note that $(\bar X,\bar\xi)$ is measurable since it is the pointwise limit of compositions of measurable functions.
        {Similar to the previous layer, since} $(\hat X,\hat\xi,\tilde W,\hat r)$ satisfies the SDE \eqref{eq weak solution W2 time changed SDE}, this construction implies that $(\bar X,\bar\xi,\tilde W,\bar s,\hat r)$ also satisfies the time changed SDE \eqref{eq weak solution two-layer time changed SDE}.
        To show that $(\breve\P_{(\bar X,\bar\xi,\tilde W,\bar s)},\hat r)$ is a weak solution to \eqref{eq weak solution two-layer time changed SDE}, it hence remains to verify that $\tilde W$ is also an $\bar{\b F}$-Brownian motion, where $\bar {\c F}_u\coloneqq \c F^{\tilde W}_u\lor \c F^{\bar X,\bar\xi,\bar s}_{s_{r_u}}$. 
        
        We have already shown that $(\breve\P_{(\hat X,\hat\xi,\tilde W)},\hat r)$ is a weak solution to \eqref{eq weak solution W2 time changed SDE}, which implies that for all $\hat\tau_n\leq v\leq T$
        \[
        \tilde W_v - \tilde W_{\hat\tau_n} \indep \hat r_{\hat\tau_n},\hat m_{\hat\tau_n},\hat X_{\hat\tau_n-},\hat\xi_{\hat\tau_n-},\hat\xi_{\hat\tau_n}.
        \]
        Since we constructed the above interpolations selected by $\pi$ in a measurable fashion only depending on $(\hat r_{\hat\tau_n},\hat m_{\hat\tau_n},\hat X_{\hat\tau_n-},\hat\xi_{\hat\tau_n-},\hat\xi_{\hat\tau_n})$, this ensures that
        \[
        \tilde W_v - \tilde W_{\hat\tau_n} \indep \pi(\hat r_{\hat\tau_n},\hat m_{\hat\tau_n},\hat X_{\hat\tau_n-},\hat\xi_{\hat\tau_n-},\hat\xi_{\hat\tau_n}).
        \]
        Using that $\bar s$ is $(\c F^{\hat\xi}_{\bar s_u})_u$-adapted and that $\tilde W$ is an $\hat{\b F}$-Brownian motion where $\hat{\c F}_u \coloneqq \c F^{\tilde W}_u\lor \c F^{\hat X,\hat\xi}_{r_u}$, we now again see that for all $u\in [0,1]$ and $\hat r_{\bar s_u}\leq v\leq T$
        \[
        \tilde W_v - \tilde W_{\hat r_{\bar s_u}} \indep (\bar X_l,\bar\xi_l,\bar s_l)_{0\leq l\leq u}.
        \]
        Therefore $\tilde W$ is also an $\bar {\b F}$-Brownian motion with $\bar {\c F}_u\coloneqq \c F^{\tilde W}_u\lor \c F^{\bar X,\bar\xi,\bar s}_{s_{r_u}}$. Hence, from our construction it follows that $(\breve\P_{(\bar X,\bar\xi,\tilde W,\bar s)},\hat r)$ is a weak solution to \eqref{eq weak solution two-layer time changed SDE}.
        Hence $((\hat m,\hat r),\breve\P_{(\bar X,\bar\xi,\bar s)})$ is a two-layer parametrisation of $\P$.
        
        \item[\textbf{ Step 3. \emph{The reward functional.}}]
        
        Turning to the reward functional for the parametrisation $((\hat m,\hat r),\breve\P_{(\bar X,\bar\xi,\bar s)})$, we note that by construction, following the same calculation steps as in \eqref{eq theorem 3.15 spliting the costs into the different layers}, {we have}
        \begin{equation}
            \label{eq theorem 3.15 reward functional epsilon optimal parametrisation}
        \begin{split}
            &J_2((\hat m,\hat r),\breve\P_{(\bar X,\bar\xi,\bar s)})\\
            &= \E^{\breve\P}\bigg[\int_t^T f({{u}},m_{{u}},\tilde X_{{u}},\tilde \xi_{{u}}) d{{u}} + g(m_T,\tilde X_T,\tilde \xi_T) - \int_{J^c_{[t,T]}(m) \cap J^c_{[t,T]}(\tilde \xi)} c({{u}},m_{{u}},\tilde X_{{u}}, \tilde \xi_{{u}}) d\tilde \xi_{{u}}\\
            &\quad - \sum_{{{z}}\in J^d_{[t,T]}(m)} \bigg(\int_{J^c_{[r_{{{z}}-},r_{{z}}]}(\hat\xi)} c({{z}},\hat m_v,\hat X_v,\hat \xi_v) d\hat\xi_v + \sum_{v \in J^d_{[r_{{{z}}-},r_{{z}}]}(\hat\xi)} \int_{s_{v-}}^{s_v} c({{z}},\hat m_v,\bar X_u,\bar\xi_u) d\bar\xi_u \bigg)\\
            &\quad - \sum_{{{v}} \in J^c_{[t,T]}(m)\cap J^d_{[t,T]}(\tilde \xi)} \int_{s_{r_{{v}}-}}^{s_{r_{{v}}}} c({{v}},m_{{v}},\bar X_u,\bar \xi_u) d\bar\xi_u \bigg].
        \end{split}
        \end{equation}
    
        Since the first three terms in \eqref{eq theorem 3.15 reward functional epsilon optimal parametrisation} are already in a suitable form, we will at first focus on the last two terms.
        For the second last term, we note that by \eqref{eq theorem 3.15 omega optimal interpolation}, $\breve\P$-a.s.,
        \begin{equation}
        \begin{split}
            &\sum_{{{z}}\in J^d_{[t,T]}(m)} \bigg(\int_{J^c_{[r_{{{z}}-},r_{{z}}]}(\hat\xi)} c({{z}},\hat m_v,\hat X_v,\hat \xi_v) d\hat\xi_v + \sum_{v \in J^d_{[r_{{{z}}-},r_{{z}}]}(\hat\xi)} \int_{s_{v-}}^{s_v} c({{z}},\hat m_v,\bar X_u,\bar\xi_u) d\bar\xi_u \bigg)\\
            &=\!\! \sum_{{{z}}\in J^d_{[t,T]}(m)} \bigg(\int_{J^c_{[r_{{{z}}-},r_{{z}}]}(\hat\xi)} c({{z}},\hat m_v,\hat X_v,\hat \xi_v) d\hat\xi_v + \sum_{v \in J^d_{[r_{{{z}}-},r_{{z}}]}(\hat\xi)} C_{D^0}({{z}},\hat m_{{v}}, \hat X_{{{v}}-},\hat\xi_{{{v}}-},\hat\xi_{{v}}) \bigg).
        \end{split}
        \end{equation}
        Further since $J^d_{[t,T]}(m)$ contains exactly the jump times $({{t}}_n)_n$ and since during the jump interval $[r_{{{t}}_n-},r_{{{t}}_n}]$ by construction $(\hat X,\hat\xi)$ is given by $(Y^n,\zeta^n)$, by which we mean that
        \[
        \hat \xi_u = \zeta^n_{\frac{u - r_{{{t}}_n-}}{r_{{{t}}_n} - r_{{{t}}_n-}}},\qquad \hat X_u = Y^n_{\frac{u - r_{{{t}}_n-}}{r_{{{t}}_n} - r_{{{t}}_n-}}}\qquad\text{for }u\in [r_{{{t}}_n-},r_{{{t}}_n}],
        \]
        we can continue rewriting this term as follows
        \begin{equation}
        \begin{split}
            &\sum_{{{z}}\in J^d_{[t,T]}(m)} \bigg(\int_{J^c_{[r_{{{z}}-},r_{{z}}]}(\hat\xi)} c({{z}},\hat m_v,\hat X_v,\hat \xi_v) d\hat\xi_v + \sum_{v \in J^d_{[r_{{{z}}-},r_{{z}}]}(\hat\xi)} C_{D^0}({{z}},\hat m_{{v}}, \hat X_{{{v}}-},\hat\xi_{{{v}}-},\hat\xi_{{v}}) \bigg)\\
            &=   \sum_{n=0}^\infty \bigg(\int_{J^c_{[0,1]}(\zeta^n)} c({{t}}_n,\mu^n_\lambda,Y^n_\lambda,\zeta^n_\lambda) d\zeta^n_\lambda + \sum_{J^d_{[0,1]}(\zeta^n)} C_{D^0}({{t}}_n,\mu^n_\lambda,Y^n_{\lambda-},\zeta^n_{\lambda-},\zeta^n_\lambda) \bigg).
        \end{split}
        \end{equation}
        Hence from \eqref{eq theorem 3.15 l2 epsilon optimal interpolation} and dominated convergence we get
        \begin{equation}
        \begin{split}
            &\E^{\breve\P}\bigg[\sum_{{{z}}\in J^d_{[t,T]}(m)} \bigg(\int_{J^c_{[r_{{{z}}-},r_{{z}}]}(\hat\xi)} c({{z}},\hat m_v,\hat X_v,\hat \xi_v) d\hat\xi_v + \sum_{v \in J^d_{[r_{{{z}}-},r_{{z}}]}(\hat\xi)} \int_{s_{v-}}^{s_v} c({{z}},\hat m_v,\bar X_u,\bar\xi_u) d\bar\xi_u \bigg)\bigg]\\
            &= \sum_{n=0}^\infty \E^{\mu^n}\bigg[\int_{J^c_{[0,1]}(\zeta^n)} c({{t}}_n,\mu^n_\lambda,Y^n_\lambda,\zeta^n_\lambda) d\zeta^n_\lambda + \sum_{J^d_{[0,1]}(\zeta^n)} C_{D^0}({{t}}_n,\mu^n_\lambda,Y^n_{\lambda-},\zeta^n_{\lambda-},\zeta^n_\lambda) \bigg]\\
            &\leq \sum_{n=0}^\infty C_{\c P_2}({{t}}_n,m_{{{t}}_n-,{{t}}_n})+ \varepsilon = \sum_{u\in J^d_{[t,T]}(m)} C_{\c P_2}(u,m_{u-,u}) + \varepsilon.
        \end{split}
        \end{equation}
                
        For the last term in \eqref{eq theorem 3.15 reward functional epsilon optimal parametrisation}, note that $J^c_{[t,T]}(m)\cap J^d_{[t,T]}(\tilde \xi)$ is countable and hence, using \eqref{eq theorem 3.15 omega optimal interpolation}, we see that $\breve\P$-a.s.,
        \begin{equation}
        \begin{split}
            &\sum_{{{v}} \in J^c_{[t,T]}(m)\cap J^d_{[t,T]}(\tilde \xi)} \int_{s_{r_{{v}}-}}^{s_{r_{{v}}}} c({{v}},m_{{v}},\bar X_u,\bar \xi_u) d\bar\xi_u\\
            &= \sum_{{{v}} \in J^c_{[t,T]}(m)\cap J^d_{[t,T]}(\tilde \xi)} C_{D^0}({{v}},m_{{v}},\bar X_{s_{r_{{v}}-}},\bar \xi_{s_{r_{{v}}-}},\bar\xi_{s_{r_{{v}}}})\\
            &= \sum_{{{v}} \in J^c_{[t,T]}(m)\cap J^d_{[t,T]}(\tilde \xi)} C_{D^0}({{v}},m_{{v}},\tilde X_{{{v}}-},\tilde \xi_{{{v}}-},\tilde \xi_{{v}}).
        \end{split}
        \end{equation}
        
        Putting these estimates together, we obtain
        \begin{equation}
        \begin{split}
            &J_2((\hat m,\hat r),\breve\P_{(\bar X,\bar \xi,\bar s)})\\
            &\geq \E^{\breve\P}\bigg[\int_t^T f(u,m_u,\tilde X_u,\tilde \xi_u) du + g(m_T,\tilde X_T,\tilde \xi_T) - \sum_{J^d_{[t,T]}(m)} C_{\c P_2}(u,m_{u-,u})\\
            &-\!\!\!\!\!\sum_{J^c_{[t,T]}(m)\cap J^d_{[t,T]}(\tilde \xi)} C_{D^0}(u,m_u,\tilde X_{u-},\tilde \xi_{u-},\tilde \xi_u) - \int_{J^c_{[t,T]}(m)\cap J^c_{[t,T]}(\tilde \xi)} c(u,m_u,\tilde X_u,\tilde \xi_u) d\tilde \xi_u\bigg]\! -\! \varepsilon\\
            &= \E^\P\bigg[\int_t^T f(u,m_u,X_u,\xi_u) du + g(m_T,X_T,\xi_T) - \sum_{J^d_{[t,T]}(m)} C_{\c P_2}(u,m_{u-,u})\\
            &-\!\!\!\!\!\sum_{J^c_{[t,T]}(m)\cap J^d_{[t,T]}(\xi)} C_{D^0}(u,m_u,X_{u-},\xi_{u-},\xi_u) - \int_{J^c_{[t,T]}(m)\cap J^c_{[t,T]}(\xi)} c(u,m_u,X_u,\xi_u) d\xi_u\bigg]\! -\! \varepsilon.
        \end{split}
        \end{equation}
        Hence by letting $\varepsilon\to 0$, {the claim follows}.
    \end{enumerate}
\end{proof}

{Combining the previous inequalities with the characterisation of the reward functional $J$ from Theorem \ref{theorem representation of J with parametrisations} yields the desired explicit representation.}

\begin{theorem}\label{theorem reward functional alternative form}
    The reward functional $J$ defined in \eqref{eq reward function general definition} allows for the following alternative representation: for all $(t,m)\in [0,T]\times \c P_2$, $\P\in\c P(t,m)$,
        \begin{equation}
        \label{eq reward functional alternative form}
        \begin{split}
        &J(t,m,\P)\\
        &= \E^\P\bigg[\int_t^T f(u,m_u,X_u,\xi_u) du + g(m_T,X_T,\xi_T) - \sum_{J^d_{[t,T]}(m)} C_{\c P_2}(u,m_{u-,u})\\
        &\quad-\sum_{J^c_{[t,T]}(m)\cap J^d_{[t,T]}(\xi)} C_{D^0}(u,m_u,X_{u-},\xi_{u-},\xi_u) - \int_{J^c_{[t,T]}(m)\cap J^c_{[t,T]}(\xi)} c(u,m_u,X_u,\xi_u) d\xi_u\bigg].
        \end{split}
        \end{equation}
\end{theorem}

{
\begin{proof}
    This follows directly from Theorem \ref{theorem representation of J with parametrisations} and Lemmas \ref{lemma reward functional upper bound} and \ref{lemma reward functional lower bound}.
\end{proof}}

\section{{The value function}}\label{section applications}

{This section applies the framework and the explicit representation of the reward functional developed in the previous sections to analyse the value function of the mean-field singular control problem. In Section \ref{section dpp}, we establish a dynamic programming principle. Subsequently, in Section \ref{section qvi}, under additional regularity assumptions, we characterise the value function as the smallest viscosity supersolution to a quasi-variational inequality (QVI) on the Wasserstein space.}

\subsection{The dynamic programming principle}\label{section dpp}
In this subsection, we use the previously proven representation of the reward functional to develop a DPP for our mean-field control setting with singular controls. Since we work with the weak formulation, and thus with the canonical filtration generated by the process $(X,\xi)$ itself, and since we restrict ourselves to a DPP for deterministic times, the proof is rather simple and does not require any measurable selection arguments. In fact, the proof is very similar to the ones in \cite{talbi2021dynamic,djete2019mckean} since the extension to singular controls does not impose any additional mathematical difficulties for the proof of this DPP. Just for this section, we extend our shorthand notation for the measure flow to $m^\P_t = \P_{(X_t,\xi_t)}$ to make it explicit at any point which underlying probability measure we are currently considering.

\begin{theorem}\label{theorem dpp}
    Let $(t,m)\in [0,T]\times \c P_2$. For all $s\in [t,T]$, we have the following dynamic programming principle
    \begin{equation}
    \begin{split}
        &V(t,m)\\
        &= \sup_{\P \in \c P(t,m)} {\bigg( V(s,m^\P_{s-}) + \E^\P\bigg[}\int_t^s f(u,m^\P_u,X_u,\xi_u) du - \sum_{J^d_{[t,s)}(m^\P)} C_{\c P_2}(u,m^\P_{u-,u})\\
        &\ -\sum_{J^c_{[t,s)}(m^\P)\cap J^d_{[t,s)}(\xi)} C_{D^0}(u,m^\P_u,X_{u-},\xi_{u-},\xi_u)
        - \int_{J^c_{[t,s)}(m^\P)\cap J^c_{[t,s)}(\xi)} c(u,m^\P_u,X_u,\xi_u) d\xi_u\bigg]{\bigg)}.
    \end{split}
    \end{equation}
\end{theorem}

\begin{proof}
\begin{enumerate}[wide]
    \item 
    We start by letting $\P \in \c P(t,m)$. Then $\P_{(X,\xi)_{[s-,T]}}\in \c P(s,m^\P_{s-})$ and hence
    \[
    J(s,m^\P_{s-},\P_{(X,\xi)_{[s-,T]}}) \leq V(s,m^\P_{s-}).
    \]
    Now by using Theorem \ref{theorem reward functional alternative form} and then taking the supremum over all $\P\in\c P(t,m)$, we get
    \begin{equation}
    \begin{split}
        &V(t,m)\\
        &= \sup_{\P \in \c P(t,m)} J(t,m,\P)\\
        &\leq \sup_{\P \in \c P(t,m)} {\bigg(V(s,m^\P_{s-}) + \E^\P\bigg[}\int_t^s f(u,m^\P_u,X_u,\xi_u) du - \sum_{J^d_{[t,s)}(m^\P)} C_{\c P_2}(u,m^\P_{u-,u})\\
        &\  -\sum_{J^c_{[t,s)}(m^\P)\cap J^d_{[t,s)}(\xi)} C_{D^0}(u,m^\P_u,X_{u-},\xi_{u-},\xi_u)
        - \int_{J^c_{[t,s)}(m^\P)\cap J^c_{[t,s)}(\xi)} c(u,m^\P_u,X_u,\xi_u) d\xi_u\bigg]{\bigg)}.
    \end{split}
    \end{equation}
    
    \item
    For the other direction, let $\P \in \c P(t,m)$ and $\b Q\in \c P(s,m^\P_{s-})$. Similar to the standard approach, we now want to concatenate these two probability measures together at the time point $s-$. More precisely, our goal is constructing a probability measure $\breve\P \in \c P(t,m)$ such that
    \[
    \breve\P_{(X,\xi)_{[0-,s)}} = \P_{(X,\xi)_{[0-,s)}}\qquad\text{and}\qquad \breve\P_{(X,\xi)_{[s-,T]}} = \b Q.
    \]
    If we assume, that we would have already constructed such a concatenated probability measure, we can use it to proceed as follows
    \begin{equation}
    \begin{split}
        &V(t,m)\\
        &\geq J(t,m,\breve\P)\\
        &= J(s,m^{\breve\P}_{s-},\breve\P_{(\tilde X,\tilde\xi)_{[s-,T]}}) + \E^{\breve\P}\bigg[\int_t^s f(u,m^{\breve\P}_u,\tilde X_u,\tilde \xi_u) du - \sum_{J^d_{[t,s)}(m^{\breve\P})} C_{\c P_2}(u,m^{\breve\P}_{u-,u})\\
        &\  -\sum_{J^c_{[t,s)}(m^{\breve\P})\cap J^d_{[t,s)}(\tilde \xi)} C_{D^0}(u,m^{\breve\P}_u,\tilde X_{u-},\tilde \xi_{u-},\tilde \xi_u)
        - \int_{J^c_{[t,s)}(m^{\breve\P})\cap J^c_{[t,s)}(\tilde \xi)} c(u,m^{\breve\P}_u,\tilde X_u,\tilde\xi_u) d\tilde \xi_u\bigg]\\
        &= J(s,m^{\P}_{s-},\b Q) + \E^{\P}\bigg[\int_t^s f(u,m^{\P}_u,X_u,\xi_u) du - \sum_{J^d_{[t,s)}(m^{\P})} C_{\c P_2}(u,m^\P_{u-,u})\\
        &\  -\sum_{J^c_{[t,s)}(m^{\P})\cap J^d_{[t,s)}(\xi)} C_{D^0}(u,m^{\P}_u,X_{u-},\xi_{u-},\xi_u)
        - \int_{J^c_{[t,s)}(m^{\P})\cap J^c_{[t,s)}(\xi)} c(u,m^{\P}_u,X_u,\xi_u) d\xi_u\bigg].
    \end{split}
    \end{equation}
    And the claim follows by taking the supremum first over all $\b Q\in \c P(s,m^\P_{s-})$ and then over all $\P\in \c P(t,m)$.
    
    So it remains to construct such a concatenated probability measure $\breve\P$. Our main tool will be the transfer argument from \cite[Theorem 6.10]{kallenberg2002foundations}. To this end, let $\tilde\P\in\c P_2(\tilde\Omega)$ be a weak solution to the SDE \eqref{eq diffusion x} corresponding to the control $\P$ and similarly $\tilde\Q\in\c P_2(\tilde\Omega)$ for the control $\Q$. As in the Definition \ref{definition weak solution sde}, we equip the complete, separable space $(\tilde\Omega,\tilde{\c F})$ with the canonical process $(\tilde X,\tilde\xi,\tilde W)$. We now want to apply the transfer theorem to the random variables $(\tilde X_{s-},\tilde\xi_{s-})$ and $(\tilde X,\tilde\xi,\tilde W)$ under $\Q$ to transfer them to the probability space $(\tilde\Omega,\tilde{\c F},\tilde\P)$. By considering a suitable extension instead if necessary, we can assume that $(\tilde\Omega,\tilde{\c F},\tilde\P)$ supports an $\c U[0,1]$-random variable $U$ which is independent of $(\tilde X,\tilde\xi)$.
    Now since
    \[
    \tilde\P_{(\tilde X_{s-},\tilde\xi_{s-})} = m^\P_{s-} = \tilde\Q_{(\tilde X_{s-},\tilde\xi_{s-})},
    \]
    we can use the transfer theorem from \cite[Theorem 6.10]{kallenberg2002foundations} to obtain a measurable function $\psi:\R^d\times\R^l\times [0,1]\to D^0\times C([0,T];\R^m)$ such that the transferred random variables
    \[
    (\tilde X^\Q,\tilde\xi^\Q,\tilde W^\Q) \coloneqq \psi(\tilde X_{s-},\tilde \xi_{s-},U)
    \]
    satisfy
    \[
    \tilde\P_{((\tilde X_{s-},\tilde\xi_{s-}), (\tilde X^\Q,\tilde \xi^\Q,\tilde W^\Q))} = \tilde\Q_{((\tilde X_{s-},\tilde\xi_{s-}),(\tilde X,\tilde\xi,\tilde W))}.
    \]
    So we can construct the concatenated processes
    \[
    (\breve X_r,\breve \xi_r,\breve W_r) \coloneqq \begin{cases}
        (\tilde X_r,\tilde \xi_r,\tilde W_r){,}& r < s{,}\\
        (\tilde X^\Q_r,\tilde \xi^\Q_r,W^\P_{s-} + (\tilde W^\Q_r - \tilde W^\Q_{s-})){,}&r \geq s.
    \end{cases}
    \]
    Then due to the concatenation, $(\breve X,\breve \xi,\breve W)$ satisfies the SDE \eqref{eq diffusion x} on $[t-,T]$ and due to the transfer argument, $\breve W$ is also an $\b F^{\breve X,\breve\xi,\breve W}$-Brownian motion on $[t,T]$. Hence $\tilde\P_{(\breve X,\breve \xi,\breve W)}$ is a weak solution to the SDE \eqref{eq diffusion x} on $[t-,T]$. Therefore, by defining
    \[
    \breve\P \coloneqq \tilde\P_{(\breve X,\breve \xi)}\in \c P_2(D^0),
    \]
    we complete our construction of the probability measure $\breve\P$ on $D^0$ since by construction $\breve\P \in \c P(t,m)$.
\end{enumerate}
\end{proof}

\subsection{The quasi-variational inequality}\label{section qvi}
In this section we characterise the value function in terms of a quasi-variational inequality (QVI) in the Wasserstein space $\c P_2 = \c P_2(\R^d\times\R^l)$. The main idea behind a QVI is that we can split the state space into two regions: the \emph{intervention region}, where we have to act and modify the current control state $\xi$, and the \emph{continuation region}, where it is optimal to hold the control fixed. 

The the proof of our main result uses an important result from \cite{cosso2021master} and hence requires the following additional assumptions. 

\begin{standingAssumption}
\begin{enumerate}[label=(B\arabic*)]
    \item The functions $b$ and $\sigma$ are bounded.
    \item The functions $f,g$ and $c$ are bounded locally in $\xi$ uniformly in $t,m,x$.
    \item The functions $f,g$ and $c$ are locally Lipschitz in $\xi$ uniformly in $t,m,x$ and globally Lipschitz in $m,x$ uniformly in $t,\xi$.
    \item The functions $b,\sigma,\gamma,f$ and $c$ are $\alpha$-Hölder continuous in $t$, for some $\alpha\in (0,1]$, locally uniformly in $\xi$ uniformly in $m,x$.
    \item \label{assumption sigma indepdent of m}The function $\sigma$ does not depend on $m$ and $\sigma \in C^{1,2}_b$, by which we mean that $\partial_t \sigma$, $D_{(x,\xi)} \sigma$ and $D^2_{(x,\xi)(x,\xi)} \sigma$ exist and are globally bounded in $t,x$ locally uniformly in $\xi$.
\end{enumerate}
\end{standingAssumption}

A key assumptions in \cite{cosso2021master} is \ref{assumption sigma indepdent of m}, which is a strong restriction as it reduces our state dynamics to
\[
dX_t = b(t,m_t,X_t,\xi_t) dt + \sigma(t,X_t,\xi_t) dW_t + \gamma(t) d\xi_t.
\]

\subsubsection{The continuation region}
On the continuation region, the control process {is} $\xi$ constant. Hence, the state dynamics is  essentially uncontrolled and continuous. As a result, we can make use of the Itô formula for continuous measure flows obtained in \cite{cosso2021master} to derive an HJB-type equation on the Wasserstein space for the value function on the continuation region.

\paragraph{Differential calculus on the Wasserstein space}
For the reader's convenience, we now recall the main tools of differential calculus in the Wasserstein space used in this paper. We start with the definition of the so called \emph{linear derivative} or \emph{flat derivative}, see also \cite{guo2020ito}. We say a function ${U}:\c P_2\to \R$ admits a linear derivative if there exists a function $\delta_m {U}:\c P_2\times\R^d\times\R^l\to \R$ such that
\[
{U}(m') - {U}(m) = \int_0^1 \int_{\R^d\times\R^l} \delta_m {U}(\lambda m' + (1-\lambda)m, x,\xi)(m'-m)(dx,d\xi) d\lambda.
\]
If it exists, the linear derivative $\delta_m {U}$ is uniquely defined up to an additive constant.  For functions of the form
\begin{equation}\label{eq example differentiable function}
    {U}(m) = \int_{\R^d\times \R^l} \psi(x,\xi) m(dx,d\xi),
\end{equation}
the linear derivative is given, up to an additive constant, by
\[
\delta_m {U}(m,x,\xi) = \psi(x,\xi).
\]

Another way to define differential calculus in the Wasserstein space is given by the so called \emph{$L$-derivatives}; for details we refer to \cite{carmona2018probabilistic}. For functions of the form \eqref{eq example differentiable function} the $L$-derivative is given by
\[
\partial_\mu {U}(t,m)(x,\xi) = \partial_{(x,\xi)} \psi(x,\xi) = \partial_{(x,\xi)} \delta_m {U}(m,x,\xi).
\]
A similar relation between the two derivatives can be established under quite general assumptions; see \cite{carmona2018probabilistic} for details.

We will work with the space $C_2^{1,2}$ { defined as follows.}

{
\begin{definition}
    We define $C_2^{1,2} \coloneqq C_2^{1,2}([0,T]\times\c P_2)$ as the space of all functions ${U}:[0,T]\times \c P_2\to\R$ such that
    \begin{enumerate}[label=(\roman*)]
        \item both the time derivative $\partial_t {U}(t,m)$ and the linear derivative $\delta_m {U}(t,m)$ exist and are jointly continuous in $(t,m)$,
        \item the derivatives $\partial_{(x,\xi)} \delta_m {U}(t,m,x,\xi)$ and $\partial_{(x,\xi)(x,\xi)}^2 \delta_m {U}(t,m,x,\xi)$ as well as the $L$-derivatives $\delta_\mu U(t,m,x,\xi)$ and $\delta_{(x,\xi)} \delta_\mu U(t,m,x,\xi)$ exist, are jointly continuous in $(t,m)$ and have quadratic growth in $(x,\xi)$ uniformly in $(t,m)$. That is, there exists a constant $C > 0$ such that for all $(t,m,x,\xi)$,
        \begin{equation}
            \begin{split}
        |\partial_{(x,\xi)} \delta_m {U}(t,m,x,\xi)| + |\partial_{(x,\xi)(x,\xi)}^2 \delta_m {U}(t,m,x,\xi)| \\
        + |\delta_\mu U(t,m,x,\xi)| + |\delta_{(x,\xi)} \delta_\mu U(t,m,x,\xi)|
        &\leq C (1 + |x|^2 + |\xi|^2).
            \end{split}
        \end{equation}
    \end{enumerate}
\end{definition}
}

{We note that} under the above assumptions the different notions of differentiability agree, that is, $\partial_\mu {U} = \partial_{(x,\xi)} \delta_m {U}$.

Under the growth assumptions on $C_2^{1,2}$, in \cite{cosso2021master} the authors proved an Itô formula for continuous measure flows. Their definition of $C_2^{1,2}$ is potentially less restrictive as they do not assume a priori that for functions in $C_2^{1,2}$ their linear derivatives exist, but instead only work with the $L$-derivatives $\partial_\mu {U}$ and $\partial_{(x,\xi)} \partial_\mu {U}$. In our work we prefer to work with the linear derivative as it will come up naturally in the intervention region as we will see in Subsection \ref{subsection intervention region}. The following corollary paraphrases \cite[Theorem 3.3]{cosso2021master} in terms of the linear derivative, adapted to our setting.

\begin{corollary}[{\cite[Theorem 3.3]{cosso2021master}}]\label{corollary ito formula}
    Let ${U}\in C_2^{1,2}$, $0\leq t\leq s\leq T$, $m\in \c P_2$ and $\P\in \c P(t,m)$. Suppose that $\xi_r = \xi_{t-}$ for all $r\in [t,s]$. Then
    \[
        {U}(s,m_s) = {U}(t,m_t) + \int_t^s \b L {U}(r,m_r) dr,
    \]
    where $\b L$ is defined as follows
    \begin{equation}
    \begin{split}
        \b L {U}(t,m) &\coloneqq \partial_t {U}(t,m) + \int_{\R^d\times\R^l} \Big(b(t,m,x,\xi) \cdot \partial_x \delta_m {U}(t,m,x,\xi)\\
        &\qquad + \frac 1 2 (\sigma\sigma^T) (t,x,\xi) : \partial_{xx}^2 \delta_m {U}(t,m,x,\xi)\Big) m(dx,d\xi),
    \end{split}
    \end{equation}
    {where $\cdot$ denotes the scalar product on $\R^d$, $:$ the Frobenius inner product on $\R^{d\times d}$ and $\sigma^T\in\R^{m\times d}$ the transpose of $\sigma\in\R^{d\times m}$.}
\end{corollary}

It is important to note that due to the rather weak quadratic growth assumptions of $C^{1,2}_2$ for the {first} and second partial derivatives, we can only use the Itô formula without jumps established in \cite{cosso2021master}. In our case this is sufficient since the process $(X,\xi)$ is continuous on the continuation region. Thus by standard theory, the DPP from Section \ref{section dpp} together with the Itô formula in Corollary \ref{corollary ito formula} suggest that $V$, given it is sufficiently smooth, should satisfy
\[
- \b L V(t,m) - \int_{\R^d\times\R^l} f(t,m,x,\xi) m(dx,d\xi) = 0\qquad\text{on }[0,T)\times\c P_2,\qquad V(T,\cdot) = g,
\]
on the continuation region.

\subsubsection{The intervention region}\label{subsection intervention region}
For the intervention region it turns out that a Markovian viewpoint (in the sense of law invariance) of the previously established jump costs (see Section \ref{subsection reward functional jump representation}) will be useful. We recall that the current state is described by $(t,m) \in [0,T]\times \c P_2$ and start with the following definition. It establishes an order on the Wasserstein space and tells which states $m' = m_t$ can be reached from the current distribution $m = m_{t-}$. The notation is borrowed from \cite{talbi2021dynamic}.

\begin{definition}
    For every fixed $t\in [0,T]$, we define the partial order $\preceq_t$ on $\c P_2$ as follows. We say that $m'\preceq_t m$, that is $m'$ is reachable at time $t$ from $m$, if and only if there exists some $\P\in \c P(t,m)$ (thus $m_{t-} = m$) such that $m_t = m'$.
\end{definition}

Following this straightforward definition, we now prove a more transparent characterisation of the set of reachable states. The idea here is similar to the one employed in Section \ref{subsection reward functional jump representation}, where we approximated jumps by interpolating paths.

\begin{definition}\label{definition interpolating paths measure flow}
    Let $t\in [0,T]$ and $m,m'\in \c P_2$. We define $\Gamma(t,m,m')$ as the set of all probability laws $\mu\in \c P_2(D^0([0,1];\R^d\times\R^l))$, which we equip with the canonical process $(Y,\zeta)$, such that
    \begin{enumerate}[label=(\roman*)]
        \item $\mu_{(Y_0,\zeta_0)} = m$ and $\mu_{(Y_1,\zeta_1)} = m'$,
        \item $\lambda\mapsto \zeta_\lambda$ is non-decreasing and càdlàg $\mu$-a.s.,
        \item $[0,1]\ni \lambda\mapsto \mu_{(Y_\lambda,\zeta_\lambda)} \in \c P_2$ is continuous,
        \item $Y_\lambda - Y_0 = \gamma(t)(\zeta_\lambda-\zeta_0)$ for all $\lambda\in [0,1]$, $\mu$-a.s.
    \end{enumerate}
    We call $\Gamma(t,m,m')$ the set of interpolating paths between $m$ and $m'$ at time $t$.
\end{definition}

\begin{lemma}\label{lemma m preceq m}
    Let $(t,m) \in [0,T]\times \c P_2$. Then, $m' \preceq_t m$ if and only if $\Gamma(t,m,m')$ is nonempty, that is, if there exists a interpolating path between $m$ and $m'$.
\end{lemma}

\begin{proof}
    If $m'\preceq_t m$, then there exists some $\P \in \c P(t,m)$ such that $m_{t-} = m$ and $m_t = m'$. Now we define
    \[
    \zeta_\lambda \coloneqq (1-\lambda) \xi_{t-} + \lambda \xi_t,\qquad Y_\lambda \coloneqq X_{t-} + \gamma(t) (\zeta_\lambda - \xi_{t-}),
    \]
    and we directly verify that $\mu_\lambda \coloneqq \P_{(Y_\lambda,\zeta_\lambda)}$ indeed has the desired property.
    
    Conversely, given such a path $\mu\in \Gamma(t,m,m')$ from $m$ to $m'$ with the canonical process $(Y,\zeta)$, we can construct a fitting control process $\xi$ via
    \[
    \xi_s \coloneqq \zeta_0 \1_{[0-,t)}(s) + \zeta_1 \1_{[t,T]}(s).
    \]
    To define the corresponding state process, we need an independent Brownian motion $W$, for which we extend the probability space if needed. Then we can define the state process as $X_s \coloneqq Y_0$ on $[0-,t)$ and on $[t,T]$ as the solution to
    \[
    dX_{{s}} = b({{s}},m_{{s}},X_{{s}},\xi_{{s}}) d{{s}} + \sigma({{s}},m_{{s}},X_{{s}},\xi_{{s}}) dW_{{s}},\qquad X_t = Y_1.
    \]
    
    Since $\mu\in \Gamma(t,m,m')$, we have that $Y_1 = Y_0 + \gamma(t)(\zeta_1-\zeta_0)$, and thus $X_t = X_{t-} + \gamma(t)(\xi_t-\xi_{t-})$. Hence $\P_{(X,\xi,W)}$ is a weak solution to \eqref{eq diffusion x} on $[t-,T]$ and thus $\P_{(X,\xi)} \in \c P(t,m)$ with $m_{t-} = m$ and $m_t = m'$.
\end{proof}

In terms of the functions $\Xi(t,\bar m)$ introduced in Definition \ref{definition c l2} it follows from Definition \ref{definition interpolating paths measure flow} that 
\[
    \Gamma(t,m,m') = \bigcup_{\bar m \in\Pi(m,m')}\Xi(t,\bar m),
\]    
where $\Pi(m,m')$ denotes the set of all couplings $\bar m \in \c P_2(\R^d\times\R^l\times\R^d\times\R^l)$ with marginal measures $m$ and $m'$. In view of Section \ref{subsection reward functional jump representation}, the natural extension of the minimal jump costs to a distributional level jumps is given as follows.

\begin{definition}\label{definition c m}
    Let $t\in [0,T]$ and $m,m'\in \c P_2$ with $m'\preceq_t m$. We define the minimal jump costs from $m$ to $m'$ at time $t$ as
    \begin{equation}
    \begin{split}
        &C_m(t,m,m')\\
        &\coloneqq \inf_{\bar m \in \Pi(m,m')} C_{\c P_2}(t,\bar m)\\
        &= \inf_{\mu \in \Gamma(t,m,m')} \E^{\mu}\bigg[\int_{J^c_{[0,1]}(\zeta)} c(t,\mu_\lambda,Y_\lambda,\zeta_\lambda) d\zeta_\lambda + \sum_{J^d_{[0,1]}(\zeta)} C_{D^0}(t,\mu_\lambda,Y_{\lambda-},\zeta_{\lambda-},\zeta_\lambda)  \bigg],
    \end{split}
    \end{equation}
    where $\mu_\lambda \coloneqq \mu_{(Y_\lambda,\zeta_\lambda)}$. 
\end{definition}

For the intervention region, we characterised the set of all reachable states via the partial order $\preceq_t$, and the corresponding jump costs by $C_m$. So from a control perspective, if it is favourable to act immediately, there should exist a more favourable state $m' \preceq_t m$ with $V(t,m) = V(t,m') - C_m(t,m,m')$ that one should move to. At the same time, by dynamic programming, $V(t,m) \geq V(t,{m'}) - C_m(t,m,m')$ for every $m'\preceq_t m$. Therefore, we expect that the following holds on the intervention region:
\[
V(t,m) - {\sup_{m'\preceq_t m}} [V(t,m') - C_m(t,m,m')] = 0.
\]
Along with the HJB equation from the continuation region this yields the following QVI on the Wasserstein space, which we would expect the value function to fulfil under sufficient regularity assumptions,
\begin{equation}
\begin{split}
    &\min\bigg\{-\b L {U}(t,m) - \int_{\R^d\times\R^l}f(t,m,x,\xi) m(dx,d\xi),\\
    &\qquad
    {U}(t,m) - {\sup_{m'\preceq_t m}} [{U}(t,{m'}) - C_m(t,m,m')] \bigg\} = 0,\qquad \text{ on }[0,T)\times\c P_2.\label{eq first qvi}
\end{split}
\end{equation}

However, the above QVI is ill-behaved. Every function ${U}$ is a subsolution to this QVI since
\[
{U}(t,m) = {U}(t,m) -  C_m(t,m,{m}) \leq {\sup_{m'\preceq_t m}} [{U}(t,m') - C_m(t,m,m')].
\]
This is a common for singular control problems. The main approach is to replace this intervention part of the QVI by a better behaved expression involving the derivative of ${U}$. The tools needed for this are provided by the following key lemma. 

The lemma characterises the optimality of a function with respect to jumps of the measure flow ${s}\mapsto m_{{s}}$ in terms of the linear derivative. A similar result has already been obtained by \cite{talbi2021dynamic} and used in \cite{talbi2022viscosity} for a mean-field optimal stopping problem. Using similar arguments as in \cite{talbi2021dynamic} the following lemma extends the previously obtained results by allowing for jump costs that may depend on the connecting path.

\begin{lemma}\label{lemma optimality}
    Let ${U}\in C_2^{1,2}$ and $t\in [0,T]$. Assume that for all $m\in \c P_2$ we have that
    \begin{equation}\label{eq key lemma C inequality}
        {U}(t,m) + C_m(t,m,m') \geq {U}(t,m')\qquad\text{for all }m'\preceq_t m.
    \end{equation}
    Then it also holds for all $m\in\c P_2$ that
    \begin{equation}\label{eq key lemma derivative inequality}
        \partial_{(x,\xi)} \delta_m {U}(t,m,x,\xi) \cdot (\gamma(t),1) \leq c(t,m,x,\xi)\qquad\text{for all }(x,\xi)\in\R^d\times\R^l.
    \end{equation}
\end{lemma}

\begin{proof}
Let us denote for every $m\in \c P_2$ the set of points, for which the claim \eqref{eq key lemma derivative inequality} does not hold, by
\[
\c N_m = \bigl\{(x,\xi)\in \R^d\times\R^l \bigm\vert \partial_{(x,\xi)} \delta_m {U}(t,m,x,\xi) \cdot (\gamma(t),1) > c(t,m,x,\xi) \bigr\}.
\]
Then it is sufficient to prove that $\c N_m$ is empty for all $m\in\c P_2$. We proceed in two steps.

\begin{enumerate}[wide]
    \item[\emph{Step 1: $m(\c N_m) = 0$ for all $m\in\c P_2$.}]
    Let us assume to the contrary that there exists an $m\in\c P_2$ with $m(\c N_m) > 0$. We now consider for each $\varepsilon > 0$ the state-control configuration
    \[
    Y_\varepsilon = X_{t-} + \1_{\c N_m}(X_{t-},\xi_{t-}) \gamma(t) \varepsilon,\qquad \zeta_\varepsilon = \xi_{t-} + \1_{\c N_m}(X_{t-},\xi_{t-}) \varepsilon,
    \]
    which corresponds to moving just the particles in $\c N_m$ by $\varepsilon$. We denote the law of $(Y_\varepsilon,\zeta_\varepsilon)$ by $m^{\c N_m}_\varepsilon$.
    By construction, 
    \[
    [0,1] \ni \lambda \mapsto m^{\c N_m}_{\lambda \varepsilon} \in \c P_2
    \]
    is a continuous path from $m$ to $m^{\c N_m}_\varepsilon$ and in $\Gamma(t,m,m^{\c N_m}_\varepsilon)$.
    In particular, the costs of moving along this path from $m$ to $m^{\c N_m}_\varepsilon$ are given by
    \[
    \E\bigg[\int_0^1 c(t,m^{\c N_m}_{\lambda\varepsilon},Y_{\lambda\varepsilon},\zeta_{\lambda\varepsilon}) d\zeta_{\lambda\varepsilon}\bigg] = \int_0^\varepsilon \int_{\c N_m} c(t,m_\lambda^{\c N_m},x + \gamma(t)\lambda, \xi +\lambda) m(dx,d\xi) d\lambda.
    \]
    Now using that by construction and Lemma \ref{lemma m preceq m}, $m^{\c N_m}_\varepsilon \preceq_t m$, we obtain that
    \[
    {U}(t,m^{\c N_m}_\varepsilon) - {U}(t,m) \leq C_m(t,m,m^{\c N_m}_\varepsilon) \leq  \int_0^\varepsilon \int_{\c N_m} c(t,m_\lambda^{\c N_m},x + \gamma(t)\lambda, \xi +\lambda) m(dx,d\xi) d\lambda.
    \]
    At the same time, we can use the definition of the linear derivative to rewrite the left-hand side as
    \begin{equation}
    \begin{split}
        &{U}(t,m^{\c N_m}_\varepsilon) - {U}(t,m) \\
        &= \int_0^1 \int_{\R^d\times \R^l} \delta_m {U}(t,m^{\c N_m}_{\lambda\varepsilon},x,\xi) (m^{\c N_m}_{\lambda\varepsilon} - m)(dx,d\xi) d\lambda\\
        &= \int_0^\varepsilon \int_{\c N_m} \frac{\delta_m {U}(t,m^{\c N_m}_\lambda, x+\gamma(t)\lambda,\xi+\lambda) - \delta_m {U} (t,m^{\c N_m}_\lambda,x,\xi)}{\varepsilon} m(dx,d\xi) d\lambda\\
        &= \int_0^\varepsilon \int_{\c N_m} \partial_{(x,\xi)} \delta_m {U}(t,m^{\c N_m}_\lambda,x+\gamma(t)\tilde\lambda,\xi+\tilde\lambda)\cdot (\gamma(t),1) m(dx,d\xi) d\lambda,
    \end{split}
    \end{equation}
    for some $\tilde \lambda \in (0,\lambda)$ depending on $x,\xi,\lambda$. Together with the previous estimate, we obtain that
    \begin{equation}
    \begin{split}
        \frac 1 \varepsilon \int_0^\varepsilon \int_{\c N_m} \Big( \partial_{(x,\xi)} \delta_m {U}(t,m^{\c N_m}_\lambda,x+\gamma(t)\tilde\lambda,\xi+\tilde\lambda)\cdot (\gamma(t),1)&\\
        - c(t,m_\lambda^{\c N_m},x + \gamma(t)\lambda, \xi +\lambda) \Big) m(dx,d\xi) d\lambda& \leq 0.
    \end{split}
    \end{equation}
    Since $\partial_{(x,\xi)} \delta_m {U}$ and $c$ have at most quadratic growth in $x,\xi$ locally uniform in $m$, and due to the compactness of $\{m^{\c N_m}_\varepsilon \mid \varepsilon \in [0,1] \}$ and the continuity of the integrand in $\lambda = 0$, we can use Lebesgue's differentiation theorem to obtain
    \[
    \int_{\c N_m} \Big(\partial_{(x,\xi)} \delta_m {U}(t,m,x,\xi) \cdot (\gamma(t),1) - c(t,m,x,\xi) \Big) m(dx,d\xi) \leq 0,
    \]
    and thus $m(\c N_m) = 0$.
    \item[\emph{Step 2: $\c N_m$ is empty for all $m\in\c P_2$.}]
    Let us assume to the contrary that $\c N_m$ is not empty, so there exists an $(x,\xi)\in \c N_m$. We now look at the following approximating sequence for $m$,
    \[
    m_\lambda \coloneqq \lambda \delta_{(x,\xi)} + (1-\lambda)m \in \c P_2,\qquad\lambda\in [0,1],
    \]
    and we recall from the previous step that $m_\lambda(\c N_{m_\lambda}) = 0$ and thus $\delta_{(x,\xi)}(\c N_{m_\lambda}) = 0$ for all $\lambda\in [0,1]$. Since $c$ and $\delta_m {U}$ are continuous, this implies now $\delta_{(x,\xi)}(\c N_m) = 0$ and hence $(x,\xi)\notin\c N_m$.
\end{enumerate}
\end{proof}

\begin{remark}
    If we restrict ourselves to functions ${U}\in C^{1,2}_2$ that satisfy the slightly stronger conditions of $\partial_{(x,\xi)} \delta_m {U}$ being of linear growth in $x,\xi$ locally uniformly in $(t,m)$ and $\partial_{(x,\xi)(x,\xi)}^2 \delta_m {U}$ being uniformly bounded in $(x,\xi)$ locally uniformly in $(t,m)$, then a simple application of the Itô formula proven in \cite{guo2020ito,talbi2021dynamic} allows for a reverse direction of the preceding lemma: given $t\in [0,T]$, if \eqref{eq key lemma derivative inequality} holds for all $m\in \c P_2$, then \eqref{eq key lemma C inequality} also holds for all $m\in\c P_2$.
\end{remark}

\subsubsection{Characterisation of the value function}

Using the characterisation of the intervention region in terms of the linear derivative leads us to the following QVI
\begin{equation}
\label{eq singular qvi}
\begin{split}
    &\min\bigg\{-\b L {U}(t,m) - \int_{\R^d\times\R^l} f(t,m,x,\xi) m(dx,d\xi),\\
    &\qquad\inf_{(x,\xi)\in\R^d\times \R^l} \Big[c(t,m,x,\xi) - \partial_{(x,\xi)} \delta_m {U}(t,m,x,\xi)\cdot (\gamma(t),1)\Big] \bigg\} = 0\qquad\text{on }[0,T)\times\c P_2,\\
    &\min \bigg\{{U}(T,m) - \int_{\R^d\times\R^l} g(m,x,\xi) m(dx,d\xi),\\
    &\qquad\inf_{(x,\xi)\in \R^d\times\R^l} \Big[c(T,m,x,\xi) - \partial_{(x,\xi)} \delta_m {U}(T,m,x,\xi) \cdot (\gamma(t),1)\Big] \bigg\} = 0\qquad\text{on }\{T\}\times\c P_2.
\end{split}
\end{equation}
In the rest of this section, we will show that this QVI indeed characterises the value function. In general we do not expect the value function to have full $C^{1,2}_2$-regularity. Hence we will now introduce a weaker notion of solutions to the above QVI.

\begin{definition}
    A lower semi-continuous function ${U}:[0,T]\times\c P_2\to\R$ is called a \emph{viscosity supersolution} to \eqref{eq singular qvi} if for every $(\bar t,\bar m)\in [0,T]\times\c P_2$ and every test function $\varphi\in C_2^{1,2}$ with $\varphi \leq {U}$ and ${U}(\bar t,\bar m) = \varphi(\bar t,\bar m)$, we have if $\bar t < T$,
    \begin{equation}
    \begin{split}
        &\min \bigg\{-\b L\varphi(\bar t,\bar m) - \int_{\R^d\times\R^l} f(\bar t,\bar m,x,\xi) \bar m(dx,d\xi),\\
        &\qquad\inf_{(x,\xi)\in\R^d\times\R^l} \Big[c(\bar t,\bar m,x,\xi) - \partial_{(x,\xi)} \delta_m \varphi(\bar t,\bar m,x,\xi)\cdot (\gamma(\bar t),1)\Big] \bigg\} \geq 0,
    \end{split}
    \end{equation}
    and if $\bar t = T$,
    \begin{equation}
    \begin{split}
        &\min \bigg\{{U}(T,\bar m) - \int_{\R^d\times\R^l} g(\bar m,x,\xi) \bar m(dx,d\xi),\\
        &\qquad\inf_{(x,\xi)\in\R^d\times\R^l} \Big[c(T,\bar m,x,\xi) - \partial_{(x,\xi)} \delta_m \varphi(T,\bar m,x,\xi)\cdot (\gamma(T),1)\Big] \bigg\} \geq 0.
    \end{split}
    \end{equation}
\end{definition}

We are now ready to state the main result of this section. The proof proceeds in several steps. 

\begin{theorem}\label{theorem value function qvi characterisation}
    The value function $V$ is a viscosity supersolution to \eqref{eq singular qvi}, and if we further assume that $V$ is bounded, then for every other bounded, continuous viscosity supersolution ${U}$ of \eqref{eq singular qvi}, we have $V\leq {U}$. In particular, if $V$ is bounded and continuous, then it is the minimal bounded, continuous viscosity supersolution to \eqref{eq singular qvi}.
\end{theorem}

To prove the above theorem, it turns our to be very useful to take a closer look again at the set of bounded velocity controls $\c L(t,m)$, which we defined as controls of the form
\begin{equation}\label{eq connection singular and regular control}
\xi_s = \xi_{t-} + \int_t^s u_r dr,\qquad s\in [t,T],
\end{equation}
for some bounded, adapted, non-negative $(u_s)_{s\in [t,T]}$.
By viewing $u$ as our control, we can transform the bounded velocity problem into a regular control problem.
This connection between singular controls and regular controls is commonly used in the singular control literature.

\begin{definition}
    We define by
    \[
    \c L^K(t,m) \coloneqq \bigl\{\P\in\c L(t,m) \bigm\vert \xi\text{ is Lipschitz continuous with constant }K\text{ on }[t-,T],\ \P\text{-a.s.} \bigr\}
    \]
    the subset $K$-bounded velocity controls and denote by 
    \[
    V^K(t,m) \coloneqq \sup_{\P\in \c L^K(t,m)} J(t,m,\P)
    \]
    the supremum of the reward function over the set of $K$-bounded velocity controls.  
\end{definition}

In view of Theorem \ref{theorem representation of J with parametrisations} we see that the restricted value functions $V^K$ converge to the original value function as $K \to \infty$.

\begin{lemma}\label{lemma bounded velocity approximation}
    For all $(t,m) \in [0,T]\times \c P_2$, we have
    \[
    V^K(t,m)\uparrow V(t,m)\qquad\text{for }K\to\infty.
    \]
\end{lemma}

\begin{proof}
{
    By Theorem \ref{theorem representation of J with parametrisations}, we have
    \[
    V(t,m) = \sup_{\P\in\c P(t,m)} J(t,m,\P) = \sup_{\P\in\c L(t,m)} J(t,m,\P).
    \]
    Thus the claim follows by noting that $\c L^K(t,m)\uparrow \c L(t,m)$ as $K\to\infty$.
}
\end{proof}

For bounded velocity controls $\P \in \c L(t,m)\subseteq \c C(t,m)$, using the representation \eqref{eq connection singular and regular control}, the reward functional is given by
\begin{equation}
\begin{split}
    J(t,m,\P) &= \E^\P\bigg[\int_t^T f(s,m_s,X_s,\xi_s) ds + g(m_T,X_T,\xi_T) - \int_t^T c(s,m_s,X_s,\xi_s) d\xi_s \bigg]\\
    &= \E^\P\bigg[\int_t^T f(s,m_s,X_s,\xi_s) ds + g(m_T,X_T,\xi_T) - \int_t^T c(s,m_s,X_s,\xi_s) \cdot u_s ds \bigg].
\end{split}
\end{equation}
This allows us to consider the process $u$ as our control and hence consider the state dynamics 
\begin{equation}
\begin{split}
    dX_t &= \big(b(t,m_t,X_t,\xi_t) + \gamma(t) u_t\big) dt + \sigma(t,X_t,\xi_t) dW_t,\\
    d\xi_t &= u_t dt.
\end{split}
\end{equation}

Restricting the set of admissible controls to the set $\c L^K(t,m)$ we obtain a regular control problem with value function $V^K$.  
This allows us to apply a comparison principle and uniqueness result for the master equation that has recently been obtained by \cite{cosso2021master}. While we consider the weak formulation of the control problem, \cite{cosso2021master} work with the strong formulation where the probability space is fixed and the filtration is generated by the driving Brownian motion and an additional uniformly distributed random variable. However, for the value function itself this makes no difference; by \cite[Theorem 1]{djete_mckeanvlasov_2022} the value functions for the weak and the strong formulation coincide. Also definition of $C_2^{1,2}$ {in \cite{cosso2021master}} is slightly different and contains possibly more functions than ours. However, all of the arguments {in \cite[Theorems 3.8 and 5.1]{cosso2021master}} still hold if we replace their definition with ours.

The master equation for the regular control problem reads as follows
\begin{equation}
\begin{aligned}
    &-\partial_t v(t,m) - \int_{\R^d\times \R^l} \sup_{u\geq 0, \norm{u}\leq K} \bigg[ f(t,m,x,\xi) - c(t,m,x,\xi) \cdot u\\*
    &\qquad+ \big(b(t,m,x,\xi) + \gamma(t) u\big) \cdot \partial_x \delta_m v(t,m,x,\xi)
    + u \cdot \partial_\xi \delta_m v(t,m,x,\xi)\hspace{20pt}\\*
    &\qquad+ \frac 1 2 (\sigma\sigma^T)(t,x,\xi) : \partial_{xx}^2 \delta_m v(t,m,x,\xi) \bigg] m(dx,d\xi) = 0 &\mathllap{\text{on }[0,T)\times\c P_2,}\\*
    &v(T,m) = \int_{\R^d\times\R^l} g(m,x,\xi) m(dx,d\xi)&\mathllap{\text{on }\{T\}\times\c P_2.}
\end{aligned}
\end{equation}
Using the operator $\b L$ {(see Corollary \ref{corollary ito formula})}, the equation can be compactly rewritten as 
\begin{equation}
\label{eq singular bounded velocity hjb}
\begin{aligned}[b]
    &- \b L v(t,m) - \int_{\R^d\times\R^l} f(t,m,x,\xi) m(dx,d\xi) + \int_{\R^d\times \R^l} \inf_{u\geq 0,\norm u\leq K} \Big[\big(c(t,m,x,\xi)\hspace{-35pt}\\*
    &\qquad- \partial_{(x,\xi)} \delta_m v(t,m,x,\xi) \cdot (\gamma(t),1)\big) \cdot u \Big] m(dx,d\xi) = 0&\text{on }[0,T)\times\c P_2,\\*
    &v(T,m) = \int_{\R^d\times\R^l} g(m,x,\xi) m(dx,d\xi)&\mathllap{\text{on }\{T\}\times\c P_2.}
\end{aligned}
\end{equation}
For this equation, \cite{cosso2021master} proved the following result.

\begin{corollary}[{\cite[Theorems 3.8 and 5.1]{cosso2021master}}]\label{corollary bounded velocity}
    The value function $V^K$ is the unique bounded, continuous viscosity solution to \eqref{eq singular bounded velocity hjb}. Furthermore, for every bounded, continuous viscosity supersolution ${U}:[0,T]\times\c P_2\to \R$ to \eqref{eq singular bounded velocity hjb}, we have $V^K \leq {U}$ on $[0,T]\times\c P_2$.
\end{corollary}

We extend this characterisation using the approximation result in Lemma \ref{lemma bounded velocity approximation} to the original value function $V$.

\begin{proof}[Proof of Theorem \ref{theorem value function qvi characterisation}]
\begin{enumerate}[wide]
    \item 
    Combining Lemma \ref{lemma bounded velocity approximation} and \ref{corollary bounded velocity}, we deduce that $V$ is lower semi-continuous.
    To prove that $V$ is a viscosity supersolution to \eqref{eq singular qvi}, we first consider the case $(\bar t,\bar m)\in [0,T)\times \c P_2$. Let $\varphi\in C_2^{1,2}$ with $V\geq \varphi$ and $
    V(\bar t,\bar m) = \varphi(\bar t,\bar m)$. Further let $m\preceq_{\bar t} \bar m$, which means that there exists a $\P\in \c P(t,m)$ with $m_{\bar t-} = \bar m$ and $m_{\bar t} = m$. Using the DPP established in Theorem \ref{theorem dpp}, we obtain that
    \[
    \varphi(\bar t,\bar m) = V(\bar t,\bar m) \geq V(\bar t,m) - C_m(\bar t,\bar m,m) \geq \varphi(\bar t, m) - C_m(\bar t,\bar m,m).
    \]
    Since this holds for all $m\preceq_{\bar t}\bar m$, we can use Lemma \ref{lemma optimality} to obtain
    \[
    \partial_{(x,\xi)} \delta_m \varphi(\bar t,\bar m,x,\xi)\cdot (\gamma(t),1) \leq c(\bar t,\bar m,x,\xi)\qquad\text{for all }(x,\xi) \in \R^d\times \R^l.
    \]
    
    For the other part of the QVI, we choose $\P\in \c P(\bar t,\bar m)$ such that $\xi_s \equiv \xi_{\bar t-}$ on $s\in [\bar t,T]$. Due to the absence of jumps in $(X,\xi)$ on $[\bar t,T]$, we can now use It\^o's formula, see Corollary \ref{corollary ito formula} and get, for $\delta > 0$,
    \[
    \varphi(\bar t+\delta,m_{\bar t+\delta}) = \varphi(\bar t,\bar m) + \int_{\bar t}^{\bar t+\delta} \b L\varphi(s,m_s)ds.
    \]
    At the same time, we obtain from the DPP in Theorem \ref{theorem dpp}
    \[
    V(\bar t,\bar m) \geq V(\bar t+\delta,m_{\bar t+\delta}) + \int_{\bar t}^{\bar t+\delta} \int_{\R^d\times\R^l} f(s,m_s,x,\xi) m_s(dx,d\xi) ds.
    \]
    Together, this implies
    \begin{equation}
    \begin{split}
    0 &\leq V(\bar t+\delta,m_{\bar t+\delta}) - \varphi(\bar t+\delta,m_{\bar t+\delta})\\
    &\leq \int_{\bar t}^{\bar t+\delta} \bigg(- \b L \varphi(s,m_s) - \int_{\R^d\times\R^l} f(s,m_s,x,\xi) m_s(dx,d\xi)\bigg)  ds.
    \end{split}
    \end{equation}

    {We let $\delta\to 0$ and apply the mean value theorem, using that $s\mapsto m_s$ is continuous by construction, $f$ is continuous by Assumption \ref{assumption f,c continuous} and $\varphi\in C^{1,2}_2$. We obtain}
    \[
    -\b L \varphi(\bar t,\bar m) - \int_{\R^d\times\R^l} f(\bar t,\bar m,x,\xi) \bar m(dx,d\xi) \geq 0.
    \]
    The case $(T,\bar m)\in \{T\}\times \c P_2$ is very similar and is hence omitted.
    \item
    The second part of the proof is the minimality of $V$. So now let ${U}$ be a bounded, continuous viscosity supersolution to \eqref{eq singular qvi}. Then for every $K$, the function ${U}$ is also a bounded, continuous viscosity supersolution to \eqref{eq singular bounded velocity hjb}, since we can show
    \begin{equation}
    \begin{split}
    \int_{\R^d\times\R^l} \inf_{u\geq 0, \norm{u} \leq K} \Big[ \big(c(t,m,x,\xi) - \partial_{(x,\xi)} \delta_m v(t,m,x,\xi) \cdot (\gamma(t),1)\big) \cdot u \Big]m(dx,d\xi) \geq 0,\\
    \text{on }[0,T]\times\c P_2,
    \end{split}
    \end{equation}
    in the viscosity sense, using the supersolution property. Note that actually even equality holds since we can choose $u = 0$.
    
    Thus by the comparison principle in Corollary \ref{corollary bounded velocity} originally obtained by \cite{cosso2021master}, we deduce that
    \[
    V^K \leq {U}\qquad\text{on }[0,T]\times\c P_2.
    \]
    Now taking the limit $K\to\infty$ together with Lemma \ref{lemma bounded velocity approximation} completes the proof.
\end{enumerate}
\end{proof}


\let\c\predefinedC
\printbibliography


\end{document}